\documentclass[sigconf]{acmart}
\AtBeginDocument{%
  \providecommand\BibTeX{{%
    \normalfont B\kern-0.5em{\scshape i\kern-0.25em b}\kern-0.8em\TeX}}}

\copyrightyear{2021}
\acmYear{2021}
\setcopyright{acmcopyright}
\acmConference[SIGMOD '21] {Proceedings of the 2021 International Conference on Management of Data}{June 20--25, 2021}{Virtual Event, China}
\acmBooktitle{Proceedings of the 2021 International Conference on Management of Data (SIGMOD '21), June 20--25, 2021, Virtual Event, China}
\acmPrice{15.00}
\acmISBN{978-1-4503-8343-1/21/06}
\acmDOI{10.1145/3448016.3457306}
\usepackage{xspace}
\newcommand{\eat}[1]{}
\newcommand{\class}[1]{{\ensuremath{\mathsf{#1}}}}

\newcommand{\ignore}[1]{}
\newcommand{\nt}[1]{{\color{blue}\footnotesize[cw: #1]}}
\newcommand{\kartik}[1]{{\color{brown} \footnotesize[Kartik: #1] }}

\newcommand{\re}[1]{{\color{black} #1} }

\newcommand{\boldparagraph}[1]{\vspace{7pt}\noindent\textbf{#1}}

\newcommand{\system}{SOGDB\xspace}
\newcommand{\appsystem}{DP-Sync\xspace}

\newcommand{\cs}{server\xspace}

\newcommand{\user}{owner\xspace}

\newcommand{\as}{analyst\xspace}

\newcommand{\sync}{\ensuremath{\class{Sync}}\xspace}

\newcommand{\setup}{\prod_{\ensuremath{\class{Setup}}}\xspace}
\newcommand{\update}{\prod_{\ensuremath{\class{Update}}}\xspace}
\newcommand{\query}{\prod_{\ensuremath{\class{Query}}}\xspace}
\newcommand{\lsetup}{\mathcal{L}_{\ensuremath{\class{S}}}\xspace}
\newcommand{\lupdate}{\mathcal{L}_{\ensuremath{\class{U}}}\xspace}
\newcommand{\lquery}{\mathcal{L}_{\ensuremath{\class{Q}}}\xspace}

\newcommand{\dec}{\ensuremath{\class{Dec}}\xspace}

\newcommand{\len}{\ensuremath{\class{len}}\xspace}
\newcommand{\writec}{\ensuremath{\class{write}}\xspace}
\newcommand{\readc}{\ensuremath{\class{read}}\xspace}

\newcommand{\ptb}{\ensuremath{\class{Perturb}}\xspace}
\newcommand{\lap}{\ensuremath{\class{Lap}}\xspace}

\newcommand{\encs}{\ensuremath{\class{EncSize}}\xspace}
\newcommand{\upatt}{\ensuremath{\class{UpdtPatt}}\xspace}

\usepackage{amsmath,amsfonts}
\usepackage{algorithm}
\usepackage{algorithmicx}
\usepackage[noend]{algpseudocode}
\usepackage{enumitem}
\usepackage{caption}
\usepackage{subcaption}
\usepackage{mathtools}

\DeclarePairedDelimiter\floor{\lfloor}{\rfloor}

\newtheorem{theorem}{Theorem}
\newtheorem{exmp}{Example}[section]

\begin{document}
\fancyhead{}
%%
%% The "title" command has an optional parameter,
%% allowing the author to define a "short title" to be used in page headers.
\title{\appsystem: Hiding Update Patterns in Secure Outsourced Databases with Differential Privacy}

%%
%% The "author" command and its associated commands are used to define
%% the authors and their affiliations.
%% Of note is the shared affiliation of the first two authors, and the
%% "authornote" and "authornotemark" commands
%% used to denote shared contribution to the research.
\author{Chenghong Wang}
\affiliation{%
  \institution{Duke University}
}
\email{chwang@cs.duke.edu}

\author{Johes Bater}
\affiliation{%
  \institution{Duke University}
  }
\email{johes.bater@duke.edu}

\author{Kartik Nayak}
\affiliation{%
  \institution{Duke University}
}
\email{kartik@cs.duke.edu}

\author{Ashwin Machanavajjhala}
\affiliation{%
 \institution{Duke University}
 }
 \email{ashwin@cs.duke.edu}

%%
%% By default, the full list of authors will be used in the page
%% headers. Often, this list is too long, and will overlap
%% other information printed in the page headers. This command allows
%% the author to define a more concise list
%% of authors' names for this purpose.
\renewcommand{\shortauthors}{Trovato and Tobin, et al.}

%%
%% The abstract is a short summary of the work to be presented in the
%% article.
\begin{abstract}
%\nt{Just curious, does \appsystem look like a good name? Or shall we have another name for it? Also do we need to change the title?}
%\am{Alternate title: Hiding Update Patterns in Secure Outsourced Databases with Differential Privacy. Alternate System name: DP-Sync, since the key innovation is around DP synchronization algorithms. }
%Secure outsourced growing database stands for the group of secure outsourced databases for dynamically growing data.  \am{Replace previous line with: In this paper, we consider privacy preserving update strategies for secure outsourced growing database.} 
In this paper, we consider privacy-preserving update strategies for secure outsourced growing databases. Such databases allow append-only data updates on the outsourced data structure while analysis is ongoing. Despite a plethora of solutions to securely outsource database computation, existing techniques do not consider the information that can be leaked via update patterns.%\am{Replace previous 2 sentences with: Despite a plethora of solutions to securely outsource database computation, existing techniques do not consider the information that can be leaked via updates.} 
~To address this problem, we design a novel secure outsourced database framework for growing data, \appsystem, which interoperate with a large class of existing encrypted databases and supports efficient updates while providing differentially-private guarantees for any single update. %\am{Replace ".. which guarantees ..." with: which allows efficient updates with a differential privacy bound on leakage about any one update.}
We demonstrate \appsystem's practical feasibility in terms of performance and accuracy with extensive empirical evaluations on real world datasets. \vspace{-1mm}
\end{abstract}

%%
%% The code below is generated by the tool at http://dl.acm.org/ccs.cfm.
%% Please copy and paste the code instead of the example below.
%%

%%
%% Keywords. The author(s) should pick words that accurately describe
%% the work being presented. Separate the keywords with commas.
\vspace{-1mm}

\setlength{\textfloatsep}{1pt plus 0.0pt minus 0.0pt}
\setlength{\floatsep}{1pt plus 0.0pt minus 0.0pt}
\setlength{\intextsep}{1pt plus 0.0pt minus 0.0pt}

%%
%% This command processes the author and affiliation and title
%% information and builds the first part of the formatted document.
%\input{revision_letter}
\maketitle
\section{Introduction}
\label{sec:intro}
\ignore{
\draft{A draft story for intro

\begin{itemize}
    \item Point out the problem that though many works considered security for dynamic databases with secure updates. Most of them ignore the potential risk raised by update history (or insertion pattern).
    
    \item Then we give examples to show one can infer critical information (privacy breach) with only access to the update history and without even posting a single query.
    
    \item After describing the example, we might say that this privacy breach of use update history is not limited to the given example, there are many real-world use cases that could be affected. 
    
    \item Then we may list the options to hide the update pattern. Obviously, one way to hide the update pattern is to modify existed encrypted database scheme, and force the update protocol leaks nothing. (This can be difficult, there are only 2 known backward privacy SSE that can somehow limit the leakage of update history but still reveal insertion patterns.) Or from other side, define sync strategies to restrict the \user's update behavior. The second option will be more general and can direct compatible with many existed solutions. But design of a ``good'' sync algorithm is a 3-way tradeoff. So we need to consider methods that can be tuned among privacy, accuracy and perf.
    
    \item Inspired by the idea, we introduce our framework \appsystem that ``helps'' existed edb to handle update pattern leakage. We may briefly describe general design or approach of \appsystem.
    
    \item Summarize the contributions.

    Few of them considered 
\end{itemize}}

\draft{\begin{itemize}
    \item define updates with sync. show two sync algorithms that satisfies this defn and provide a trade-off
    \item Show that this can be merged with many existing secure databases, but not all
    \item evaluate for taxi dataset and show it is performant, provides privacy, and does not kill accuracy. Maybe present a couple of numbers.
\end{itemize}}
}
%\am{I wonder if we should call out and define the "Synchronization" or "Update Synchronization" problem? We talk about synchronization in quite a few places, but we do not formally or informally define it as the problem we are solving. }
%{\bf Claim}: The pattern of Data Synchronization between local data and the outsourced database can leak private information of certain individuals.
In the last couple of decades, organizations have been rapidly moving towards outsourcing their data to the cloud. While this brings in inherent advantages such as lower costs, high availability, and ease of maintenance, this also results in privacy concerns for organizations. Hence, many solutions leverage cryptography and/or techniques such as differential privacy to keep the data private while simultaneously allowing secure query processing on this data~\cite{hacigumucs2002executing, hore2012secure, popa2012cryptdb, arasu2013oblivious, bellare2007deterministic, chowdhury2019cryptc,eskandarian2017oblidb, stefanov2014practical, bater2018shrinkwrap, poddar2016arx}. However, while most practical systems require us to maintain dynamic databases that support updates, research in the space of private database systems has focused primarily on static databases~\cite{hacigumucs2002executing, hore2012secure, arasu2013oblivious, kellaris2017accessing, chowdhury2019cryptc, bater2018shrinkwrap}. There have been a few works which consider private database updates and answering queries on such dynamic databases~\cite{hahn2014searchable, kamara2018sql, popa2012cryptdb, curtmola2011searchable, stefanov2014practical, eskandarian2017oblidb, agarwal2019encrypted}. However, none of these works consider the privacy of \emph{when} a database is updated.  In this work, we consider the problem of hiding such database update patterns.

\eat{\re{For example, ~\cite{agarwal2019encrypted} and ~\cite{lecuyer2019sage} both propose systems that prevent query outputs from violating privacy guarantees, but neither account for information leaked by database updates that occur even before query execution. In this work, we  consider the problem of preventing privacy violations due to these update patterns. }}

Let us consider the following example where an adversary can breach privacy by using the timing information of updates. \re{Consider an IoT provider that deploys smart sensors (i.e., security camera, smart bulb, WiFi access point, etc.) for a building. The provider also creates a database to back up the sensors’ event data. For convenience, the database is maintained by the building administrator, but is encrypted to protect the privacy of people in the building. By default, the sensor will backup immediately when any new sensor event (i.e. a new connection to WiFi access point) occurs. Suppose that at a certain time, say 7:00 AM, only one person entered the building. Afterwards, the building admin observes three backup requests posted at times 7:00:00, 7:00:10, 7:00:20, respectively. Also suppose that the admin has access to additional non-private building information, such as that floor 3 of this building is the only floor which has three sensors with a 10 second walking delay (for an average person). Then, by looking at the specific times of updates (10 second delays) and the number of updates, the building admin can learn private information about the activity (i.e. the person went to the 3rd floor), without ever having to decrypt the stored data. This type of attack generalizes to any event-driven update where the event time is tied to the data upload time. In order to prevent such attacks, we must decouple the relationship between event and upload timings.

%collects information about a user's daily schedule, such as what hours they are typically at home, that the user can view later through an online dashboard. This information is collected locally by the sensors when a user enters or leaves the home then is immediately synchronized with the provider's cloud storage in order to populate the dashboard with up to date information. 

%To protect user privacy, the provider ensures that any data uploaded and stored on the cloud is fully encrypted end-to-end and that no one but the original user, not even the provider who hosts the cloud, can see the un-encrypted data. 

%However, this is not enough to ensure privacy. Suppose that a provider employee, Bob, wants to learn the private daily schedule of a user, Alice. Bob has access to the the encrypted cloud database, but cannot decrypt Alice's stored data to see its contents. If however, Bob can access side channel data available to the provider, such as sensor IP addresses or sensor registration histories, he can determine which updates in the database belong to Alice's sensor. Since her sensor uploads data immediately whenever Alice enters or leaves the home, Bob can look at the sensor update times to learn Alice's daily schedule, without ever having to decrypt her stored data. This type of attack generalizes to any event-driven update where the event time is a public function of the upload time. In order to prevent this class of attacks, we must decouple the relationship between event and upload timings.
}

\eat{
Consider an IoT provider that deploys smart sensors (i.e., Wifi access point, smart bulb, thermal sensor for person counting, etc) in a shopping mall with 2 buildings. The provider also creates two separate databases (DB1 and DB2) to back up the sensors’ event data for each building in encrypted form. By default, the data backup policy is set to backup as long as a new sensor event (i.e. a new connection to WIFI access point) occurred. Suppose at time $t$, only one customer, say Alice, entered the mall. Meanwhile, the database manager observes at time $t$, DB2 receives a backup request while DB1 does not, then the manager knows immediately that Alice went to Building 2, which is a breach of Alice's privacy. To protect the privacy of customers such as Alice from the database manager, IoT sensors must decouple the relationship between the sensor event time and database record insertion time.

Consider a taxi provider, such as New York Yellow Cab, that wants to guarantee privacy to its customers. Suppose the taxi provider maintains an encrypted Yellow Cab database on the cloud with Cloud Provider A. Whenever a passenger starts a taxi trip, Yellow Cab will add this trip information to the cloud database. For simplicity, assume that at any time, at most one passenger starts a trip. Suppose that Alice starts a trip at a specific time, say 11 am, from a sensitive location, say a courthouse. Then, a record `Alice started a trip from the courthouse at 11 am’ will immediately be added to the database. While the content of the record itself may be encrypted, the record insertion time is not hidden from Cloud Provider A. Since the pickup time is directly tied to insertion time, Cloud Provider A can use the insertion time to learn the sensitive pickup time, even though the pickup time field is encrypted in the record. Furthermore, Cloud Provider A can combine the pickup time data with other side information, such as courthouse trial records, to deduce the identity of Alice, breaching her privacy. To protect the privacy of customers such as Alice from Cloud Provider A, Yellow Cab must decouple the relationship between taxi pickup time and database record insertion time. 
}

\re{
%\ashwin{which in this case is exactly the same as the sensitive pickup time field in the data record being inserted. Knowledge of this pickup time and the taxi this update arose from could be combine with other side information to breach Alice's privacy}. In the presence of additional information,\kartik{what specific info?}\nt{From 10:00am to 12:00pm there is only one pickup happens around the hospital area?} this can potentially be linked back to the fact that Alice indeed visited the hospital at 11 am, which is a breach of Alice’s privacy.

%Ideally, if an update to the database at 11 am did not depend on whether Alice made a trip, Alice would not have lost her privacy. 

There are two straightforward solutions to solve this concern. The first option is to never upload any sensor data at all. While such a solution does provide necessary privacy, it does not provide us with the functionality of a database that supports updates. If an employee from the IoT provider queries the database to obtain, for example, the number of sensor events happened in a day, she will receive an inaccurate result. A second option is to back up the sensor event record at each time unit, independent of whether the sensor event actually occurred or not. Again, this does solve the privacy concern since the update does not depend on the sensor events at all. However, this introduces performance concerns: If sensor events occur relatively infrequently, then most updates are likely to be empty, or ``dummy'', updates, meaning that the provider will waste valuable resources on unnecessary computation. The above examples illustrate the 3-way trade-off between privacy, accuracy, and performance in the database synchronization problem. Each of the three approaches we discussed, immediate synchronization, no synchronization, and every time unit synchronization, achieves precisely two of the three properties, but not the third.

%The above example shows that there exists a 3-way trade-off between privacy, accuracy, and performance for the database synchronization problem. Each of the three solutions we discussed achieves precisely two of the three properties, but not the third one.

 %A second option is to add a record to the database at every time unit independent of whether a sensor event really occurred. Again, this does solve Alice’s privacy concern since the update does not depend on Alice's schedule. However, this introduces performance concerns: If sensor events occur relatively infrequently, then most updates are likely to be empty, or "dummy", updates, meaning that the provider will waste valuable resources on unnecessary computation. The above examples illustrate the 3-way trade-off between privacy, accuracy, and performance in the database synchronization problem. Each of the three approaches we discussed, immediate synchronization, no synchronization, and every time unit synchronization, achieves precisely two of the three properties, but not the third. %\kartik{can we claim all three not achievable?}

}

In this work, we build \appsystem, an append-only database outsourced by a data owner to one or more untrusted cloud service providers (server). In addition, a trusted \as, possibly the owner, is allowed to query the database at any point in time.  
To ensure consistency of the outsourced data, the \user synchronizes local records and updates the outsourced data. However, making updates on outsourced data structures may leak critical information. 
For instance, the \cs can potentially detect the size of synchronized records~\cite{bost2017forward, kamara2019computationally, amjad2019forward, poddar2020practical}. Cryptographic techniques such as ORAMs~\cite{stefanov2014practical} or structured encryption~\cite{hahn2014searchable} prevent leaking critical information on updates. However, all these methods are primarily designed to ensure that when an update occurs, attackers cannot learn sensitive information by observing changes in the outsourced data structure and not \emph{when} these changes happen. If the adversary/cloud server has access to the exact time of the updates, even if the system employs the techniques described above to protect individual updates, it can still result in privacy breaches of \user's data. The goal of \appsystem is to prevent such an update pattern leakage while still being performant and accurate. We now elaborate on our key contributions:

\vspace{-2mm}
\boldparagraph{Private update synchronization.} We introduce and formalize the problem of synchronizing updates to an encrypted database while hiding update patterns. Our goal is to provide a bounded differentially-private guarantee for any single update made to the cloud server. To navigate the 3-way trade-off between privacy, accuracy, and performance, we develop a framework where users can obtain customizable properties by modifying these parameters.

\vspace{-2mm}
\boldparagraph{Differentially-private update synchronization algorithms.} We provide two novel synchronization algorithms, DP-Timer and DP-ANT, that can obtain such trade-offs. The first algorithm, DP-Timer algorithm, parameterized by time $T$, synchronizes updates with the server every $T$ time. Thus, for a fixed parameter  $T$, to achieve a high amount of privacy, the algorithm asymptotes to never update the server (and hence, will not achieve accuracy). As we weaken our privacy, we can gracefully trade it for better accuracy. Similarly, by modifying $T$, we can obtain different trade-offs between accuracy and performance. The second algorithm DP-ANT, parameterized by a threshold $\theta$, synchronizes with the server when there are approximately $\theta$ records to update. Thus, for a fixed parameter $\theta$, when achieving high accuracy, the algorithm asymptotes to updating the server at each time unit and thus, poor performance. By reducing the accuracy requirement, we can gracefully trade it for better performance. Moreover, we can modify the parameter $\theta$ to obtain different trade-offs. Comparing the two algorithms, DP-ANT dynamically adjusts its synchronization frequency depending on the rate at which new records are received while DP-Timer adjusts the number of records to be updated each time it synchronizes.
%\kartik{fill}\nt{This should be SET}.

\vspace{-2mm}
\boldparagraph{Interoperability with existing encrypted databases.} We design our update synchronization framework such that it can interoperate with a large class of existing encrypted database solutions. To be concrete, we provide the precise constraints that should be satisfied by the encrypted database to be compatible with \appsystem, as well as classify encrypted databases based on what they leak about their inputs. 

\vspace{-2mm}
\boldparagraph{Evaluating \appsystem with encrypted databases.} We implement multiple instances of our synchronization algorithms with two encrypted database systems: Crypt$\epsilon$ and ObliDB. We evaluate the performance of the resulting system and the trade-offs provided by our algorithms on the New York City Yellow Cab and New York City Green Boro taxi trip record dataset. The evaluation results show that our DP strategies provide bounded errors with only a small performance overhead, which achieve up to 520x better in accuracy than never update method and 5.72x improvement in performance than update every time approach.

%\kartik{present a high-level result?}
\eat{
\vspace{-1mm}
\boldparagraph{Summary of contributions.} To summarize, we make the following contributions: %\kartik{possibly remove?}
\begin{itemize}[noitemsep,nolistsep]
    \item We introduce the database update synchronization problem (Sections~\ref{sec:description}~and~\ref{sec:model}).
    \item We present two differentially-private update synchronization algorithms DP-Timer and DP-ANT (Section~\ref{sec:sync}).
    \item We describe how \appsystem can interoperate with existing encrypted databases (Section~\ref{sec:cmp}).
    \item We evaluate the performance of DP-sync with encrypted databases (Section~\ref{sec:exp}).
\end{itemize}
}

% pattern the  none of these works consider information leaked due to synchronization of the database. updates have been updates are performed to a fixed set of memory locations~\cite{cash2014dynamic, naveed2014dynamic, ghareh2018new}. Moreover, most works, when considering updates, ignore the times at which these updates are performed~\cite{kamara2018sql, popa2012cryptdb, curtmola2011searchable}. In this work, we consider securing updates to a growing outsourced append-only database using differential privacy techniques.

\ignore{
%\johes{Consider changing this example to taxis to match the experiment and maybe have a figure describing the example}
Let us consider the following example where an adversary can breach privacy by using the timing information of updates. Suppose there are two WiFi providers, $A$, and $B$  in an airport. Each provider owns several WiFi access points and an encrypted database for their access points to back up connection sessions. The airport helps maintain two databases $\mathcal{D}_A$ and $\mathcal{D}_B$ for providers $A$ and $B$, respectively. By default, the access points will back up the session as soon as it receives a new connection and remains inactive if no new connection is made. For simplicity, we assume that any access point has, at any time, at most one new connection. Suppose at some time, Alice enters the airport and connects to one of the airport WiFi services. At the same time, the airport's server administrator observes that a backup request arrives at $\mathcal{D}_A$. If the administrator has some prior knowledge that only one person has entered the airport area during that time, then the administrator immediately learns that Alice is a member of provider $ A $, which is a breach of Alice's privacy. 
 
Alice's privacy breach is caused the inappropriate backup policy (data update) used by the access points. The ``instant backup'' strategy provides the attacker with additional information to deduce the user's membership. Two straightforward solutions for this problem are to change the backup policy to either backup at every time unit or to never backup. In the former strategy, access points send a backup at each time unit, regardless of whether the node truly receives a new connection. This strategy solves the privacy issue but raises performance concerns at the same time. The ``backup at every time unit'' strategy keeps both access points %\kartik{what are nodes?} 
and databases busy with backups, taking up a lot of system resources. The other strategy, which is never to back up, also safeguards data privacy with no additional performance burden. But this approach completely sacrifices data accuracy or availability, i.e., if a user loses a session or switches access points, they must re-authenticate and establish a new session with the access points. 

%An optimization of ``backup every time'', %\kartik{the quotes should be ``blah''} is to backup a fixed size batch with certain intervals. Such a method raises two more problems: First, how to decide the interval length, if the interval set too long then lots of newly connected user who loses a session or switches access points must re-establish new sessions (i.e. re-authentication). Second, how to select the batch size? If the batch size is too small, the system can't handle burst connections. If it's too large, then the performance equivalents to "backup every time" strategy. 

% Therefore, it is non-trivial to address the privacy issue introduced by synchronization patterns while maintaining data availability and system performance.
The design of an appropriate synchronization protocol requires a 3-way tradeoff between privacy, accuracy, and performance. In this work, we propose a secure outsourced growing database framework, \appsystem, with a differentially-private record synchronization strategy. More specifically, the proposed framework allows the \user to temporarily store some records locally. When a data update is needed, the synchronization strategy counts how many records have been received since last update, distorts the count with Laplace noise to ensure differential privacy and sends exactly as many as records as the noisy count specifies, hiding the true count from the server. Our framework offers considerable accuracy, efficiently handles record updates, and guarantees that the server learns nothing beyond some differentially-private update ``leakage'' profiles. The main contributions of this work are summarized as follows:
\begin{itemize}
\item We introduce and formalize a new type of leakage called update pattern leakage, which reveals the entire data update history applied on certain outsourced databases. 
\item We propose \appsystem, a novel design and security model for secure outsourced growing databases using differentially-private update patterns.
\item We design novel differentially-private synchronization algorithms that protect individual record updates.
\item We implement \appsystem using our proposed design and evaluate the tradeoff between privacy, accuracy, and performance using real world taxi trip data.%, a SOGDB system that only leaks differentially-private information to the untrusted server and evaluate \appsystem under our proposed security model. %\johes{evaluation could be another bullet here}

\end{itemize}
}
\eat{
{\bf Example 1 - General Case} : In an airport, there are two WIFI providers $A$ and $B$. Each provider owns several WIFI access points as well as an encrypted database for their access points to back up connection sessions. The airport helps to maintain these two database $\mathcal{D}_A$ and $\mathcal{D}_B$ for providers $A$ and $B$, respectively. By default, the access points will back up the session as soon as it receives 
a new connection and remains nonaction if no new connection is made. In addition, we assume that any wireless point has, at any timestamp, and at most, only one new connection. Now let's assume the following scenario, at some point of time, user $U$ enters the airport and connects to one of the airport WIFI service. At the same time, the airport's server administrator observed that a backup request arrives at $\mathcal{D}_A$. If the administrator has some priori knowledge that only one user has entered the airport area during that time. Then the administrator learns immediately that $U$ is a member of provider $A$, which is a serious privacy breach. The reason for the privacy breach is mainly due to the inappropriate backup policy 
taken by access points. When the $U$ enters the area, $A$'s access point receives a new connection and backup immediately, while $B$'s point chooses to silence as it receives nothing. Such an "instant backup" strategy thus providing the attacker with additional information to deduce the user's membership information. %\kartik{Seems like the second access point is not used any where? Is the access point a trusted client? the way we fix this is by assuming a client storage of 1. So, this example doesn't motivate what we are trying to do.}
%\nt{I modified the example 1, and state that it is because improper update strategy taken by access points that results in privacy breach }

{\bf Example 2 - A fix to Example 1 } : Consider the same setting as Example 1, to solve the privacy issue. The backup policy for all access points is now configured as "backup at any time", which can be interpreted as a policy that sends a backup request at any time, regardless of whether the node truly receives a new connection. Wireless access points will choose to back up new session if they do received new connections or backup a dummy session if they receives nothing. This backup strategy solves the privacy breach issue mentioned earlier as at any moment, the access points of both providers $A$ and $B$ submit backup requests to the backup database, which provides plausible deniability for $U$'s membership information. However, this strategy, which 
appears to solve the privacy issue well, raises performance concern at the same time. The "backup at any time" strategy keeps both nodes and databases busy with backups, taking up a lot of system resources. \kartik{Minor: it looks like access point B is not used here.}

{\bf Example 3 - A fix to Example 2}: Consider the same setting as Example 1, now the access points back up sessions every 4 hours. When backing up, the access point packs up all new connections arrive within last 4 hours together with some dummy sessions, then push the packed records as a batch backup to the database. Now the system is not overloaded with backup transactions, however, the nodes are set to back up every four hours. This means that within four hours of a backup, any newly connected user who loses a session or switches access points must re-establish new sessions. \kartik{is it bad to establish new sessions?} If a large number of users entered the airport and use the wireless service within the 4 hours of the access points not being backed up, there will be many users who will have to make frequent new sessions.

{\bf Conclusion}: Combining the above three examples, we conclude: First, inappropriate data update strategies taken by the client under the outsourced database model can lead to serious privacy leakage issues. Secondly, how to choose or design a data synchronization strategy is really an optimization problem for the triple objectives of privacy, performance and utility.
}

\vspace{-2mm}
\section{Problem Statement}\label{sec:ps}
%In this section, we formulate our main problem in Section~\ref{sec:pf} and provide the design principle in Section~\ref{sec:designp}.
%\subsection{Problem Formulation}~\label{sec:pf}
The overarching goal of this work is to build a generic framework for secure outsourced databases that limits information leakage due to database updates. We must ensure that the \cs, which receives outsourced data, cannot learn unauthorized information about that data, i.e., the true update history. We achieve this by proposing private synchronization strategies that the \user may use to hide both how many records are currently being outsourced and when those records were originally inserted. Though there are simple methods that effectively mask the aforementioned update history, significant tradeoffs are required. For example, one may simply prohibit the \user from updating the outsourced database, or force them to update at predefined time intervals, regardless of whether they actually need to. Though both approaches ensure that the true update history is masked, they either entirely sacrifice data availability on the outsourced database or incur a significant performance overhead, respectively. Navigating the design space of private synchronization protocols requires balancing a 3-way tradeoff between privacy, accuracy, and performance. To tackle this challenge, we formalize our research problems as follows:
\begin{itemize}
    \item Build a generic framework that ensures an \user's database update behavior adheres to private data synchronization policies, while supporting existing encrypted databases.
    \item Design private synchronization algorithms that (i) hide an \user's update history and (ii) balance the trade-off between privacy, accuracy and efficiency.
\end{itemize}

In addition to the research problems above, we require our design to satisfy the following principles.

\vspace{-1mm}
\boldparagraph{P1-Private updates with a differentially private guarantee.}
The proposed framework ensures that any information about a single update leaked to a semi-honest \cs is bounded by a differentially private guarantee. We formally define this in Definition~\ref{def:dpsogdb}. 
%\johes{maybe guarantee instead of sense}.

\vspace{-1mm}
\boldparagraph{P2-Configurable privacy, accuracy and performance.} Rather than providing a fixed configuration, we develop a framework where users can customize the level of privacy, accuracy, and performance. For example, users can trade privacy for better accuracy and/or improved performance. 
%\kartik{SOLVED:We develop a framework where users can customize the levels of privacy, accuracy, and performance. For example, users can trade privacy for better accuracy and/or improved performance.}

\vspace{-1mm}
\boldparagraph{P3-Consistent eventually.} The framework and synchronization algorithms should allow short periods of data inconsistency between the logical (held by the \user) and the outsourced (held by the \cs) databases. To abstract this guarantee, we follow the principles in~\cite{burckhardt2014principles} and define the concept of \emph{consistent eventually} for our framework as follows. First, the outsourced database can temporarily lag behind the logical database by a number of records. However, once the \user stops receiving new data, there will eventually be no logical gaps. Second, all data should be updated to the \cs in the same order in which they were received by the \user. In some cases, the consistent eventually definition can be relaxed by removing the second condition. In this work, we implement our framework to satisfy the definition without this relaxation. 
%\kartik{I assume this came from the link Ashwin sent? perhaps cite it here?} \nt{Note to myself, May have 2-3 works need to cite, check back when other changes have applied.}

\vspace{-1mm}
\boldparagraph{P4-Interoperable with existing encrypted database solutions} The framework should be interoperable with existing encrypted databases. However, there are some constraints. \re{First, the encrypted databases should encrypt each record independently into a separate ciphertext. Schemes that encrypt data into a fixed size indivisible ciphertext (i.e., the ciphertext batching in Microsoft-SEAL~\cite{laine2017simple}) do not qualify. Since batching may reveal additional information, such as the maximum possible records per batch. Second, the database should support or be extensible to support data updates (insertion of new records). Thus, a completely static scheme~\cite{shi2007multi} is incompatible. In addition, our security model assumes the database's update leakage can be profiled as a function solely related to the update pattern. Therefore, dynamic databases with update protocol leaks more than the update pattern~\cite{kamara2012dynamic, naveed2014dynamic} are also ineligible. Third, the corresponding query protocol should not reveal the exact access pattern~\cite{goldreich2009foundations} or query volume~\cite{kellaris2016generic} information. Despite these constraints, our framework is generic enough to support a large number of existing encrypted databases such as~\cite{chowdhury2019cryptc,cash2014dynamic, vinayagamurthy2019stealthdb, eskandarian2017oblidb, wu2019servedb, kamara2012dynamic, agrawal2004order, boldyreva2009order,bater2018shrinkwrap, amjad2019forward, ghareh2018new, bost2017forward}. Later, in Section~\ref{sec:cmp}, we provide a detailed discussion on the compatibility of existing encrypted database schemes with \appsystem.}

%The design of \system with DP update pattern is inherently generic and does not rely on specific cryptographic primitives or frameworks. \kartik{Pose it as 'the goal of our framework is to be generic. Though there are some constraints: we want each record to be encrypted separately, etc. then say, despite these constraints, it is generic enough to support these databases. Later, we compare with two of them.'} 

%It can be architected on top of many existing encrypted databases~\cite{chowdhury2019cryptc,cash2014dynamic, vinayagamurthy2019stealthdb, eskandarian2017oblidb, wu2019servedb, kamara2012dynamic, agrawal2004order, boldyreva2009order}. 

%Therefore, to make it more general, we require that the overall design of our framework is independent of the underlying primitives. This indicates some flexibility in the choice of specific encrypted database schemes on which our framework is based. Actually, any encrypted database can be used by our framework as long as it satisfies the following three conditions. Firstly, the encrypted database should have ``fixed up-link overhead'', schemes that pack and encrypt data into a fixed size indivisible ciphertext (i.e. the ciphertext batching in Microsoft-SEAL~\cite{laine2017simple} is incompatible with our design). Second,  In Section~\ref{sec:cmp}, we provide a detailed discussion of the compatibility of existing encrypted database schemes with our proposed framework.
%\kartik{do we need something wrt dummy records?}

\eat{
\begin{definition}[Consistent eventually]\label{def:const} We say a \system achieves consistent eventually if it satisfies the following 2 conditions:  {\bf(1)} For a given $\hat{t}$, where $\forall j > \hat{t}, u_j = \emptyset$, there exists $\hat{t}<t<\infty$, s.t. $~\mathcal{D}_{t} = \left(\mathcal{D}_{t}\cap \dec(\mathcal{DS}_t) \right)$.  {\bf(2)} $\forall t \in \mathbb{N}^{+}, \nexists ~0< t_1 < t_2 \leq t$, s.t.  \smash{$u_{t_1} \notin \left(\mathcal{D}_{t}\cap \dec(\mathcal{DS}_t) \right)$}, \smash{$u_{t_2} \in \left(\mathcal{D}_{t}\cap \dec(\mathcal{DS}_t) \right)$}.
\end{definition}
Definition~\ref{def:const} points out that the outsourced database can temporarily lag behind the logical one by a number of records. However, once the \user stops receiving new data, there will eventually be no logical gaps (condition 1). Moreover, all data should be updated to the \cs in the same order in which they were received by the \user (condition 2). In some use cases, the consistent eventually definition can be relaxed by removing the second condition. In this work, we implement our framework to satisfy the strong (unrelaxed) consistent eventually property. 
}

\section{\appsystem Description}
\label{sec:description}
In this section, we introduce \appsystem, a generic framework for encrypted databases that hides update pattern leakage. The framework does not require changes to the internal components of the encrypted database, but rather imposes restrictions on the \user's synchronization strategy.  We illustrate the general architecture and components of \appsystem in Section~\ref{sec:workflow} and Section~\ref{sec:component}, respectively.  

\vspace{-1mm}
\subsection{Framework Overview}\label{sec:workflow}
Our framework consists of an underlying encrypted database with three basic protocols, $\mathsf{edb}=(\mathsf{Setup}, \mathsf{Update}, \mathsf{Query})$, a synchronization strategy $\sync$, and a local cache $\sigma$. Our framework also defines a dummy data type that, once encrypted, is indistinguishable from the true outsourced data. %\kartik{solved:start wtih database, then sync and then cache and support for dummy records.} 
The local cache $\sigma$ is a lightweight storage structure that temporarily holds data received by the \user, while \sync determines when the \user needs to synchronize the cached data to the \cs (poses an update) and how many records are required for each synchronization. \appsystem makes no changes to the $\mathsf{edb}$ and will fully inherit all of its cryptographic primitives and protocols. Figure~\ref{fig:arch} illustrates the general workflow of \appsystem. %That is $\setup = \mathsf{edb.Setup}$, $\update = \mathsf{edb.Update}$, and $\query = \mathsf{edb.Query}$. 

%$\appsystem = (\sync, \sigma, \mathsf{edb})$,
\begin{figure}[h]
\centering
\includegraphics[width=0.8\linewidth]{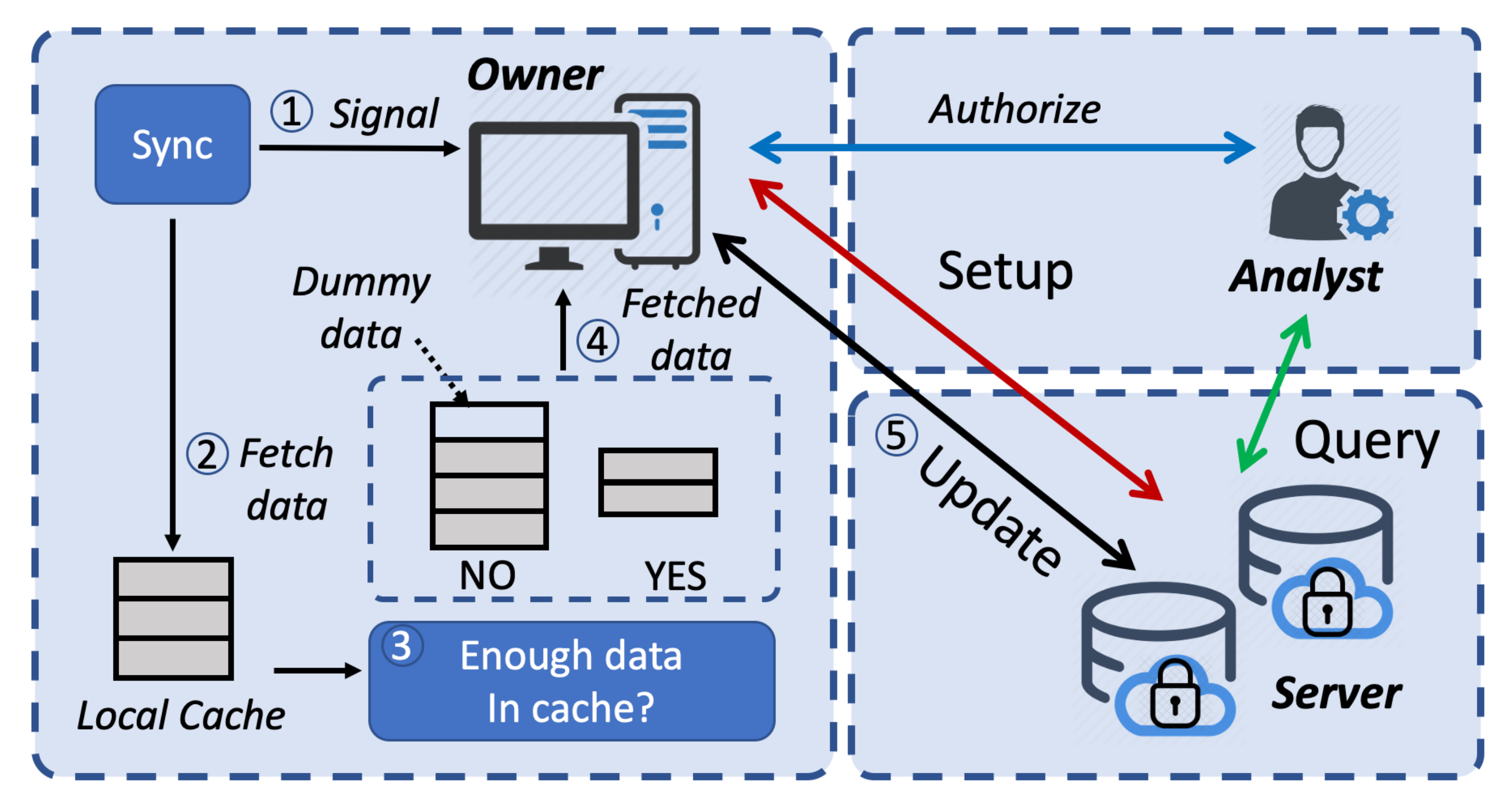}
\vspace{-4mm}
\caption{Overview of \appsystem's architecture.}
\label{fig:arch}
\end{figure}

 Our proposed framework operates as follows. Initially, the \user sets up a synchronization strategy $ \sync$ and a local cache $\sigma$, then authorizes the \as. The \user starts with an initial database with which it invokes \sync to obtain a set of records, $\gamma_0$, to be outsourced first. The \user then runs the setup protocol ($\mathsf{edb.Setup}$) with $\gamma_0$ as the input. An initial outsourced data structure is then created and stored on the \cs. For each subsequent time step, whenever the $\sync$ algorithm signals the need for synchronization, the \user reads relevant records from the cache and inputs them to the update protocol ($\mathsf{edb.Update}$) to update the outsourced structure. When there is less data than needed, the \user inputs sufficiently many dummy records in addition to the cached data.
 
 \eat{
 where there is more data in the cache than the amount of data needed for synchronization, we can follow different policies for selecting records, e.g.,  LIFO, FIFO. By default, \appsystem uses the FIFO policy. Otherwise, in addition to the data in cache,
 
 }
 %Otherwise, if the local cache does not have enough data, it will generate additional dummy records and synchronize together with the data in the local cache. 
 %\kartik{Otherwise, in addition to the data in cache, the \user inputs sufficiently many dummy records.} 
 
 Since all records are encrypted, the server does not know which records are dummy records and which are true records. The outsourced data structure will only change if the \user runs the update protocol, in other words, if \sync does not signal, then the outsourced structure remain unchanged.
 %$\mathcal{DS}_t = \mathcal{DS}_{t-1}$.
 The analyst independently creates queries and runs the query protocol ($\mathsf{edb.Query}$) to make requests. The \cs evaluates each query and returns the result to \as. For simplicity, we assume that all queries arrive instantaneously and will be executed immediately.

 %$\mathcal{DS}_t$.

  %whenever a record comes in, the \user stores it in the local cache. If synchronization is needed (signaled by \sync), the \user then reads as many records as are needed to be synchronized from the cache and uses the acquired data as input to run the $\update$ protocol. \kartik{Whenever $\sync$ algorithm signals the need for synchronization, the \user reads relevant records from the cache and inputs it to the $\update$ protocol to obtain $\mathcal{DS}_t$.}
 
 %The analyst independently creates queries and runs protocol $\prod_{\text{query}}$ to make query requests at each time $t$.  The \cs securely \kartik{what does securely mean here?} processes each query and evaluates them over the latest outsourced data structure, and returns corresponding results to the \as. For the sake of simplicity, unless otherwise stated,\kartik{we always assume this right?} we assume that all queries posted at time $t$ arrive at the \cs immediately, and it is guaranteed that they will only be applied on $\mathcal{DS}_t$, no query will be delayed for future execution. \kartik{The analyst independently creates queries and runs protocol $\prod_{\text{query}}$ to make query requests at each time $t$. The \cs evaluates each query $q_t$ and returns the result to \as. For simplicity, we assume that a query $q_t$ arrives instantaneously and will be executed over $\mathcal{DS}_t$.}

 %We introduce the details of the local cache design, and the synchronization strategies, respectively, in Section~\ref{sec:lcache}, and~\ref{sec:sync}.

\eat{
Besides the above descriptions, for the sake of simplicity, unless otherwise stated, we consider the following assumptions: {\bf A-1. Atomic protocols}. All protocols are considered atomic in the sense that once a protocol is executing, it can not be preempted by other protocols, and any other protocols can only be executed after the current running protocol exits.  All communication (transmission of any message) between participants is considered to have zero latency and all protocol execution is done immediately. For instance, when the \user issues $\update$, the transmitted records will immediately arrive at the \cs and update the outsourced storage on the \cs. }
 
%

%\nt{Note that, \sync decides when the \user needs to synchronize according to it's internal states, and will guide the \user to pick the proper size of records for synchronization. }

\eat{
If \sync signals for synchronization, then the \user follows \sync's instruction and runs $\update$ to update the outsourced database on the \cs, otherwise the \user does nothing and move on to the next timestamp. The \as, on the other hand, independently issues $\query$ protocol with a set of queries, denoted as $q_t$. The \cs evaluates those queries over the latest outsourced data structure and returns the corresponding results to the \as. All three participants repeat the steps over time until certain termination condition is triggered.

\sync decides when the \user needs to synchronize according to it's internal states, and will guide the \user to pick the proper size of records for synchronization. 
 The analyst independently creates queries and calls protocol $\query$ to make query requests at each time $t$. Within the protocol, \as's queries will be translated to implementation-specific security protocols using Crypt$\epsilon$'s query translator. Next, the \cs executes these security protocols over the outsourced data and returns corresponding encrypted results to the \as. The \as decrypts the returned encryptions to obtains the plaintext results. 
 }
 
\eat{
\subsection{Data Provisioning}
Each data obtained by user $\user$ is considered to have the form of $\langle A_1,...A_l\rangle$ where ${A}_j$ is an attribute. $\user$ needs to provision the private data before spawning any protocols. The data are encoded in its respective per attribute one-hot-encoding format. The one-hot-encoding is a way of representing categorical attributes and is illustrated by the following example. Assume a record schema is given by  $\langle Age \in (1,100), Gender \in (0,1) \rangle$, then the corresponding one-hot-encoding representation for a record of $\langle 28, Female\rangle$, is given by ${\langle[\underbrace{0,\ldots,0}_{27},1,\underbrace{0,\ldots,0}_{72}],[0,1]\rangle}$. \kartik{revisit this. why do we need attributes or one-hot encoding?} \nt{This is the part I want to discuss, we haven't yet discussed about how records are encoded or encrypted. }

For the prototype \appsystem presented in this work, we implemented data provision method using the aforementioned per attribute one-hot-encoding. However, the provision can be substituted by other encoding scheme like multi-attribute one-hot-encoding, range based encoding, hashing based encoding, etc. In this work we choose one-hot-encoding for prototype \appsystem implementation as it is a general representation that can be easily translated to other encoding schemes. Moreover, some other works such as Crypt$\epsilon$ \cite{chowdhury2019cryptc} has demonstrated the "Crypto Friendly" feature of one-hot-encoding. By taking one-hot-encoding, Crypt$\epsilon$ supports a rich set of secure queries.
}

%Another example is that in this paper, we use FIFO policy to select records from the local cache to be synchronized. However, one can apply many other strategies to read the local data. In some scenarios, the \as may only query the most recent data. For example, in the covid19 use case, the \as only interested in data from the most recent 14 days. Therefore, a LIFO policy is better for read the cached data, as the strategy ensures that newly received data is outsourced with a higher priority than old records.
\subsection{Framework Components}\label{sec:component}
%In this section, we introduce the local cache design (Section~\ref{sec:lcache}) and the dummy records used in our framework (Section~\ref{sec:dumy})
\subsubsection{Local cache}\label{sec:lcache}
The local cache is an array $\sigma[1,2,3...]$ of memory blocks, where each $\sigma[i]$ represents a memory block that stores a record. By default, the local cache in \appsystem is designed as a FIFO queue that supports three types of basic operations:
%\cw{I change notations of local cache from \lcache to $\sigma$ for spacing concern}
\begin{enumerate}
    %\item {\bf Initialization} \texttt{init()}: Create a local cache object, allocate memory space for data storage, and set up hyper-parameters related to cache configurations, such as cache flush rate.
    \item {\bf Get cache length ($\len(\sigma)$).} The operation calculates how many records are currently stored in the local cache, and returns an integer count as the result.
    
    \item {\bf Write cache ($\writec(\sigma, r)$).} The write cache operation takes as input a record $r$ and appends the record to the end of the current local cache, denoted as $\sigma \mathbin\Vert r \gets \writec(\sigma, r)$.

    \item {\bf Read cache ($\readc(\sigma, n)$).} Given a read size $n$, if $n \leq \len(\sigma)$, the operation pops out the first $n$ records, $\sigma\left[1,...,n\right]$, in the local cache. Otherwise, the operation pops all records in $\sigma$ along with a number of dummy records equal to $|n - \len(\sigma)|$. 
    %The operator then returns the union set of the dummy and cached data.
    
\end{enumerate}
The FIFO mode ensures all records are uploaded in the same order they were received by the \user. In fact, the local cache design is flexible and can be replaced with other design scenarios. For example, it can be designed with LIFO mode if the analyst is only interested in the most recently received records.
%if the query of an outsourced database only uses data from the most recent period (the most recent window query), and the order in which the data is uploaded does not affect the query result. Then the \user can actually design the local cache in LIFO mode so that the newly received data has a higher priority for synchronized to the \cs, which provides better accuracy for the most recent window queries. Note that when applying LIFO mode to design the local cache, \appsystem will only guarantee the {\it relaxed consistent eventually} definition. %\kartik{CHECK:opportunity to reduce text, if needed.} 
\subsubsection{Dummy records}\label{sec:dumy}
%\subsubsection{Types of dummy records}
Dummy records have been widely used in recent encrypted database designs~\cite{arasu2013oblivious, eskandarian2017oblidb, kellaris2017accessing, he2017composing, bater2018shrinkwrap, mazloom2018secure, asharov2019locality, allen2019algorithmic} to hide access patterns, inflate the storage size and/or distort the query response volume. In general, dummy data is a special data type that cannot be distinguished from real outsourced data when encrypted. Moreover, the inclusion of such dummy data does not affect the correctness of query results. 

\subsubsection{Synchronization strategy}
The synchronization strategy \sync takes on the role of instructing the \user how to synchronize the local data. It decides when to synchronize their local records and guides the \user to pick the proper data to be synchronized. We explain in detail the design of \sync in section~\ref{sec:sync}.

%run in the background as a ``daemon process''.\kartik{SOLVED:I don't think it is a daemon process. It is a function that outputs how many records to synchronize.} 

%In Table~\ref{tab:sumenc}, we summarize some existing secure outsourced database scheme that supports dummy data by default, and list the type of dummy data employed by each scheme. 

\eat{
\noindent{\bf (1) Encrypted search.} The dummy data for encrypted search schemes refer to the encrypted bit strings, memory blocks or files that are indistinguishable from not yet decrypted records. In such a way, the \cs cannot distinguish between dummy and non-dummy data in polynomial time, but the \as can easily exclude these dummy records from the query results after decryption. 

\noindent{\bf (2) Secure relational database.} For secure relational databases, to support dummy data, it requires to extend all data with an additional ``{\it isDummy}'' attribute. And the dummy data is the random tuples with {\it isDummy = True}. The \cs will include both dummy and real outsourced data for query processing, and the \as needs to eliminate dummy answers from the query results after decryption. For those scenarios where its corresponding query protocol is obviousness and leaks nothing about the size pattern~\cite{allen2019algorithmic,chowdhury2019cryptc}, we can apply query rewriting to let the \cs ignores those dummy data when processing queries (See Appendix~\ref{sec:rewrite} for details). 

\noindent{\bf (3) Specialized database.} Some specialized database scheme may have their respective own dummy data types. For instance, the dummy data for blockchain type encrypted database~\cite{cryptoeprint:2020:827} can be an encrypted blank block or invalid smart contract. For databases where the data is encoded, the dummy data can be a specially encoded record. For example the dummy data for Crypt$\epsilon$~\cite{chowdhury2019cryptc} is a vector with all 0 bits.
}

\eat{
The FIFO mode ensures that each record is outsourced in the order in which it was received by the \user, thus satisfying the strong consistency requirement.

In general, the \user's storage is limited in size and lightweight, thus cache flush mechanisms are needed to prevent the local cache from growing too large. The cache flush mechanism is a policy that flushes a portion of the cached data from time to time, and it is used to control the cache size. To ensure that no data is lost, the flushed data is immediately backed up to the cloud. There are many types of cache flushing strategies, a common method is to trigger a whole cache flush (read all records from the cache) whenever the local cache size exceeds a specified threshold (MySQL applies this policy to flush the buffer pool~\footnote{https://dev.mysql.com/doc/refman/8.0/en/innodb-buffer-pool-flushing.html}). This approach leaks the cache size when flushes if the attacker knows the flush threshold, which violates the security model we proposed earlier. To meet the security and privacy guarantee, we provide two cache flush strategies as follows:

\boldparagraph{Flush after k times of synchronization}: A cache flush will take place after every $k$ times of data synchronizations, with a fixed amount of data been flushed.

\boldparagraph{Flush after k time units} Cache flushes are performed after a fixed time interval, with a fixed amount of data being flushed at a time. Since cache flushing doesn't need to be done frequently, the time interval here is often set relatively large.\kartik{doesn't this depend on the rate at which updates are performed and the cache size?} \kartik{on first read, the diff between flush and sync is not confusing}
\kartik{please check the grammar in this section and the next}
}

%The depending on the synchronization policy selected, the user  periodically perform \texttt{read()} or \texttt{write()} operations. Two trigger operations \texttt{reset()} and \texttt{flush()} keep tracks the size of local cache. Whenever the cache size is less than or exceeds the predefined thresholds, the operations will perform the corresponding actions accordingly. To demonstrate the details of each operation more clearly, we have graphically illustrated the key operations above, please refer to Figure \ref{fig:dcache} for details.

\eat{
\begin{figure}[h]
  \centering
  \begin{minipage}[b]{0.5\textwidth}
  \includegraphics[width=\linewidth]{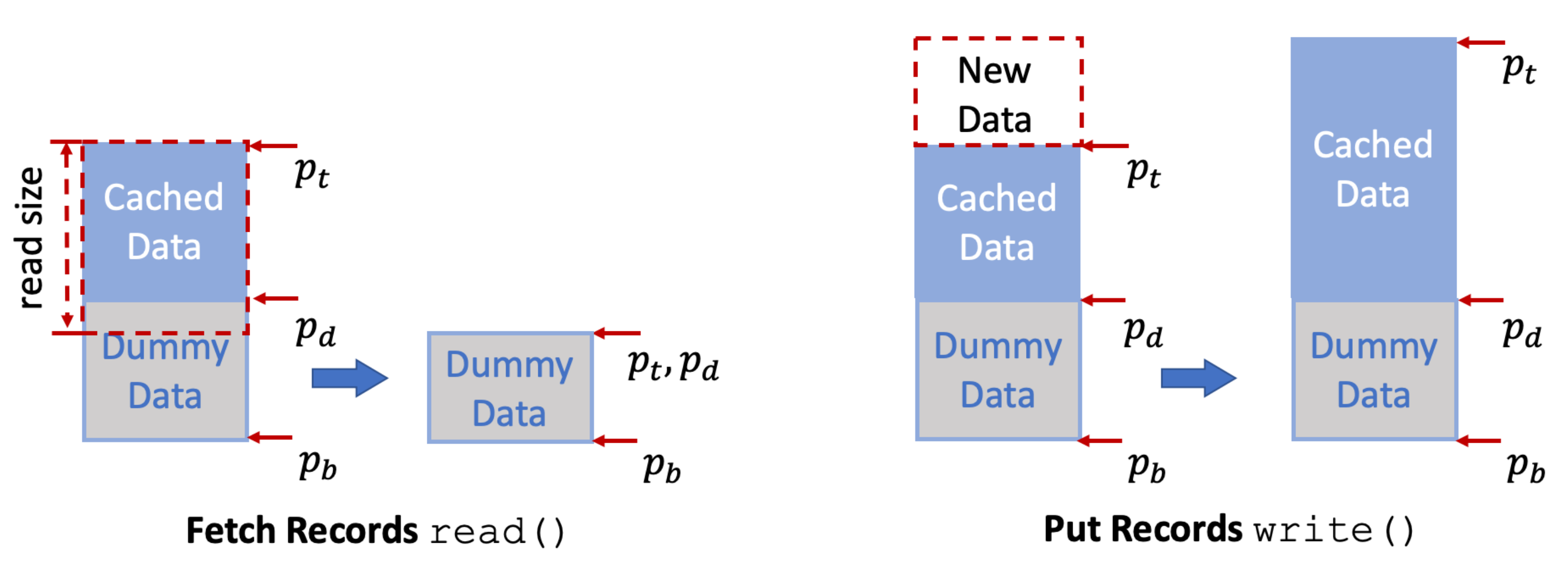}
  \subcaption{Fetch and Put records operations}
  \end{minipage}
 \begin{minipage}[b]{0.5\textwidth}
  \includegraphics[width=\linewidth]{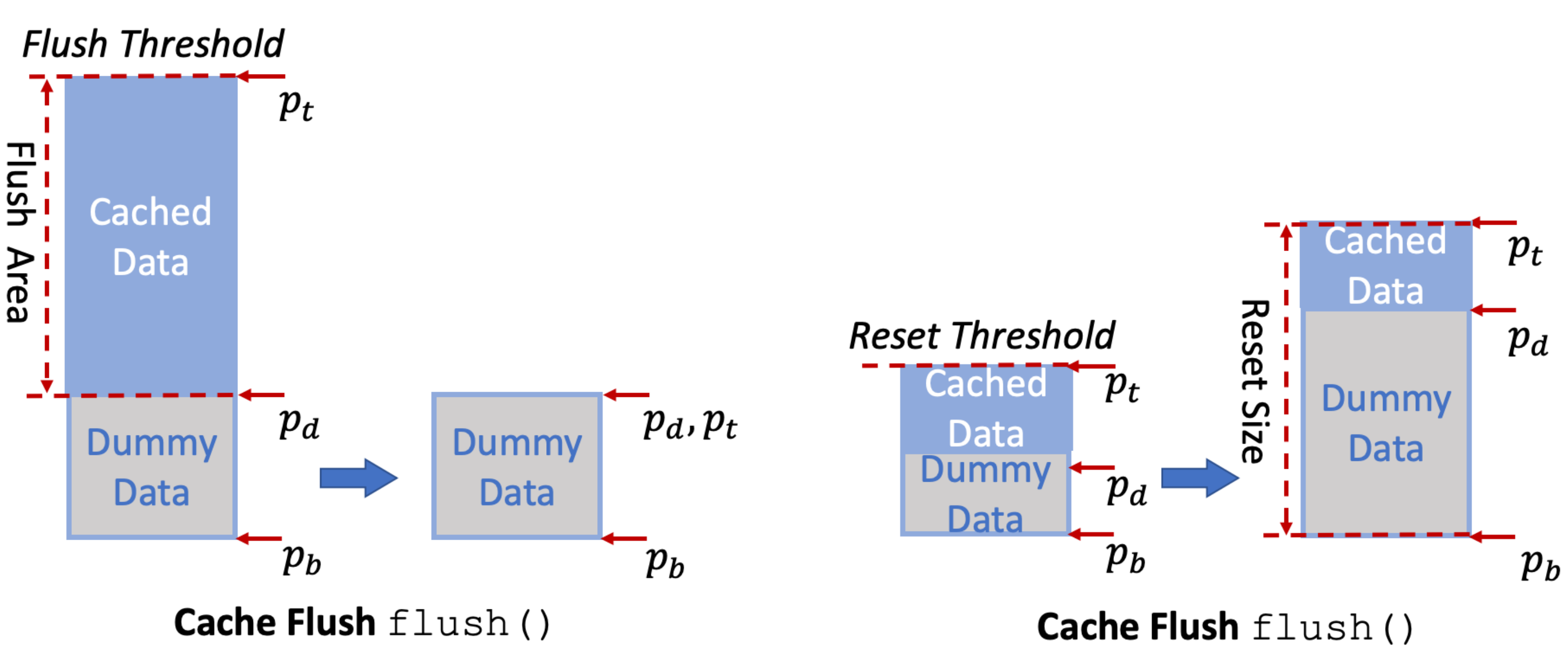}
  \subcaption{Cache flushing and reset operations}
  \end{minipage}
  \caption{The details of cache operations: (a) The \texttt{read()} and \texttt{write()} over local cache; (b) The cache flush operation \texttt{flush()} when cache size exceeds upper bound threshold and the cache reset operation \texttt{reset()} when cache size is smaller than lower bound threshold. }
  \label{fig:dcache}
\end{figure}

}
\section{\appsystem Model}
\label{sec:model}
%Following the model of the static outsourced database in~\cite{kellaris2016generic}, we extend it to growing databases.\kartik{can remove the previous sentence and mention this later} 
In this section, we describe the abstract model of \appsystem as a secure outsourced growing database, including the key definitions (Section~\ref{sec:keydef}), security model (Section~\ref{sec:spmodel}), \re{privacy semantics (Section~\ref{sec:semantics})}, and evaluation metrics (Section~\ref{sec:measure}).

\subsection{Secure Outsourced Growing Database}
\label{sec:keydef}
We begin by introducing the main concepts of outsourced growing databases and the notations used in this work. A summary of key concepts and notations is provided in Table~\ref{tab:notations}.  
\begin{table}[h]
    \centering
    \begin{tabular}{r|l}
        $\mathcal{D}_0$ & initial logical database   \\
        $u_t$ & logical update at time $t$; can be a single record or null\\
        $\gamma_t$ & real update at $t$, can be a set of records or null\\
        $\mathcal{D}_t$ & logical database at time $t$; $\mathcal{D}_t = \{ \mathcal{D}_0 \cup u_1 \cup u_2 \ldots \cup u_t\}$ \\
        $\mathcal{D}$ & the logical instance of a growing database\\
        $\mathcal{DS}_t$ & physical data on the server at time $t$ \\
        $\mathcal{DS}$ & the physical data on the server over time\\
        $q_t$ & list of queries at time $t$ \\ 
        $Q$ & set of queries over time, where $Q=\{q_{t}\}_{t\geq 0}$\\
        \system & secure outsourced growing database. %\kartik{can we call it extendable DB or so?}
    \end{tabular}
    \caption{Summary of notations.} 
    \label{tab:notations}
\vspace{-4mm}
\end{table}

A growing database consists of an initial database $\mathcal{D}_0$ and a set of logical updates $U = \{u_t\}_{t\geq0}$ to be appended to $\mathcal{D}_0$, where $u_t \in U$ is either a single record or $\emptyset$. The former corresponds to the data received at $t$, while $\emptyset$ indicates no data arrives. We consider the case where at most one record arrives at any time unit for the sake of simplicity, however this can be generalized to the case where multiple records arrive in one time unit. We define the growing database as $\mathcal{D} = \{\mathcal{D}_t\}_{t \geq 0}$, where $\mathcal{D}_t$ is the logical database at time $t$, and $\mathcal{D}_t = \{ \mathcal{D}_0 \cup u_1 \cup u_2 \ldots \cup u_t\}$. We stress that when we say a growing database has length $L$, it means that there could be up to $L$ logical updates in $U$, that is $|U |= L$. %\kartik{are you saying that there could have been up to $L$ updates until this time?} 
%\am{SOLVED:The language "query is an algorithm performed ..." seems a bit weird. Given this is a DB conference, you don't need to define queries. Rather you can just say: We consider databases that support select, project, join and  aggregations (please check is this is correct). }
We consider databases that support select (search), project, join and aggregations. We use $Q=\{q_{t}\}_{t\geq 0}$ to denote the set of queries evaluated over a growing database, where $q_{t}$ is the query over $\mathcal{D}_t$. %\kartik{SOLVED:not sure: should we say that this can be generalized to multiple records in one time unit?}
%A query is an algorithm performed on a database, denoted as $q(\mathcal{D})$. For example, counting queries are algorithms that count the total number of records in a database that satisfy certain predicates.

There are three entities in the secure outsourced data model: the \user, the \cs, and the \as. The \user holds a logical database, encrypts and outsources it to the \cs, and continually updates the outsourced structure with new data. The \cs stores the outsourced structure, on which it processes queries sent by an authorized \as. For growing databases, all potential updates posted by the \user will be insertions only. We denote the records to be updated each time as $\gamma_t$, which can be a collection of records, or empty (no update has occurred). We use $\mathcal{DS} = \{\mathcal{DS}_t\}_{t \geq 0}$ to represent the outsourced structure over time, where $\mathcal{DS}_t$ is an instance of outsourced structure at time $t$. Typically, an instance of the outsourced structure contains a set of encrypted records as well as an optional secure data structure (i.e., secure index~\cite{chase2010structured}). We now define the syntax of a secure outsourced database as follows:
%{We assume that the \user encrypts the data using a semantically secure encryption scheme, and
\vspace{-1mm}
\begin{definition}[Secure Outsourced Growing Database]\label{def:sodb} A secure outsourced database is a suite of three protocols and a polynomial-time algorithm with the following specification: %\kartik{is enc-dec important?}\nt{I made changes here, for definition I think we only need to introduce the three major protocols and the Sync, for enc, dec we can introduce later if we talk about the protocol implementation details.}
%\kartik{SOLVED:need to explain somewhere that the tuple corresponds to inputs and outputs from owner, server, and analyst}
%\am{you can drop the itemize - SOLVED}
%\begin{description}
    %\setlength{\itemsep}{5pt} 
    %\item $(\pk, \sk) \gets\gen(\lambda)$: It takes as input of a security parameter $\lambda$ and returns a key pair \pk, and \sk.
    %\item $ct \gets \enc(m, \pk)$: It takes as input of a message and the encryption key \pk, and outputs a ciphertext $ct$ w.r.t. $m$.
    %\item $m\gets \dec(ct, \sk)$: It takes as input of a ciphertext $ct$ and a secret key \sk, and outputs the corresponding message $m$ of $ct$.
    %\item $sc \gets \auth(\sk, )$
    
    \vspace{-1mm}
    \boldparagraph{}$(\perp, \mathcal{DS}_0, \perp) \leftarrow \setup((\lambda, \mathcal{D}_0), \perp, \perp)$: is a protocol that takes as input a security parameter $\lambda$, and an initial database $\mathcal{D}_0$ from the \user. The protocol sets up the internal states of the \system system and outputs an outsourced database $\mathcal{DS}_0$ to the \cs.
    %(or necessary hyper-parameters i.e. encryption mode)%\kartik{SOLVED,what are hyper-parameters?} 
    
    \vspace{-1mm}
    \boldparagraph{}{$(\perp, \mathcal{DS}'_t, \perp) \leftarrow \update(\gamma, \mathcal{DS}_t, \perp)$: is a protocol that takes an outsourced structure $\mathcal{DS}_t$ from the \cs, and a collection of records $\gamma$ from the \user, which will be inserted into the outsourced data. The protocol updates the outsourced structure and outputs the updated structure $\mathcal{DS}'_t$ to \cs. }
    
    \vspace{-1mm}
    \boldparagraph{}  $(\perp, \perp, a_t) \leftarrow \query(\perp, \mathcal{DS}_t, q_t)$: is a protocol that takes an outsourced database $\mathcal{DS}_t$ from the \cs and a set of queries $q_t$ from the \as. The protocol reveals the answers $a_t$ to the \as. 
    
    \vspace{-1mm}
    \boldparagraph{} $\sync(\mathcal{D})$: is a (possibly probabilistic) stateful algorithm that takes as input a logical growing database $\mathcal{D}$. The protocols signals the \user to update the outsourced database from time to time, depending on its internal states. %\kartik{SOLVED:perhaps sync should go last?}
%\end{description}
\end{definition}
%\am{SOLVED:It seems like Setup and Query are part of the encrypted database. Whereas Update and Sync are part of the new solution we propose. Should we just be defining the parts of the solution we are proposing? Do we need setup and query?}

The notation $(\mathsf{c_{out}}, \mathsf{s_{out}}, \mathsf{a_{out}})\gets{\mathsf{protocol}}(\mathsf{c_{in}}, \mathsf{s_{in}}, \mathsf{a_{in}})$ is used to denote a protocol among the \user, \cs and \as, where $\mathsf{c_{in}}, \mathsf{s_{in}}$, and $\mathsf{a_{in}}$ denote the inputs of the \user, \cs and \as, respectively, and $\mathsf{c_{out}}, \mathsf{s_{out}}$, and $\mathsf{a_{out}}$ are the outputs of the \user, \cs and \as. We use the symbol $\perp$ to represent nothing input or output. We generally follow the abstract model described in~\cite{kellaris2016generic}.  %However~\cite{kellaris2016generic} models the static setting whereas we consider dynamic setting with updates.
%\kartik{SOLVED:mention here that we are following the modeling of [46]. However they model the static setting whereas we consider dynamic setting with updates.} 
However, the above syntax refers to the {\it dynamic} setting, where the scheme allows the \user to make updates (appending new data) to the outsourced database. The {\it static} setting~\cite{kellaris2016generic} on the other hand, allows no updates beyond the setup phase. We assume that each record from the logical database is atomically encrypted in the secure outsourced database. The outsourced database may, in addition, store some encrypted dummy records. This model is also referred to as {\it atomic database}~\cite{kellaris2016generic}. In addition, we assume that the physical updates can be different from the logical updates. For instance, an \user may receive a new record every 5 minutes, but may choose to synchronize once they received up to 10 records. 

%We assume the secure outsourced database is {\it atomic}~\cite{kellaris2016generic},  where the \cs stores a set of encrypted data and each record (or search key) in the logical database is encrypted as one of them, but there may be additional ciphertexts for dummy data.
%\kartik{SOLVED:Maybe: We assume that each record from the logical database is atomically encrypted in the secure outsourced database. The outsourced database may, in addition, store some encrypted dummy records.} 
%\johes{the previous sentence's wording is a bit unclear. may need more explanation of the encryption process, which is currently missing from the definition}

%Under the dynamic setting, if the logical database held by the \user is a growing database, that is all potential updates will be insertion only, then the model is referred to as the secure outsourced growing database (\system) model.
\eat{
The notation $(c\_{out}, s\_{out}, a\_{out}) \gets \prod(c\_{in},s\_{in}, a\_{in})$ is used to denote a protocol among an owner, an analyst and a server, where $<c\_in$, $c\_out>$, $<o\_in$, $o\_out>$ and $<a\_in$, $\_out>$, represents the input and output of \user, \cs, and \as, respectively.
}

%\nt{I deleted the descriptions about how protocols operates, as we will introduces the details in section 3, so remove the redundant descriptions here.}
\eat{
In the aforementioned model, we mostly followed the outsourced database description in \cite{kellaris2016generic} but made necessary extensions to capture features applied to growing databases. Given the aforementioned components, the \system system operates as follows: The \user picks a security parameter $\lambda$ and runs $\setup$ to setup a \system system followed by outsourcing the initial database $\mathcal{D}_0$. The \sync is initialized during the execution of $\setup$ and continues to run in the background as a ``daemon process''. When $\setup$ completes, the \user, \cs, and \as then perform their respective tasks asynchronously. For each subsequent time $t$, whenever a new record arrives, the \user stores it locally. If \sync signals for synchronization, then the \user follows \sync's instruction and runs $\update$ to update the outsourced database on the \cs, otherwise the \user does nothing and move on to the next timestamp. The \as, on the other hand, independently issues $\query$ protocol with a set of queries, denoted as $q_t$. The \cs evaluates those queries over the latest outsourced database and returns the corresponding results to the \as. All three participants repeat the steps over time until certain termination condition is triggered.

}

\subsection{Update Pattern Leakage}\label{sec:updat}
We now introduce a new type of volumetric leakage~\cite{blackstone2019revisiting} called {\it update pattern} leakage. %\kartik{SOLVED:generally, either use quotes or italics, not both. Please update everywhere. I use italics to emphasize and quotes when I want to say something but not formally define it}
In general, an update pattern consists of the \user's entire update history transcript for outsourcing a growing database. It may include information about the number of records outsourced and their insertion times.
\re{
\begin{definition}[Update Pattern]
Given a growing database $\mathcal{D}$ and a SOGDB scheme $\Sigma$, the update pattern of $\Sigma$ when outsourcing $\mathcal{D}$ is $\upatt(\Sigma,\mathcal{D}) = \{\upatt_t\left(\Sigma, \mathcal{D}_{t}\right)\}_{t\in\mathbb{N}^{+}\land t\in t'}$, with:
$$\upatt_{t}\left(\Sigma, \mathcal{D}_{t}\right) = \left(t, |\gamma_{t}|\right)$$
where $t'=\{t'_1,t'_2,... ,t'_k\}$ denotes the set of timestamps $t'_i$ when the update occurs, and $\gamma_t$ denotes the set of records synchronized to the outsourcing database at time $t$. We refer to the total number of records $|\gamma_{t}|$ updated at time $t$ as the corresponding update volume. 
\end{definition}
\begin{exmp} Assume an outsourced database setting where the \user synchronizes 5 records to the \cs every 30 minutes and the minimum time span is 1 minute. Then the corresponding update pattern can be written as $\{(0, 5), (30, 5), (60,5), (90,5)...\}$.
\end{exmp}
\eat{Any outsourced database that independently encrypts a single record into a separate ciphertext is at risk of leaking the update pattern. More generally, for any scheme, if its corresponding update leakage can be written as a function of update volume, then it leaks the update pattern.}}
%{\it fixed up-link overhead}.
%\kartik{SOLVED:why do we need a defn for fixed up-link overhead? we don't seem to use it later?} \nt{Can we say something about, the underlying edb's update leakage can be written as $\mathcal{L}'(|\gamma|)$, and $\mathcal{L}'$ is stateless.}
%Continue with Definition~\ref{def:sodb}, Thus for \system with fixed up-link volume, the server has access to the entire history of the corresponding up-link volume for all posted updates. Thus it's obvious that for all 

%To distinguish from the logical one, we use $\gamma_t$ to denote the real update at time $t$, which can be a set of records or null (no real update posted).
%where $\gamma_t$ denotes the set of records was updated to the \cs at time $t$, if no update happens then $\gamma_t = \emptyset$.

%\am{This section is super technical. I wonder if this level of technicality is needed before we give the high level overview and design of the system. Also, while it is precise, I am not getting an intuition for what the threat model is.}
%\subsection{\system with DP Update Pattern}
\vspace{-2mm}
\subsection{Privacy Model}
\label{sec:spmodel}
%\am{SOLVED:Should this section be called "privacy model"?}
Recall that in \appsystem, there are three parties: the owner (who outsources local data), the server (who stores outsourced data), and the analyst (who queries outsourced data). Our adversary is the server, whom we want to prevent from learning unauthorized information about individuals whose records are stored in the local data. We assume a semi-honest adversary, meaning that the server will faithfully follow all \appsystem protocols, but may attempt to learn information based on update pattern leakage. 

Update pattern leakage may reveal the number of records inserted at each time step, as the server can keep track of the insertion history. To ensure privacy, we need to strictly bound the information the server can learn. In this section, we formally define the privacy guarantee for update pattern leakage in \appsystem. %Intuitively, we want to provide a differential privacy bound on the update pattern leakage.
%We consider the security and privacy guarantee for an \user against an honest-but-curious \cs~\cite{goldreich2009foundations}. We assume the \as is trusted and is authorized by the \user. Intuitively, we want to provide a differential privacy bound on ``leakage'' by any single update.  
%To keep the definitions general, we define the security according to the seminal real-ideal paradigm~\cite{goldreich2009foundations}. We now describe the following experiments:
%\johes{Marked.differential privacy is not defined yet}.
%\kartik{SOLVED:say analyst is trusted?}
\begin{definition}[$\epsilon$-Differential Privacy~\cite{dwork2010differential}]\label{def:dp} A randomized mechanism $\mathcal{M}$ satisfies $\epsilon$-differential privacy (DP) if for any pair of neighboring databases $D$ and $D'$ that differ by adding or removing one record, and for any $O\subseteq\mathcal{O}$, where $\mathcal{O}$ is the set of all possible outputs, it satisfies:
$$\textup{Pr}\left[\mathcal{M}(D)\in O\right] \leq e^{\epsilon} \textup{Pr}\left[\mathcal{M}(D')\in O\right]$$
\end{definition}
With DP, we can provide provable, mathematical bounds on information leakage. This allows us to quantify the amount of privacy leaked to the server in our scheme. 
\re{
\begin{definition}[Neighboring growing databases]\label{def:ngdb}
$\mathcal{D}$ and $\mathcal{D'}$ are neighboring growing databases if for some parameter $\tau \geq 0$, the following holds: (i) $\mathcal{D}_t = \mathcal{D}'_t$ for $t \leq \tau$ and (ii) $\mathcal{D}_t$ and $\mathcal{D}'_t$ differ by the addition or removal of a single record when $t > \tau$.
%Given two growing databases $\mathcal{D}$ and $\mathcal{D'}$, for some time $\tau \geq 0$, such that, for all $0 \leq t \leq \tau$, $\mathcal{D}_t = \mathcal{D}'_t$ and for all $t > \tau$, $\mathcal{D}_t$ and $\mathcal{D}'_t$ are differ by at most one record. Then we call $\mathcal{D}$ and $\mathcal{D'}$ are neighboring growing databases. \kartik{SOLVED:Two growing databases $\mathcal{D}$ and $\mathcal{D'}$ are neighboring, if for some parameter $\tau \geq 0$, the following holds: (i) $\mathcal{D}_t = \mathcal{D}'_t$ for $t \leq \tau$ and (ii) $\mathcal{D}_t$ and $\mathcal{D}'_t$ differ by at most one record when $t > \tau$.}
\end{definition}
In practice, Definition~\ref{def:ngdb} defines a pair of growing databases that are identical at any time before $t=\tau$, and differ by at most one record at any time after $t=\tau$. After defining neighboring growing databases, we now follow the definition of event level DP~\cite{dwork2010differential} under continual observation, and generalize it to \system setting. This allows us to describe and bound the privacy loss due to update pattern leakage in \appsystem.
%Static DP setting requires that an algorithm produces similar results on any neighboring databases that are differ by at most one record. For growing database, to meet DP guarantees, algorithms are required to produce similar outputs on any two instance of growing databases $\mathcal{D}$ and $\mathcal{D'}$ satisfy that for some time $\tau \geq 0$, such that, for all $0 \leq t \leq \tau$, $\mathcal{D}_t = \mathcal{D}'_t$ and  for all $t > \tau$, $\mathcal{D}_t$ and $\mathcal{D}'_t$ are neighboring. Such pair of growing databases are called the neighboring growing databases.
\begin{definition}[\system with DP update pattern]\label{def:dpsogdb} Let $\lupdate$ be the update leakage profile for a \system system $\Sigma$. The \system $\Sigma$ has a differentially-private (DP) update pattern if $\lupdate$ can be written as:
$$\lupdate(\mathcal{D}) = \mathcal{L}'\left(\upatt(\Sigma, \mathcal{D})\right)$$
where $\mathcal{L}'$ is a function, and for any two neighboring growing databases $\mathcal{D}$ and $\mathcal{D}'$, and any $O\subseteq\mathcal{O}$, where $\mathcal{O}$ is the set of all possible update patterns, $\lupdate(\mathcal{D})$ satisfies:
$$\textup{Pr}\left[\lupdate(\mathcal{D})\in O \right] \leq e^{\epsilon}\cdot\textup{Pr}\left[\lupdate(\mathcal{D}') \in O\right]$$
\end{definition}
%\kartik{is defn 4 trying to talk about the entire system or only updates? If we never talk about leakage for setup and querying, should this talk about the security of $\Sigma$?}
%\nt{def 4 only says that the update pattern is in DP, but the rest leakage may not bounded in DP, such as query response volume, so it's actually only talks about the update pattern.}
%\kartik{what about setup? it seems you perturb and then invoke setup protocol?} \nt{Yes, the initial outsource also provide DP guarantee. Can discuss how to formulate it.}
Definition~\ref{def:dpsogdb} ensures that, for any \system,  if the update leakage is a function defined as $\upatt (\Sigma, \mathcal{D})$, then the information revealed by any single update is differentially private. Moreover, if each update corresponds to a different entity's (\user's) record then privacy is guaranteed for each entity. The semantics of this privacy guarantee are discussed further in Section~\ref{sec:semantics}.}
%\johes{SOLVED:what does entity here mean?} 
Note that although Definition~\ref{def:dpsogdb} provides information theoretic guarantees on update pattern leakage, the overall security guarantee for \appsystem depends on the security of the underlying encrypted database scheme. If the encrypted database provides information theoretic guarantees, then \appsystem also provides information theoretic DP guarantees. If the encrypted database is semantically secure, then \appsystem provides computational differential privacy, i.e., Definition~\ref{def:dpsogdb} only holds for a computationally bounded adversary.  
%We continue to provide a more generic security model in the Appendix.
\re{
\subsection{Privacy Semantics}
\label{sec:semantics}
In this section, we explore the privacy semantics of Definition~\ref{def:dpsogdb} from the perspective of disclosing secrets to adversaries. To achieve this, we utilize the Pufferfish~\cite{kifer2014pufferfish} framework to interpret the privacy semantics. One can show that if a \system satisfies Definition~\ref{def:dpsogdb}, then for any single user $u$, and any pair of mutually exclusive secrets of $u$'s record that span a single time step, say $\phi_u(t)$, and $\phi'_u(t)$ (an example of such pair of secrets is whether $u$'s data was inserted or not to an growing database), the adversary’s posterior odds of $\phi_u(t)$ being true rather than $\phi'_u(t)$ after seeing the \system's update pattern leakage is no larger than the adversary’s prior odds times $e^{\epsilon}$. Note that this strong privacy guarantee holds only under the assumption that the adversary is unaware of the possible correlation between the user's states across different time steps. Recent works~\cite{kifer2011no,liu2016dependence, xiao2015protecting} have pointed out that with knowledge of such correlations, adversaries can learn sensitive properties even from the outputs of differentially private algorithms. Nevertheless, it is still guaranteed that the ratio of the adversary's posterior odds to the prior odds is bounded by $e^{l\times\epsilon}$~\cite{cao2017quantifying, song2017pufferfish}, where $l$ is the maximum possible number of records in a growing database that corresponds to a single user. The actual privacy loss may be much smaller depending on the strength of the correlation known to the adversary ~\cite{cao2017quantifying, song2017pufferfish}. We emphasize that our algorithms are designed to satisfy Definition~\ref{def:dpsogdb} with parameter $\epsilon$, while simultaneously satisfying all the above privacy guarantees, though the privacy parameters may differ. Thus, for the remainder of the paper, we focus exclusively on developing algorithms that satisfy Definition~\ref{def:dpsogdb}. We continue to provide a more generic security model in the Appendix~\ref{sec:security-continued}.

}
\subsection {Evaluation Metrics} 
\label{sec:measure}
%\am{SOLVED:The reader does not yet have intuition for why these evauation metrics are needed. We have not told them yet that there are tradeoffs when you want to hide the update pattern -- either you will get some staleness in results, or there will a high performance overhead}\nt{We should tell the reader there is a 3-way trade-off of privacy, utility and performance, in the intro and mention back here}
\subsubsection{Efficiency metrics}
To evaluate \system's efficiency, we use two metrics: (1) query execution time (QET) or the time to run $\query$ and (2) the number of encrypted records  outsourced to the \cs. Note that in some cases the QET and the number of outsourced data may be positively correlated, as QET is essentially a linear combination of the amount of outsourced data.
%relative to \kartik{SOLVED:does ``relative to'' mean ``divided by''? If yes, may be clearer to say `ratio of x and y'} the number of records in the logical database
%\kartik{what does ``positively corrected'' mean?}

%These more specific metrics are essentially different linear combinations of the total amount of outsourced data. %Thus, we use it directly as the performance metrics for evaluating the \system system.

\eat{
\begin{definition}
Given $\Sigma$ = ($\sync$, $\setup$, $\update$, $\query$), a growing database $\mathcal{D} = \{\mathcal{D}_t\}_{0 \leq t \leq K}$ with length of $K$.  For a $\beta \in (0,1)$, we say the $\Sigma$ is $(\alpha, \beta)$ - outsourcing efficient relative to $\mathcal{D}$, if $~\textup{Pr}\left[|\mathcal{DS}_{T}| \geq \alpha \right] \leq \beta $. %\kartik{check if we can simplify this if $\beta$ is always negligible.}
\nt{I don't know if it's necessary to include this definition}
\end{definition}
}

\eat{
Communication efficiency consists of two parts: the uplink and downlik efficiency. Uplik efficiency stands for the number of encrypted records uploaded at each synchronisation, while downlink refers to the number of encrypted records sent back as the result of a query. \kartik{Is this different from storage efficiency because data stored on the server can be replaced? } \nt{The storage efficiency might be different from communication efficiency if we consider to support deletion, modification, etc. and yes, when we consider modification operations, the server-side storage may be replaced. But for now, if only consider "add" operation, communication efficiency = storage efficiency}

\begin{definition}
Given a \system $\Sigma=(Gen, Sync,$ $\prod_{setup},$ $\prod_{update},$ $\prod_{query})$, a dynamic database stream $\mathcal{D} = \{\mathcal{D}_t\}_{0 \leq t \leq T}$ with length $T$. Using $S_u$ and $S_d$ denotes the size of records batch in a single synchronization and the size of returned result of a single query. For a $\beta_1, \beta_2 \in (0,1)$, we say the $\Sigma$ is $(\alpha_1, \beta_1, \alpha_2, \beta_2)$ - communication efficient relative to $\mathcal{D}$, if $~\textup{Pr}\left[S_u \geq \alpha_1 \right] \leq \beta_1 $, and $~\textup{Pr}\left[S_d \geq \alpha_1 \right] \leq \beta_2 $. \kartik{do we care about the size of query result? I think we get true results, right?} \nt{Yes we get true encrypted results, i.e. counting query we will have a single encrypted count as the result}
\end{definition}
}
\subsubsection{Accuracy metrics}
%Ideally, the logical database should always be consistent with the outsourced database. That is, at any time, the outsourced database should contain at least \kartik{why is it at least and not exactly one?} one encrypted copy of each record in the logical database. \kartik{SOLVED:Possibly replace the previous two sentences with: Ideally, the outsourced database should contain all records from the logical database at every point in time.} 
Ideally, the outsourced database should contain all records from the logical database at every point in time. In practice, for efficiency and privacy reasons, an \user can only sync records intermittently. This temporary data inconsistency may result in some utility loss. To measure this utility loss, we propose two accuracy metrics as follows:
%In practice, however, an \user is more likely to store some records locally and sync them together to the \cs at a later point.\kartik{In practice, for efficiency and privacy reasons, an \user can sync records only intermittently.}
%\kartik{The temporary data inconsistency may result in some utility loss.}
%\begin{itemize}

{\bf Logical gap.} For each time $t$, the logical gap between the outsourced and logical database is defined as the total number of records that have been received by the \user but have not been outsourced to the \cs. We denote it as \smash{$LG(t) = \left|\mathcal{D}_t - \mathcal{D}_t \cap \hat{\mathcal{D}_t}\right|$}, where $\hat{\mathcal{D}_t} = \{\gamma_0 \cup \gamma_1\cup ...\gamma_t\}$ denotes the set of records that have been outsourced to the \cs until time $t$. Intuitively, a big logical gap may cause large errors on queries over the outsourced database. 
    
{\bf Query error.} For any query $q_t$, query error $QE(q_t)$ is the L1 norm between the true answer over the logical database and the result obtained from $\query$. Thus, $QE(q_t) = |\query(\mathcal{DS}_t, q_t) - q_t(\mathcal{D}_t)|$. While query error is usually caused by the logical gap, different types of query results may be affected differently by the same logical gap. Hence, we use query error as an independent accuracy metric. 
    %\johes{DP noise (as in crypte) may also affect query error right? }
%\end{itemize}
    
    %\kartik{SOLVED: can we somehow get by without using Dec here? I don't have a nice way to do it. Also, Dec is a little misleading since we said we encrypt each record separately.}

    \eat{
    Given a logical update value $u_t$ ($u_t \neq \emptyset$) arrives at time $t$. The record is then synchronized with the \cs at time $t'$. The difference between $t$ and $t'$ is called the unavailable window. Unavailable window denotes the time span that a given record is available in the logical database but not in the outsourced database. This measure can also be interpreted as if the \user stops receiving data after $t$, how long it takes for the outsourced and logical database to be converged to the same final state.\kartik{seems to be on a per record basis? do we want some aggregate too? also, does this not depend on subsequent records?}\nt{I changed this definition, I define it as the record gap between the logical and the outsourced database.}}

%\kartik{should we have some privacy metric too?}\nt{I think privacy metric is covered in the security model.}

\eat{Typically, for a outsourced database, the additional encryption of records or dummy instances won't the query results. Thus, the main factor affecting the accuracy of outsourced databases' queries is that users do not synchronize local data in time. To capture this, we define the error of a \system as the maximum possible size of unsynchronized records.  

\begin{definition}
Given a \system $\Sigma=(Gen, Sync,$ $\prod_{setup},$ $\prod_{update},$ $\prod_{query})$, a dynamic database stream $\mathcal{D} = \{\mathcal{D_t}\}_{0 \leq t \leq T}$ with length $T$. $S_c$ is the maximum possible size of unsynchronized records (under the synchronization strategy provided by $\Sigma$) hosted by a user $\user$. For $\beta \in (0,1)$, we say that $\Sigma$ is $(\alpha, \beta)-$accurate if, $~\textup{Pr}\left[S_c \geq \alpha \right] \leq \beta $
\end{definition}
}

\section{Record Synchronizing Algorithms}\label{sec:sync}
In this section, we discuss our secure synchronization strategies, including na\"ive methods (section ~\ref{sec:na}) and DP based strategies (section ~\ref{sec:dps}).  A comparison concerning their accuracy, performance, and privacy guarantees is provided in Table~\ref{tab:sync}.

\subsection{Na\"ive Synchronization Strategies}\label{sec:na}
We start with three na\"ive methods illustrated as follows:
\begin{enumerate}
    \item \textit{Synchronize upon receipt (SUR).} The SUR policy is the most adopted strategy in real-world applications, where the \user synchronizes new data to the \cs as soon as it is received, and remains inactive if no data is received.
    \item \textit{One time outsourcing (OTO).} The OTO strategy only allows the \user to synchronize once at initial stage $t = 0$. From then on, the \user is offline and no data is synchronized. % all data received after $t = 0$ will not be synchronized at all.
    \item \textit{Synchronize every time (SET).} The SET method requires the \user to synchronize at each time unit, independent of whether a new record is to be updated. More specifically, for any time $t$, if $u_{t} \neq \emptyset$, the \user updates the received record. If $u_{t} = \emptyset$, \user updates a dummy record to \cs. %SET guarantees that data is uploaded to the server as soon as it is received by \user. 
    %\johes{So at each timestep there is a max of one updated record?} 
    %\nt{yes, max of one }
    \end{enumerate}
%To better understand the na\"ive methods, we compare the 3 na\"ive strategies in terms of their respective privacy, accuracy, and performance guarantees.
Given a growing database $\mathcal{D} = \{D_0, U\}$. 
SUR ensures any newly received data is immediately updated into the outsourcing database, thus there is no logical gap at any time. Besides, SUR does not introduce dummy records. However, SUR provides zero privacy guarantee as it leaks the exact update pattern. OTO provides complete privacy guarantees for the update pattern but achieves zero utility for all records received by the \user after $t = 0$. Thus the logical gap for any time equals to $|\mathcal{D}_t| - |\mathcal{D}_0|$. Since OTO only outsources the initial records, the total amount of data outsourced by OTO is bounded by $O(|\mathcal{D}_0|)$. SET provides full utility and complete privacy for any record, and ensures 0 logical gap at any time. However, as a cost, SET outsources a large amount of dummy records, resulting in significant performance overhead. %\kartik{CHECK:opportunity to reduce text if needed.}
%Furthermore, the OTO does not satisfy the P2 principle, where the outsourced database cannot converge to the same state as the logical database after the \user stops receiving new data. 
\eat{
SET strategy, unlike OTO, provides full utility and complete privacy for any records. SET policy ensures the logical database and the outsourced database are consistent at every moment. However, as a cost, the \user has to synchronize even without receiving data. If the record arrival rate is relatively low (i.e., the time between two records arriving the \user is pretty long), then dummy records will take up the majority portion of the outsourced storage, resulting in significant performance overhead.\kartik{maybe add a small table/diagram showing what each strategy provides.}
}
In addition, all of the methods provide fixed privacy, performance, and/or utility. As such, none of them comply with the P3 design principle. OTO also violates P2 as no data is outsourced after initialization.

%\kartik{provides full privacy, worst performance, full utility}

%\kartik{to complement the first two, add another baseline that updates only when a relevant record arrives. Obtains no privacy, ideal performance, ideal accuracy.}

%\kartik{Until now, no cache was used. What is the best one can achieve with 0 or O(1) cache?}

\begin{table}[]
\scalebox{0.97}{
\begin{tabular}{r|c|c|c}
\toprule
\textbf{Group}    & \textbf{Privacy} & \textbf{Logical gap} & \textbf{\begin{tabular}[c]{@{}l@{}}Total number of\\ outsourced records\end{tabular}} \\ \midrule
\textbf{SUR}      & $\infty$-DP             & 0                                                              & $|\mathcal{D}_t|$                \\ 
\textbf{OTO}      & $0$-DP             & $|\mathcal{D}_t| - |\mathcal{D}_0|$ & $|\mathcal{D}_0|$         \\ 
\textbf{SET}      & $0$-DP             & 0                                                                  & $|\mathcal{D}_0| + t$            \\ 
\textbf{DP-Timer} & $\epsilon$-DP             &        $c_{t*}^{t} + O(\frac{2\sqrt{k}}{\epsilon})$                                                                              &  $|\mathcal{D}_t| + O(\frac{2\sqrt{k}}{\epsilon}) + \eta$                            \\ 
\textbf{ANT}      & $\epsilon$-DP             &    $c_{t*}^{t} + O(\frac{16\log{t}}{\epsilon})$                                                                                 &  $|\mathcal{D}_t| + O(\frac{16\log{t}}{\epsilon}) + \eta$                    \\ \bottomrule
\end{tabular}
}
\caption{Comparison of synchronization strategies. $c_{t*}^{t}$ counts the number of record received since last update, $k$ denotes the number of synchronization posted so far, $f$ is cache flush span,  $s$ is the cache flush size, and $\eta =s\floor*{{t}/{f}}$.}\vspace{-3mm}
%\kartik{$< |\mathcal{D}_0| + t$}
\label{tab:sync}
\end{table}

\subsection{Differentially Private Strategies}\label{sec:dps}
\eat{ Under the growing outsourcing database model, one of the indispensable conditions for ensuring data privacy is that $\cs$ should never learn exactly how many new records $\user$ received between two consecutive synchronizations. Thus we design a perturbation mechanism for the $\user$, i.e. every time when synchronize data, a perturbation is performed to hide the true data size. An intuitive method is to add dummy records every-time when synchronize. \cite{he2017composing} and \cite{bater2018shrinkwrap} provide a dummy record generation method based on truncated Laplace Mechanism which could eventually achieve $(\epsilon, \delta)-$DP. In this work, we invent an innovative data perturbation mechanism that would provide $\epsilon$-DP privacy guarantee. Algorithm \ref{algo:perturb} illustrates the perturb operation.

\begin{algorithm}[]
\caption{data perturbation function}
\begin{algorithmic}[1]
\Function{\ptb}{$c, \Delta, \epsilon, \texttt{cache} $}
    \State $\tilde{c} \gets c + \lap(\frac{\Delta}{\epsilon})$
    \If {$\tilde{c} \geq 0$}
    \State {\bf return} $\texttt{cache.read}(\tilde{c})$
    \Else
    \State {\bf return} $\emptyset$
\end{algorithmic}
\label{algo:perturb}
\end{algorithm}

}

\subsubsection{Timer-based synchronization (DP-Timer)} The timer-based synchronization method, parameterized by $T$ and $\epsilon$, performs an update every $T$ time units with a varying number of records. The detailed algorithm is described in Algorithm~\ref{algo:timer}. %\kartik{whenever you refer something use the following Algorithm~\ref{algo:timer}. The words Algorithm/Figure are always in full and start with a capital letter. Place a tilde between the word and the reference. This ensures that there is a space between them but they are not separated on different lines.}
%For ease of explanation, we introduces a new term called update event stream. For a given update stream $U = {u_1, u_2,...,u_t,...}$, the update event stream of $U$ is denoted as $E_U = {x_1, x_2, ..., x_t,...}$, where $\forall x_i \in E_U, x_i = 1$ if $u_i \neq \emptyset$, else $x_i = 0$. 

%As we mentioned before, the update scheduler $U$ is not necessarily to be exactly match $\user$'s update strategy $U'$. In this section, we discuss 5 update strategies for $\user$ to dynamically  synchronize it's local database with outsourced data on $\cs$:\\
%\am{What is the `noend` at the start of the each of the algorithm names?}
\begin{algorithm}[]
\caption{Timer Method (DP-Timer)}
\begin{algorithmic}[1]
\Statex \textbf{Input}: growing database $\mathcal{D} = \{\mathcal{D}_0, U\}$,  privacy budget $\epsilon$, timer $T$, and local cache $\sigma$.
\State $c\gets0$, $t^{*}\gets0$ %\kartik{Minor: consistency in notation. if $\gets$ used for assignment, what is $=$ used for?}
\State $\gamma_0 \gets \ptb(|\mathcal{D}_0|, \epsilon, \sigma) $
\State $\mathsf{Signal}$ the \user to run $\setup(\gamma_0)$.  %\kartik{should this be Pi-setup? Should it take gamma0 as input? what is $d$?}
%\kartik{does nto type check. $\prod_{update}$ takes many parameters.}
\For{$t \gets 1, 2, 3, ...$}
\If {$u_t \neq \emptyset$}
%\State $c \gets c + 1$
\State $\writec(\sigma, u_t)$ (store $u_t$ in the local cache) %\kartik{can use alg environment that does not print 'end if'. will save a line.}
\EndIf
\If{$t \mod T = 0$}
\State $c \gets \sum_{i\gets t - T + 1}^{t} x_i$ $ ~|~(x_i \gets 0$, if $u_i = \emptyset$, else $x_i \gets 1)$
\State $\gamma_t \gets \ptb(c, \epsilon, \sigma) $
\State $\mathsf{Signal}$ the \user to run $\prod_{\text{update}}(\gamma_t, \mathcal{DS}_t)$.
%\State $t^{*} \gets t$ \kartik{perhaps $t^*$ is not needed. Can check from t-T+1 to t}
\EndIf
%\If {$sc \geq K$}
%\State $d\gets\lcache.\texttt{flush(sz)}$, {\bf run} $\prod_{update}\left( d \right)$ \kartik{I don't see the diff between flush, write, reset. Why don't we just perform a write here?}
%\State $sc \gets 0$
%\EndIf
\EndFor
\end{algorithmic}
\label{algo:timer}
\end{algorithm}

\re{Initially, we assume the \user stores $\mathcal{D}_0$ in the local cache $\sigma$. DP-Timer first outsources a set of data $\gamma_0$ to the \cs (Alg~\ref{algo:timer}:1-3), where $\gamma_0$ is fetched from $\sigma$ using \ptb (defined in Algorithm~\ref{algo:perturb}) operator. \ptb takes as input a count $c$, a privacy parameter $\epsilon$ and a local cache $\sigma$ to be fetched from. It first perturbs the count $c$ with Laplace noise $\lap(\frac{1}{\epsilon})$, and then fetches as many records as defined by the noisy count from $\sigma$. When there is insufficient data in the local cache, dummy data is added to reach the noisy count. After the initial outsourcing, the \user stores all the received data in the local cache $\sigma$ (Alg~\ref{algo:timer}:5-7), and DP-Timer will signals for synchronization every $T$ time steps. Whenever a synchronization is posted, the \user counts how many new records have been received since the last update, inputs it to the \ptb operator, and fetches $\gamma_t$. The fetched data is then synchronized to the \cs via the $\update$ protocol (Alg~\ref{algo:timer}:8-11).
The logic behind this algorithm is to provide a synchronization strategy with a fixed time schedule but with noisy record counts at each sync. The DP-Timer method strictly follow the policy of updating once every $T$ moments, but it does not synchronize exactly as much data as it receives between every two syncs. Instead, it may synchronize with additional dummy data, or defer some data for future synchronization.}
\eat{According to Algorithm~\ref{algo:timer}, whenever a synchronization is posted, the \user counts how many new records have been received since the last synchronization, perturbs this count with Laplace noise, and fetches from the local cache as many records as specified by the noisy count. When there is insufficient data in the local cache, dummy data is added to reach the noisy count. We refer the fetching data operation with the noisy size as the \ptb operation. Algorithm~\ref{algo:perturb} illustrates the details of \ptb operation. After \ptb returns the fetched data, the \user runs $\update$ to synchronize with the \cs. } 
\eat{Then the \user fetches records from the local cache by calling \ptb$(c, \epsilon, \sigma)$. What \ptb does is to add a Laplace noise \lap($\frac{1}{\epsilon}$) to the number of records that the \user needs to read , and then fetch the cached data with the perturbed size.} 
\vspace{2mm}
\begin{algorithm}[h]
\caption{Perturbed Record Fetch}
\label{algo:perturb}
\begin{algorithmic}[1]
\Function{\ptb}{$c, \epsilon, \sigma $}
    \State $\tilde{c} \gets c + \lap(\frac{1}{\epsilon})$
    \If {$\tilde{c} > 0$}
    \State {\bf return} $\readc(\sigma, \tilde{c})$ (read records with noisy size)
    \Else
    \State {\bf return} $\emptyset$ (return nothing if $\tilde{c} \leq 0$)%\kartik{there seems to be a compilation concern here.}
    \EndIf
\EndFunction
\end{algorithmic}
\end{algorithm}
\begin{theorem}\label{lg:timer}
Given privacy budget $\epsilon$, and \smash{$k\geq 4\log{\frac{1}{\beta}}$} where $k$ denotes the number of times the \user has synchronized so far, $\beta \in (0,1)$, and \smash{$\alpha = \frac{2}{\epsilon}\sqrt{k\log{\frac{1}{\beta}}}$. This satisfies $\textup{Pr}\left[LG(t) \geq \alpha + c_{t*}^{t}\right] \leq \beta$}, where $LG(t)$ is the logical gap at time $t$ under DP-Timer method, and $c_{t*}^{t}$ counts how many records received since last update. %\kartik{not entirely sure, but writing this as $\Pr[LG \geq X] < \beta$ may be helpful. On first read, it felt like the logical gap is always large.}
\end{theorem}

Theorem~\ref{lg:timer} provides an upper bound on the logical gap incurred by DP-Timer, due to space concerns we defer the proof in the Appendix~\ref{sec:proofs}. As a direct corollary of Theorem~\ref{lg:timer}, the logical gap is always bounded by $O({2\sqrt{k}}/{\epsilon})$. Knowing that, the logical gap can also be used to denote the total records that are on-hold by the \user, thus we can conclude that the local cache size of DP-Timer is also bounded by $O({2\sqrt{k}}/{\epsilon})$.
%\kartik{we may want to show this theoretical curve on the experiment figure.} 
However, if we consider an indefinitely growing database, then the 
local cache size (logical gap) grows indefinitely. Thus to prevent the local cache (logical gap) from getting too large, we employ a cache flush mechanism which refreshes the local cache periodically. The cache flush mechanism flushes a fixed size data with a fixed interval (usually far greater than $T$). The flushed data will be synchronized to the \cs immediately. If there is less data than the flush size, the mechanism empties the cache, and synchronizes with additional dummy records. This further guarantees every time when flush is triggered, it always incurs a fixed update volume. Moreover, Theorem~\ref{lg:timer} also reveals that it is possible to get a bounded local cache size. For example, if we set the flush size $s=C$, and the flush interval $f < {T(\epsilon C)^2}/{4\log({1}/{\beta})}$, where $C > 0, C \in \mathbb{Z}^{+}$. Then at any time $t$, with probability at least $1-\beta$, the cache size is bounded by $O(C)$. Next, we discuss the performance overhead with respect to the DP-Timer.
%\kartik{Is this cache flush evaluated?}

\begin{theorem}\label{sz:timer}
Given privacy budget $\epsilon$, flush interval $f$, flush size $s$, and $\beta \in (0,1)$. Let \smash{$\alpha = \frac{2}{\epsilon}\sqrt{k\log{\frac{1}{\beta}}}$}, and $\eta = s\floor*{{t}/{f}}$. Then for any \smash{$t > 4T\log{(\frac{1}{\beta})}$}, the total number of records outsourced under DP-Timer, $|\mathcal{DS}_t|$, satisfies $\textup{Pr}\left[|\mathcal{DS}_t| \geq |\mathcal{D}_t| + \alpha + \eta\right]\leq \beta$.
\end{theorem}

Theorem~\ref{sz:timer} provides an upper bound for the outsourced data size at each time $t$. Moreover, it shows that the total amount of dummy data incorporated is bounded by $\eta + O(2\sqrt{k}/\epsilon)$.
%We show that DP-Timer satisfies definition~\ref{def:dpsogdb} in section~\ref{sec:sp} (P1 principle).
Due to the existence of the cache flush mechanism, DP-Timer guarantees that for a logical database with length $L$, all records will be synchronized before time $t = f\times L/s$. Recall that a FIFO based local cache preserves the order of incoming data, thus DP-Timer satisfies the strong eventually consistency property (P3). In addition, as shown by Theorem~\ref{lg:timer} and~\ref{sz:timer}, both accuracy and performance metrics are related to $\frac{1}{\epsilon}$, which shows that DP-Timer satisfies the P2 principle. %as by changing $\epsilon$ the strategy could provide different accuracy and performance guarantees. 

\eat{
For a \system implemented with DP-Timer synchronization, the corresponding leakage profiles are summarized as follows: The setup protocol leaks the number of records transmitted at the initial phase. The query protocol leaks the size of encrypted results for each queries. And the update protocol leak reveals the entire update pattern (history), which consists of a series of synchronization timestamps and the number of records transmitted during each synchronization. We theoretically analyze the security guarantees in Section~\ref{sec:sp}, that proves DP-Timer satisfy Definition~\ref{def:dpsogdb}.
}

\eat{Similar to SUR and SET method, the DP-Timer strategy satisfy the consistent eventually principle (P2). Besides, the DP-Timer method provides a privacy parameter $\epsilon$, that can be used the the users to adjust the privacy level it guarantees, which compliant with P3 principle.}

%the aforementioned leakage profiles for DP-Timer method are bounded in a differentially private sense. %\cw{added the leakage profiles for DP-Timer method}

\eat{
\texttt{Sync} flush the local cache with a certain frequency, currently, it is defined as flushing the cache after every $K$ times of synchronizations.

\kartik{you should explain the intuitive advantage and potential disadvantage of this method. for privacy, this seems worse than the other two. for performance and accuracy, you should intuitively claim why this is better.}

 we introduce the notion of a perturbed record fetch \kartik{this seems like the size of an update. Why do we call it a fetch?} 
 }
 
%Under the growing outsourcing database model, one of the indispensable conditions for ensuring data privacy is that $\cs$ should never learn exactly how many new records $\user$ received between two consecutive synchronizations. Thus we design a perturbation mechanism for the $\user$, i.e. every time when synchronize data, a perturbation is performed to hide the true data size. 
%An intuitive method is to add dummy records every-time when synchronize. \cite{he2017composing} and \cite{bater2018shrinkwrap} provide a dummy record generation method based on truncated Laplace Mechanism which could eventually achieve $(\epsilon, \delta)-$DP. In this work, we invent an innovative data perturbation mechanism that would provide $\epsilon$-DP privacy guarantee. 

\subsubsection{Above noisy threshold (DP-ANT)} \re{The Above noisy threshold method, parameterized by $\theta$ and $\epsilon$, performs an update when the \user receives approximately $\theta$ records. The detailed algorithm is described in Algorithm~\ref{algo:cbuffer}.

Similar to DP-Timer, DP-ANT starts with an initial outsourcing (Alg~\ref{algo:cbuffer}:1-2) and the \user then stores all newly arrived records in the local cache $\sigma$ (Alg~\ref{algo:cbuffer}:6-9). After the initial outsourcing, DP-ANT splits the privacy budget to two parts $\epsilon_1$, and $\epsilon_2$, where $\epsilon_1$ is used to distort the threshold as well as the counts of records received between two updates, and $\epsilon_2$ is used to fetch data. The \user keeps track of how many new records received since the last update at every time step, distorts it with DP noise, and compares the noisy count to a noisy threshold (Alg~\ref{algo:cbuffer}:10,11). The \user synchronizes if the noisy count exceeds the noisy threshold. After each synchronization, the user resets the noise threshold with fresh DP noise (Alg~\ref{algo:cbuffer}:14) and repeats the aforementioned steps. }

\eat{Under the DP-ANT strategy, the \user keeps track of the count of new records received since the last update, distorts it with a DP noise, and compares the noisy count to a noisy threshold. The \user synchronizes if the noisy count exceeds the noisy threshold. After each synchronization, the user resets the noise threshold with fresh DP noise and repeats the aforementioned steps. Algorithm~\ref{algo:cbuffer} describes this process in detail. 

After the initial outsourcing, the strategy splits the privacy budget to $\epsilon_1$, and $\epsilon_2$, where $\epsilon_1$ is used to distort the threshold as well as the counts of records received between two updates, and $\epsilon_2$ is used to make the perturbed fetch.}

\begin{algorithm}[]
\caption{Above Noisy Threshold (ANT)}
\begin{algorithmic}[1]
\Statex
\textbf{Input}: growing database $\mathcal{D} = \{\mathcal{D}_0, U\}$, privacy budget $\epsilon$, threshold $\theta$, and the local cache $\sigma$.
%\Statex
\State $\gamma_0 \gets \ptb(|\mathcal{D}_0|, \epsilon, \sigma) $
\State $\mathsf{Signal}$ the \user to run $\setup(\gamma_0)$.%\kartik{Pi-setup?}
\State $\epsilon_1 \gets \frac{1}{2}\epsilon, \epsilon_2 \gets \frac{1}{2}\epsilon$
\State $ \tilde{\theta} \leftarrow \theta + \lap(2/\epsilon_1)$, $c \leftarrow 0$, $t^{*} \gets 0$
\For{$t \leftarrow 1,2,...$} %\kartik{consistency with $\gets$, $=$}
\State $v_t \leftarrow \lap(4/\epsilon_1)$
\If{$u_t \neq \emptyset$}
\State store $u_t$ in the local cache, $\writec(\sigma, u_t)$ 
\EndIf
\State $c \gets \sum_{i\gets t^{*} + 1}^{t} x_i$ $ ~|~(x_i \gets 0$, if $u_i \gets \emptyset$, else $x_i \gets 1)$
\If{$c + v_t \geq \tilde{\theta}$}
\State $\gamma_t \gets \ptb(c, \epsilon_2, \sigma) $%\kartik{perturb takes three parameters. remove 1?}
\State $\mathsf{Signal}$ the \user to run $\update \left( \gamma_t, \mathcal{DS}_{t} \right)$
\State $\tilde{\theta} \leftarrow \theta + \lap(2/\epsilon_1), c \gets 0$, $t^{*} \gets t$
\EndIf
\EndFor
\end{algorithmic}
\label{algo:cbuffer}
\end{algorithm}

\eat{Unlike the DP-Timer algorithm, which has a fixed synchronization frequency, DP-ANT dynamically adjusts its synchronization frequency according to the rate at which new records are received.} %\kartik{the statement is phrased as if this is an advantage compared to dp-timer? is that so? independently, we should perhaps bring this up in the evaluation?} \nt{I think it's not an advantage, it's just a different sync pattern, as Timer has fixed sync scheduler but ANT does not.} 
\re{DP-ANT synchronizes based on how much data the \user receives. However, it does not simply set a fixed threshold for the \user to synchronize whenever the amount of data received exceeds that threshold. Instead, it utilizes a strategy that allows the \user to synchronize when the amount of received data is approximately equal to the threshold. Below, we analyze DP-ANT's accuracy and performance guarantees.}

%Thus DP-ANT can dynamically adjust its synchronization frequency according to the rate at which new records are received.

%Same as DP-Timer, the DP-ANT also satisfies P1 (see security proofs in section~\ref{sec:sp}) principle.
\begin{theorem}\label{lg:ant}
Given privacy budget $\epsilon$ and let \smash{$\alpha = \frac{16(\log{t} + \log{{2}/{\beta}})}{\epsilon}$}. Then for $\beta \in (0,1)$, it satisfies $\textup{Pr}\left[LG(t) \geq \alpha + c_{t*}^{t}\right] \leq \beta$, where $LG(t)$ is the logical gap at time $t$ under DP-ANT method, and $c_{t*}^{t}$ counts how many records received since last update.
%Given the privacy budget $\epsilon$ and $k$. For any time $t$, $c_{t*}^{t}$ denotes the number of records received since last update, then with probability at most $\beta$, the logical gap of DP-Timer satisfies
%$$~~\textup{Pr}\left[ LG(t) \geq c_{t*}^{t} +  \frac{16(\log{t} + \log{{4}/{\beta}})}{\epsilon} ~\right] \leq {\beta}$$
\end{theorem}
The above theorem provides an upper bound for DP-ANT's logical gap as well as its local cache size, which is $c_{t*}^{t} + O({16\log{t}}/{\epsilon})$. Similar to DP-Timer, we employ a cache flush mechanism to avoid the cache size grows too large. We use the following theorem to describe DP-ANT's performance:
\begin{theorem}\label{sz:ant}
Given privacy budget $\epsilon$, flush interval $f$, flush size $s$, and $\beta \in (0,1)$. Let \smash{$\alpha = \frac{16(\log{t} + \log{{2}/{\beta}})}{\epsilon}$}, and $\eta = s\floor*{{t}/{f}}$. Then for any time $t$, it satisfies $\textup{Pr}\left[|\mathcal{DS}_t| \geq |\mathcal{D}_t| + \alpha + \eta\right]\leq \beta$, where $|\mathcal{DS}_t|$ denotes the total number of records outsourced until time $t$.
%Given the privacy budget $\epsilon$, flush interval $f$, flush size $s$, and $\beta \in (0,1)$. The for any time $t$, the outsourced data size under DP-ANT satisfies:
%$$~~\textup{Pr}\left[ |\mathcal{DS}_t| \geq |\mathcal{D}_t| + s\floor*{\frac{t}{f}} +  \frac{16(\log{t} + \log{{4}/{\beta}})}{\epsilon} ~\right] \leq {\beta}$$
\end{theorem}

 This theorem shows that the total overhead of DP-ANT at each time $t$ is bounded by $s\floor*{{t}/{f}} + O(16\log{t}/\epsilon)$. Note that both the upper bound for the logical gap and the performance overhead is related to $1/\epsilon$, which indicates a trade-off between privacy and the accuracy or performance. With different values of $\epsilon$, DP-ANT achieves different level of accuracy and performance (P2 principle). And the FIFO cache as well as the flush mechanism ensures the consistent eventually principle (P3). We provide the related proofs of Theorem~\ref{lg:ant} and~\ref{sz:ant} in the Appendix~\ref{sec:proofs}. \re{ Later in Section~\ref{sec:exp} we further evaluate how different parameters would affect the accuracy and performance of DP strategies, where readers can better understand how to set these parameters according to the desired~goals.}

\eat{ {\bf (5) Cascading Buffer with Sliding Window:} It is same as cascading buffer method but added a sliding window. Specifically, given a window size say $\tau$, then at each time $t$ the user keep tracks of the counts of new records arrived within time period $\left[ \max(t', t - \tau), ~t \right]$, where $t'$ is the time stamp for the most recent sychronization. 

{\bf (6) Cascading Buffer with Decayed-Sum Counter:} At each time t, the decayed sum is given by $F(x_{j}, ..., x_{t}) = \sum_{i = j }^{t} x_{i}g(t - i)$. I.e. F is the convolution
of the input and a non-increasing function g.
}

\eat{
\subsection{Privacy with Expiration}

\begin{definition}[Neighboring Database Stream]
 We define that two time based data stream $X(T) = (x_1, x_2,...,x_T)$ and $X'(T) = (x'_1, x'_2,...,x'_T)$ is considered to be neighboring(denoted as $X \sim X'$) if there exists some time $j \geq 0$, that $x_j \in X \neq x'_j \in X'$, and for all other time stamps $x_j = x'_j$.
\end{definition}

\begin{definition}[$\epsilon$-deferentially private \system] Given a randomized algorithm $\mathcal{A}$, and a non-decreasing expiration function $g$. We say $\mathcal{A}$ is $\epsilon$-differential privacy under expiration $g$ if for all possible pair of neighboring data stream $X(T)$ and $S'(T)$ differing by one element ($j^{th}$ element) and for all possible $S$ it satisfy:

$$\textup{Pr}[\mathcal{A}(X(T))] \leq e^{g(T - j)\epsilon} \textup{Pr}[\mathcal{A}(X'(T))]$$
\end{definition}

An example of $g(i)$ is $\forall i \leq W, g(i)=1$, and $\forall i > W, g(i) = \infty$, this notation says that, starting from time $T$, any records that old enough (older than a certain window $W$), then the privacy of these data are expired. 

}

\eat{
\subsubsection{Support for decayed queries} 
Both DP-Timer and ANT ensure that all data is uploaded in the order in which it was received by the \user. In some cases, however, the \as may only be interested in the data that been received in the ``last $W$ time units''. Or more generally, they discount the interest of records based on how far they are in the past, then query over the decayed data. Under those circumstances, synchronization strictly in the order in which the data is received is not necessary. In stead, the more recent the data is received, the higher priority should be given to synchronize it. To handle this special case, we propose modifications on DP-Timer and ANT in two phases:
\begin{enumerate}
    \item The read operation of the local cache will be set to LIFO mode. That is, when executing read operation, and there is more data in the local cache than that is need, the cache read will fetch out the least stored record first. 
    \item We changed the method of counting received records in DP-Timer (Algorithm~\ref{algo:timer}:9) and ANT (Algorithm~\ref{algo:cbuffer}:10) from general summation to decayed sum. Where the decayed sum is defined as $c(x_{i}, ..., x_{t}) = \sum_{i = t^{*} + 1 }^{t} x_{i}w(t - i)$, where $t^{*}$ denotes the time for last update, $x_i = 1$, if $u_i \neq \emptyset$, otherwise $x_i = 0$. And $w:\mathbb{N}\rightarrow\mathbb{R}^{+}$ is a non-increasing function. 
\end{enumerate}

With these changes, the modified DP-Timer and ANT can be used to support the decay queries. As an example, let's consider covid19 use cases, researchers will most likely only use data from the last 30 days and are more interested in data for the most recent 14 days. Then we follow the instruction for modifying the synchronization policy and set $w$ as follows:
\begin{equation}
    w(x) =
    \begin{cases}
      1, & \text{if}\ x \leq 14 \\
      0.25, & \text{if}\ 14 < x \leq 30\\
      0, & \text{otherwise}
    \end{cases}
\end{equation}
}
\section{Connecting with Existing EDBs}\label{sec:cmp}
Interoperability of \appsystem with an existing encrypted database is an important requirement (P4 design principle). In this section, we discuss how to connect existing encrypted databases with \appsystem. Since our privacy model constrains the update leakage of the encrypted database to be a function only related to the update pattern, in this section we mainly focus on query leakage associated with the encrypted database to discuss the compatibility of our framework. Inspired by the leakage levels defined in~\cite{cash2015leakage}, we categorize different encrypted database schemes based on our own leakage classification. Then we discuss which schemes under those categories can be directly connected with \appsystem and which databases need additional improvements to be compatible with our framework. In Table~\ref{tab:sumenc}, we summarize some notable examples of encrypted databases with their respective leakage groups. We focus on two types of leakage patterns: {\it access pattern}~\cite{goldreich2009foundations} and {\it query response volume}~\cite{kellaris2016generic}. The access pattern is the transcript of entire memory access sequence for processing a given query, and query response volume refers to the total number encrypted records that matches with a given query. The four leakage categories are~as~follows: %\kartik{somewhere we should mention this is related to P4.} \nt{Note to myself, check after other changes applied}

%As our framework will ask the clients to upload dummy records to the outsourced database, thus we need to ensure that the underlying outsourced database scheme is compatible with the presence of dummy data and does not disclose the amount of inserted dummy records. Although it’s not hard to design a dummy data type that is indistinguishable from real data and does not affect any query results, the fact that many outsourced secure databases that leaks certain type of information that can be exploit by an attacker to infer the amount of dummy records. 

\begin{table}[]
\begin{tabular}{ r | l}
\toprule
\textbf{Leakage groups} & \textbf{Encrypted database scheme}  \\ \midrule
                                     & VLH/AVLH~\cite{kamara2019computationally}, ObliDB~\cite{eskandarian2017oblidb}, SEAL~\cite{demertzis2020seal}                                                                             \\
\textbf{L-0}                                  & Opaque~\cite{zheng2017opaque}, CSAGR19~\cite{cui2019privacy}
                                        \\ \midrule
                                     & dp-MM~\cite{patel2019mitigating}, Hermetic~\cite{xuhermetic}, KKNO17~\cite{kellaris2017accessing}                                                                                        \\
\textbf{L-DP}                                 & Crypt$\epsilon$~\cite{chowdhury2019cryptc},     AHKM19~\cite{agarwal2019encrypted}, Shrinkwrap~\cite{bater2018shrinkwrap}                                    
                                     \\\midrule

\textbf{L-1}                                      & PPQED$_{a}$~\cite{samanthula2014privacy}, StealthDB~\cite{vinayagamurthy2019stealthdb}, SisoSPIR~\cite{ishai2016private}                                                                                             \\ \midrule
                                     & CryptDB~\cite{popa2012cryptdb}, Cipherbase~\cite{arasu2013orthogonal}, Arx~\cite{poddar2016arx} \\                                                 
\textbf{L-2}                                  & HardIDX~\cite{fuhry2017hardidx}, EnclaveDB~\cite{priebe2018enclavedb}                                                                                             \\ \bottomrule
\end{tabular}
\caption{Summary of leakage groups and corresponding encrypted database schemes}%\kartik{SOLVED:vertically align left column to be at the center, also horizontally at the center or aligned right} %\kartik{SOLVED:should we call it L-DP instead of L-DP?}
\vspace{-3mm}
\label{tab:sumenc}
\end{table}
\vspace{-2mm}
\boldparagraph{L-2: Reveal access pattern}. Encrypted databases that reveal the exact sequence of memory accesses and response volumes when processing queries fall into this category. These include many practical systems based only on searchable symmetric encryption, trusted execution environments (TEE), or on deterministic and order-preserving encryption. Recent leakage-abuse
attacks~\cite{cash2015leakage, blackstone2019revisiting, markatou2019full} have pointed out that attackers can exploit the access pattern to reconstruct the entire encrypted database. Databases in this category are not compatible with \appsystem. If we add our techniques to these systems, then due to the leakage from these databases, our update patterns will be leaked as well.
%\kartik{pose the next few sentences as: databases in this category are not compatible. If we combine our techniques with them, then due to the leakage from these databases, our update pattern will be leaked as well.} For databases with L-2 leakage, there is no direct compatibility with our framework. 
%To compatible with our framework, schemes in L-2 category need to make the necessary adjustments to mitigate the risks posed by the access pattern.

\vspace{-3mm}
\boldparagraph{L-1: Reveal response volume}. To hide access patterns, some schemes perform computations obliviously, e.g., using an oblivious RAM. However, many databases in this category still leak the query response volume (since obliviousness does not protect the size of the access pattern). Example databases in this category include HE-based PPQED$_a$~\cite{samanthula2014privacy} and ORAM-based SisoSPIR~\cite{ishai2016private}. Moreover, recent research~\cite{markatou2019full, poddar2020practical, kellaris2016generic, grubbs2018pump, lacharite2018improved} has shown that database reconstruction attacks are possible even if the system only leaks response volume. Therefore, there is still a risk that such systems will leak information about the amount of dummy data. Thus, to be compatible with \appsystem, necessary measures must be taken to hide the query volume information, such as na\"ive padding~\cite{cui2019privacy}, pseudorandom transformation~\cite{kamara2019computationally}, etc.

%To hide the access pattern, some schemes incorporates Oblivious RAM (ORAM) or combines Homomorphic Encryption with oblivious data structure. \kartik{To hide access patterns, some schemes perform computations obliviously, e.g., using an oblivious RAM.} \kartik{the use of homomorphic encryption is orthogonal to access pattern. One could use trusted h/w or secure computation too.} However, outsourced databases employ heavyweight cryptographic techniques that hide the access patterns can still leak the query response volume.\kartik{However, many databases in this category still leak the query response volume (since obliviousness does not protect the size of the access pattern).} \kartik{heavyweight cryptography is orthogonal.} 
%Thus any T1 type of dummy records added to such database system can be identified by an adversary. The risk may be mitigated if T2 type dummy data is used, as Type-2 dummy can distort the true response volume to some extent. \\
\vspace{-2mm}
\boldparagraph{L-DP: Reveal differentially-private response volume}. Some secure outsourced database schemes guarantee the leakage of only differentially-private volume information. These schemes either ensure that both access patterns and query volumes are protected using differential privacy, or they completely hide the access patterns and distort the query response volume with differential privacy. Databases with L-DP leakage are directly compatible with \appsystem, as such schemes prevents attackers from inferring information about dummy data through the query protocol.

\vspace{-2mm}
\boldparagraph{L-0: Response volume hiding}. Some encrypted databases support oblivious query processing and only leak computationally-secure response volume information. These schemes are usually referred to as access pattern and volume hiding schemes. Encrypted databases in this category can be directly used with our framework as well, as there is no efficient way for attackers to identify dummy data information via their query protocols. 

In addition, most methods that fall in L-DP and L-0 category support dummy data by default~\cite{patel2019mitigating, xuhermetic, kellaris2017accessing, eskandarian2017oblidb}, as they use dummy data to populate the query response volume or hide intermediate sizes. In this case, our framework can directly inherit the dummy data types defined in the corresponding database scheme with no additional changes. For those schemes that do not support dummy data by default (e.g.~\cite{chowdhury2019cryptc}), we can either let the scheme return both dummy and real data, and let the analyst to filter true records after decryption, or we can extend all records with a {\it isDummy} attribute and then apply query re-writing to eliminate the effect of dummy data.%\kartik{does there exist any such database? anyone with L-Dp and L-0 need to have some support for dummies right?}  
 We continue to provide query re-writing examples in our full version. To concretely demonstrate the compatibility of \appsystem with existing encrypted databases, we choose database schemes ObliDB\cite{eskandarian2017oblidb} and Crypt$\epsilon$\cite{chowdhury2019cryptc} in L-0 and L-DP groups respectively and evaluate the resulting implementation in Section~\ref{sec:exp}.

%\kartik{should mention somewhere in this section that we will evaluate one of L-0 and L-DP.}

%\kartik{note to myself: revisit this section after all changes are performed.}
\section{Security Proofs}\label{sec:sp}
In this section, we provide a sketch of the security proof for our proposed \appsystem implemented with DP strategies. \eat{Recall that we only use databases where: (1) The insertion operation is supported;  (2) The corresponding update leakage leaks nothing more than the update volume; (3) Query leakages are bounded in L-0 and L-DP groups.}
\eat{
As mentioned before, protocols leaks particular information to the \cs. In detail, we define the leakage profiles as follows:

\begin{itemize}
    \item $\lsetup(\mathcal{DS}_0) = \encs\left(\mathcal{DS}_0\right)$
    \item $\lupdate(t) = \encs\left(\mathcal{DS}_t\right) - \encs(\mathcal{DS}_{t-1})$
    \item $\lquery(q_t, t) = \{ \encs({\bf v}_t^j) \mid j \gets 1,2,...\}$
\end{itemize}

where $\encs$ is the operation that computes the number of encrypted records. Next we interpret these leakage profiles in more understandable semantics. The \cs learns how many encrypted records submitted at initial stage through $\lsetup$. By continuously observing $lupdate$ over time, the \cs learns the entire synchronization history, that consists the exact time when each synchronization occurs and the number of encrypted records submitted at each update. Throughout the $\lquery$, the \cs learns how many encrypted results have been sent back to the \as at each time, where $q_t = \{q_t^1, q_t^2,...\}$, and ${\bf v}_t^j$ is the encrypted answer of $q_t^j \in q_t$. Now we are ready to proof the following theorem:
}

%\kartik{Is it fair to say $\mathcal{L}_{u}^{\mathsf{edb}} = |\gamma|$? The leakage due to update is leakage due to update in edb + leakage due to our updates?} \nt{I think we need to make some restriction on the update leakage of the underlying edb. Otherwise if the update itself leaks access pattern, then the server may identify dummy records from update protocol. We may say that $\mathcal{L}_{u}^{\mathsf{edb}} = \mathsf{L}'(|\gamma|)$, where $\mathsf{L}'$ is stateless, the stateless means given a set of constant parameters we can construct $\mathcal{L}_{u}^{\mathsf{edb}}$ with $|\gamma|$. And this means that the underlying edb may have other leakages with their update protocol, but can be reduced to $|\gamma|$. Actually in many SSE scheme (such as ~\cite{ghareh2018new}) they require that $\lupdate(w, \mathsf{op})$ should be equal to $\mathcal{L}'(|w|, \mathsf{op})$, where $|w|$ measures the number of documents uploaded. If the scheme reveals which keys each updated record is associated with (access pattern), then our technique fails. In addition, in our case the operation is always "insert" so we won't worry about $\mathsf{op}$}
%\kartik{Okay. In that case, before the theorem you should mention that we are constrained to using databases with such constraints and may be link it back to the previous section (L-0 and L-DP).}

%In what follows, we demonstrate that the \appsystem implemented with either DP-Timer or ANT strategy satisfies definition~\ref{def:dpsogdb}.
\begin{table}[]
\scalebox{0.98}{\small
\begin{tabular}{|rl|}
\hline
                                 & \multicolumn{1}{c|}{{\bf $\mathcal{M}_{\mathsf{timer}}(\mathcal{D},  \epsilon, f, s, T)$}}\\
$\mathcal{M}_{\mathsf{setup}}$:  & {\bf output} $\left(0, |\mathcal{D}_0| + \lap(\frac{1}{\epsilon}) \right)$ \\
$\mathcal{M}_{\mathsf{update}}$: & $\forall i \in \mathbb{N}^{+}$, run $\mathcal{M}_{\mathsf{unit}}(U[i\cdot T, (i+1)T], \epsilon, T)$\\
                                 & $\mathcal{M}_{\mathsf{unit}}$: {\bf output} $\left(i\cdot T, ~\lap(\frac{1}{\epsilon}) + \sum_{k = i\cdot T + 1}^{(i+1)T} 1 \mid u_{k}\neq\emptyset \right)$\\
$\mathcal{M}_{\mathsf{flush}}$:  & $\forall j \in \mathbb{N}^{+}$, {\bf output} $\left(j\cdot f, ~s\right)$ \\
                                 & \multicolumn{1}{c|}{{\bf $\mathcal{M}_{\mathsf{ANT}}(\mathcal{D},  \epsilon, f, s, \theta)$}} \\
$\mathcal{M}_{\mathsf{setup}}$:  & {\bf output} $\left(0, |\mathcal{D}_0| + \lap(\frac{1}{\epsilon}) \right)$ \\
$\mathcal{M}_{\mathsf{update}}$: & $\epsilon_1 = \epsilon_2 = \frac{\epsilon}{2}$, repeatedly run $\mathcal{M}_{\mathsf{sparse}(\epsilon_1, \epsilon_2, \theta)}$.                                                                                                                       \\
                                 & $\mathcal{M}_{\mathsf{sparse}}$:                                                                                                                                                                                                                      \\
                                 & $\tilde{\theta} = \theta + \lap(\frac{2}{\epsilon_1})$, $t^*\gets$ last time $\mathcal{M}_{\mathsf{sparse}}$'s output $\neq \perp$.                                                                                                                   \\
                                 & $ \forall i\in\mathbb{N}^{+}$, {\bf output} $\begin{cases}    \left(t^{*}+i, c_i + \lap(\frac{1}{\epsilon_2})\right)  &  \text{if}~ v_i + c_i \geq  \tilde{\theta},\\    \perp               &  \text{otherwise}.\end{cases}$ \\
                                 & where $c_i = \sum_{k=t^*}^{t^*+i} 1 \mid u_k\neq \emptyset$, and $v_i = \lap(\frac{4}{\epsilon_1})$.                                                                                                                                                  \\
                                 & {\bf abort} the first time when output $\neq \perp$.                                                                                                                                                                                 \\
$\mathcal{M}_{\mathsf{flush}}$:  & $\forall j \in \mathbb{N}^{+}$, {\bf output} $\left(j\cdot f, ~s\right)$  \\ \hline
\end{tabular}
}
\caption{Mechanisms to simulate the update pattern}
\label{tab:mech}
\vspace{-4mm}
\end{table}
\vspace{-1mm}
\begin{theorem}\label{tm:dptimer}
The update pattern of an \appsystem system implemented with the DP-Timer strategy satisfies Definition~\ref{def:dpsogdb}. 
\end{theorem}
\vspace{-3mm}
\begin{proof}\label{pf:dp-timer}
(Sketch) To capture the information leakage of the update pattern, we rewrite the DP-Timer algorithm to output the total number of synchronized records at each update, instead of signaling the update protocol. The rewritten mechanism $\mathcal{M}_{\mathsf{timer}}$ (see Table~\ref{tab:mech}) simulates the update pattern when applying the DP-Timer strategy. We prove this theorem by illustrating that the composed privacy guarantee of $\mathcal{M}_{\mathsf{timer}}$ satisfies $\epsilon$-DP.

The mechanism $\mathcal{M}_{\mathsf{timer}}$ is a composition of several separated mechanisms. We now discuss the privacy guarantees of each. $\mathcal{M}_{\mathsf{setup}}$ is a Laplace mechanism, thus its privacy guarantee satisfies $\epsilon$-DP. $\mathcal{M}_{\mathsf{flush}}$ reveals a fixed value with fixed time span in a non data-dependent manner, thus it's output distribution is fully computational indistinguishable (satisfies 0-DP). $\mathcal{M}_{\mathsf{update}}$ is a mechanism that repeatedly calls $\mathcal{M}_{\mathsf{unit}}$. $\mathcal{M}_{\mathsf{unit}}$ acts on a fixed time span (T). It counts the total number of received records within the current time period, and outputs a noisy count with $\lap(\frac{1}{\epsilon})$ at the end of the current time period. Thus $\mathcal{M}_{\mathsf{unit}}$ satisfies $\epsilon$-DP guarantee. Since $\mathcal{M}_{\mathsf{update}}$ repeatedly calls  $\mathcal{M}_{\mathsf{unit}}$ and applies it over disjoint data, the privacy guarantee of $\mathcal{M}_{\mathsf{unit}}$ follows parallel composition~\cite{kairouz2015composition}, thus satisfying $\epsilon$-DP. The composition of $\mathcal{M}_{\mathsf{setup}}$ and $\mathcal{M}_{\mathsf{update}}$ also follows parallel composition and the composition of $\mathcal{M}_{\mathsf{flush}}$ follows simple composition~\cite{kairouz2015composition}. Thus the entire algorithm $\mathcal{M}_{\mathsf{timer}}$ satisfies $\left(\max(\epsilon, \epsilon) + 0 \right)$ -DP, which is $\epsilon$-DP.

%We prove this theorem by illustrating the view of any adversary over the neighboring ideal experiments (the ideal experiments that run over adjacent databases), is bounded by DP.
\eat{
For query protocol, the leakage $\lquery$ depends solely on the input query $q_t$ generated by $\mathcal{A}$ itself, thus $\mathcal{A}$'s view over the query protocol does not affect the system's privacy guarantee, in general. According to Algorithm~\ref{algo:timer}, the setup leakage $\lsetup$ is distorted by Laplace mechanism, $\lap(\frac{1}{\epsilon})$. Thus any $\mathcal{A}$'s view over $\lsetup$ is bounded by $\epsilon$-DP. For update protocol, we assume that $\mathcal{A}$ is able to distinguish which updates are caused by synchronization and which are issued by cache flush. For those issued by cache flush, the corresponding leakage are all the same, thus the indistinguishability of $\mathcal{A}$ view follows semantic security definition. Moreover, the DP-Timer algorithm can be abstracted as a Laplacian mechanism satisfying $\epsilon$-DP that runs repeatedly on disjoint data. By applying parallel composition~\cite{kairouz2015composition} theorem, the whole process satisfy $\epsilon$-DP. As a result $\mathcal{A}$'s view of update leakage $\lupdate$ over any neighboring ideal experiments, is bounded by $\epsilon$-DP.

The combined $\mathcal{A}$'s view over $\lsetup$ and $\lupdate$ follows simple composition theorem~\cite{kairouz2015composition}, that is any adversary's view over any neighboring ideal experiments satisfies $\epsilon$-DP. The claim thus holds.}
\end{proof}
\vspace{-4mm}
\begin{theorem}
The update pattern of an \appsystem system implemented with the ANT strategy satisfies Definition~\ref{def:dpsogdb}. 
\end{theorem}
\vspace{-3mm}
\begin{proof} (Sketch)
Similar to the proof of Theorem~\ref{pf:dp-timer}, we first provide $\mathcal{M}_{ANT}$ (Table~\ref{tab:mech}) that simulates the update pattern of ANT strategy. We prove this theorem by illustrating the composed privacy guarantee of $\mathcal{M}_{\mathsf{ANT}}$ satisfies $\epsilon$-DP. 

The mechanism $\mathcal{M}_{\mathsf{ANT}}$ is a composition of several separated mechanisms. $\mathcal{M}_{\mathsf{setup}}$ and $\mathcal{M}_{\mathsf{flush}}$ satisfy $\epsilon$-DP and 0-DP, respectively. We abstract the $\mathcal{M}_{\mathsf{update}}$ as a composite mechanism that repeatedly spawns $\mathcal{M}_{\mathsf{sparse}}$ on disjoint data. Hence, in what follows we show that $\mathcal{M}_{\mathsf{sparse}}$, and thus also $\mathcal{M}_{\mathsf{update}}$ (repeatedly call $\mathcal{M}_{\mathsf{sparse}}$), satisfies $\epsilon$-DP guarantee.

Assume a modified version of $\mathcal{M}_{\mathsf{sparse}}$, say $\mathcal{M'}_{\mathsf{sparse}}$, where it outputs $\top$ once the condition $v_i + c_i >\tilde{\theta}$ is satisfied,  and outputs $\bot$ for all other cases. Then the output of $\mathcal{M'}_{\mathsf{sparse}}$ can be written as $O = \{o_1, o_2,...,o_m \}$, where $\forall~ 1 \leq i < m$, $o_i = \bot$, and $o_m = \top$. Suppose that $U$ and $U'$ are the logical updates of two neighboring growing databases and we know that for all $i$, $\textup{Pr}\left[\tilde{c}_i < x \right] \leq \textup{Pr}\left[\tilde{c}'_i < x+1 \right]$ is satisfied, where $\tilde{c}_i$ and $\tilde{c}'_i$ denotes the $i^{th}$ noisy count when applying  $\mathcal{M'}_{\mathsf{sparse}}$ over $U$ and $U'$ respectively, such that:
\begin{equation}
\begin{split}
  & \textup{Pr}\left[~ \mathcal{M'}_{\mathsf{sparse}}(U) = O \right] \\
  =  \int_{-\infty}^{\infty} & \textup{Pr}\left[\tilde{\theta} = x \right]\left( \prod_{1 \leq i < m}\textup{Pr}\left[\tilde{c}_i < x\right] \right) \textup{Pr}\left[\tilde{c}_m \geq x \right]dx\\
\leq \int_{-\infty}^{\infty} & e^{\epsilon/2}\textup{Pr}\left[\tilde{\theta} = x + 1 \right]\left( \prod_{1 \leq i < m}\textup{Pr}\left[\tilde{c}'_i < x + 1\right] \right) \textup{Pr}\left[ v_m \geq x - {c}_m \right]dx\\
\leq \int_{-\infty}^{\infty} & e^{\epsilon/2}\textup{Pr}\left[\tilde{\theta} = x + 1 \right]\left( \prod_{1 \leq i < m}\textup{Pr}\left[\tilde{c}'_i < x + 1\right] \right) \\
  \times & e^{\epsilon/2} \textup{Pr}\left[ v_m + c'_m \geq x + 1 \right]dx\\
  = \int_{-\infty}^{\infty} & e^{\epsilon}\textup{Pr}\left[\tilde{\theta} = x + 1 \right]\left( \prod_{1 \leq i < m}\textup{Pr}\left[\tilde{c}'_i < x + 1\right] \right) \textup{Pr}\left[ \tilde{c}'_m \geq x + 1 \right]dx\\
  = e^{\epsilon}\textup{Pr} & [\mathcal{M'}_{\mathsf{sparse}}(U') = O]
\end{split}
\end{equation}
Thus $\mathcal{M'}_{\mathsf{sparse}}$ satisfies $\epsilon$-DP, and  $\mathcal{M}_{\mathsf{sparse}}$ is essentially a composition of a $\mathcal{M'}_{\mathsf{sparse}}$ satisfying $\frac{1}{2}\epsilon$-DP together with a Laplace mechanism with privacy parameter equal to $\frac{1}{2}\epsilon$. Hence by applying simple composition~\cite{kairouz2015composition}, we see that $\mathcal{M}_{\mathsf{sparse}}$ satisfies $(\frac{1}{2}\epsilon + \frac{1}{2}\epsilon)-DP$. Knowing that $\mathcal{M}_{\mathsf{update}}$ runs $\mathcal{M}_{\mathsf{sparse}}$ repeatedly on disjoint data, with parallel composition~\cite{kairouz2015composition}, the $\mathcal{M}_{\mathsf{update}}$ then satisfies $\epsilon$-DP. Finally, combined with $\mathcal{M}_{\mathsf{setup}}$ and $\mathcal{M}_{\mathsf{flush}}$, we conclude that $\mathcal{M}_{\mathsf{ANT}}$ satisfies $\epsilon$-DP, thus the theorem holds.
\end{proof}
\eat{
\begin{algorithm}[]
\caption{$\mathcal{M}_{\mathsf{timer}}(\Sigma, \mathcal{D},  \epsilon, k, s, T)$}
\begin{algorithmic}[1]
\Statex \textbf{Input}: growing database $\mathcal{D} = \{\mathcal{D}_0, U\}$,  $\epsilon$, $k$, $s$, and $T$.
\Statex $\mathcal{M}_{\mathsf{setup}} $:
\State ~~{\bf output} $|\mathcal{D}_0| + \lap(\frac{1}{\epsilon}) $
\Statex $\mathcal{M}_{\mathsf{update}}$:
%\kartik{does nto type check. $\prod_{update}$ takes many parameters.}
\For{$t = 1, 2, 3, ...$}
%\If {$u_t \neq \emptyset$}
%\State $c \gets c + 1$
%\State store $u_t$ in the local cache, $\writec(\sigma, u_t)$   %\kartik{are we updating this at every time instant or when a record is received?}
%\State \lcache.\texttt{write}($u_t$)
%\EndIf
\State Repeatedly running $\mathcal{M}_{\mathsf{unit}}(\epsilon, T):$ 
\State {\bf output} $\lap(\frac{1}{\epsilon}) +  (\sum_{t-T}^t 1 \mid u_t\neq \emptyset$) $\mid t \mod T = 0$

\Statex $\mathcal{M}_{\mathsf{flush}}$:
\State {\bf output} $s$ every $k*T$ time steps. 
%\If {$sc \geq K$}
%\State $d\gets\lcache.\texttt{flush(sz)}$, {\bf run} $\prod_{update}\left( d \right)$ \kartik{I don't see the diff between flush, write, reset. Why don't we just perform a write here?}
%\State $sc \gets 0$
%\EndIf
\end{algorithmic}
\label{algo:re-timer}
\end{algorithm}
}
\eat{
\begin{algorithm}[]
\caption{$\mathcal{M}_{\mathsf{ANT}}(\Sigma, \mathcal{D},  \epsilon, l, s, \theta)$}
\begin{algorithmic}[1]
\Statex \textbf{Input}: growing database $\mathcal{D} = \{\mathcal{D}_0, U\}$,  $\epsilon$, $k$, $s$, and $T$.
\Function{$\mathcal{M}_{\mathsf{sparse}}$}{$\epsilon_1, \epsilon_2, \theta$}
    \State $ \tilde{\theta} \leftarrow \theta + \lap(\frac{2}{\epsilon_1})$
    \For{$i = 1,2,3,...$}
    \State compute total records received since last update as $c_i$.
    \State $\tilde{c}_i\gets c_i + v_i \mid v_i \gets \lap(\frac{4}{\epsilon_1})$
    \If {$\tilde{c}_i \geq \tilde{\theta}$}
    \State {\bf output} $c_i + \lap(\frac{1}{\epsilon_2})$, {\bf break}
    \EndIf
    \EndFor
\Statex
\Statex $\mathcal{M}_{\mathsf{setup}} $:
\State ~~{\bf output} $|\mathcal{D}_0| + \lap(\frac{1}{\epsilon}) $
\Statex $\mathcal{M}_{\mathsf{update}}$:
\While{True}
\State $\epsilon_1 = \frac{8}{9}\epsilon, \epsilon_2 = \frac{1}{9}\epsilon$, $\mathcal{M}_{\mathsf{sparse}}(\epsilon_1,\epsilon_2, \theta)$ 
%\State {\bf output} $\lap(\frac{1}{\epsilon}) +  (\sum_{t-T}^t 1 \mid u_t\neq \emptyset$) $\mid t \mod T = 0$
\Statex $\mathcal{M}_{\mathsf{flush}}$:
\State {\bf output} $s$ every $l$ time steps. 
\end{algorithmic}
\label{algo:re-ant}
\end{algorithm}
}

\section{Experimental Analysis}\label{sec:exp}
In this section, we describe our evaluation of \appsystem along
two dimensions: accuracy and performance. Specifically, we address the following questions in our experimental studies: %\kartik{indent space}

\begin{itemize}
    \item \textbf{Question-1}: How do DP strategies compare to na\"ive methods in terms of performance and accuracy under a fixed level of privacy? Do DP strategies guarantee bounded accuracy?
    \item \textbf{Question-2}: What is the impact on accuracy and performance when changing the privacy level of the DP strategies? Can we adjust privacy to obtain different levels of accuracy or performance guarantees?
    %\kartik{Q1 is for fixed privacy vs Q2 is for varying privacy?} \nt{Yes Question1 is end-to-end analysis with fixed privacy} \kartik{I changed Q1 a bit to reflect that.}
    \item \textbf{Question-3}: With a fixed level of privacy, how does accuracy and performance change if we change the non-privacy parameters $T$ or $\theta$ for DP-Timer and DP-ANT, respectively? 
\end{itemize}
\begin{figure*}[h]
\captionsetup[sub]{font=small,labelfont={bf,sf}}
    \begin{subfigure}[b]{0.198\linewidth}
    \centering    \includegraphics[width=1\linewidth]{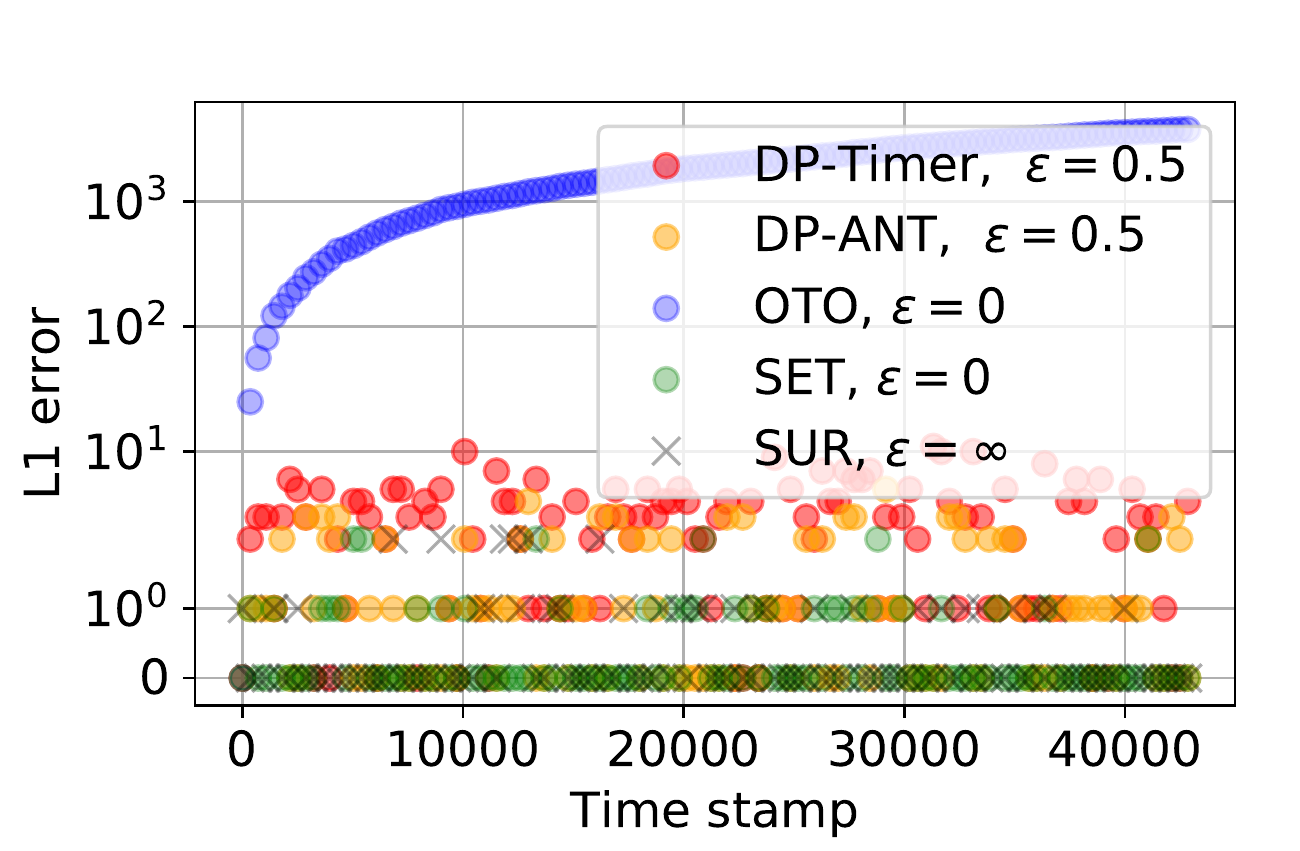}
        \caption{Crypt$\epsilon$: Q1 error}
        \label{fig:acc-q1-c}\end{subfigure}%%
      \begin{subfigure}[b]{0.198\linewidth}
    \centering    \includegraphics[width=1\linewidth]{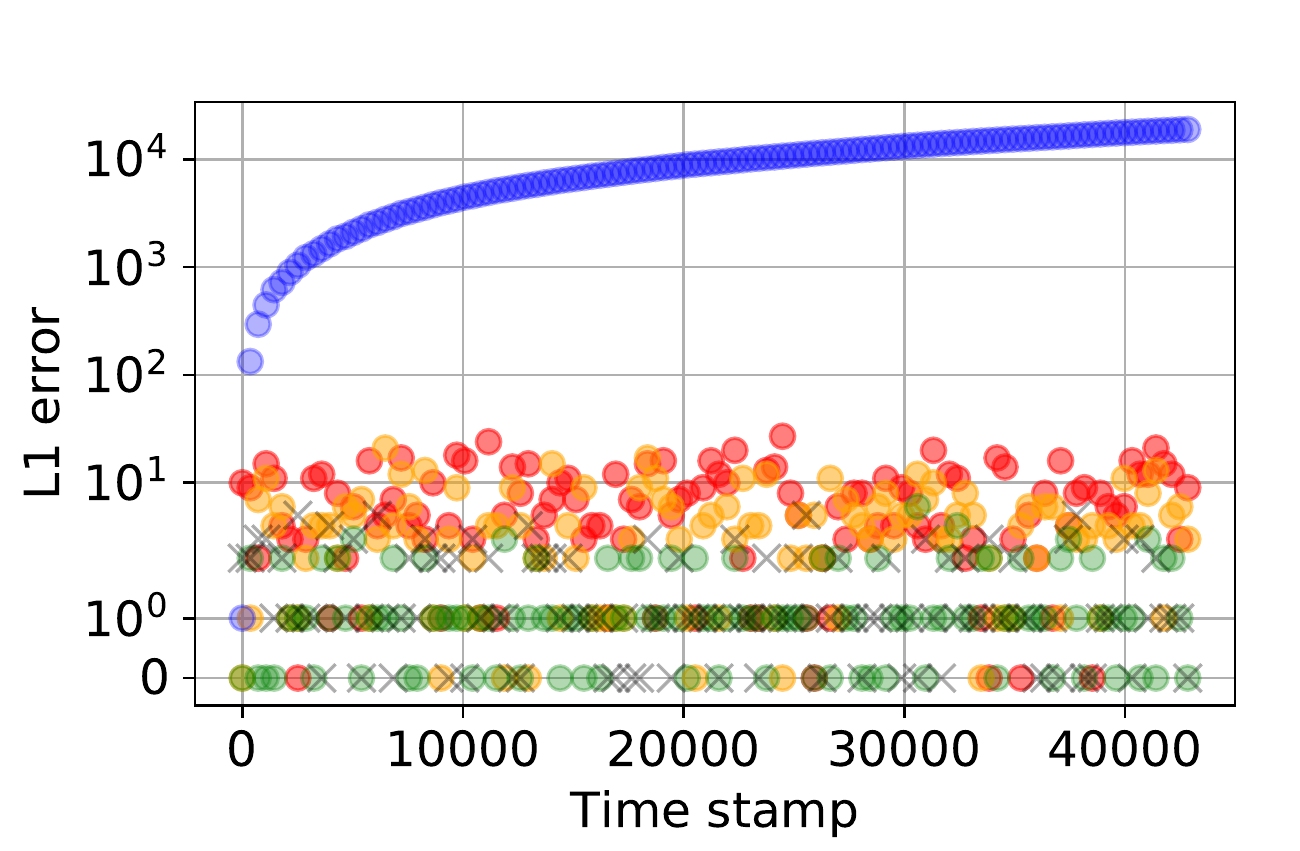}  
        \caption{Crypt$\epsilon$: Q2 error}
        \label{fig:acc-q2-c}
    \end{subfigure}
    \begin{subfigure}[b]{0.198\linewidth}
    \centering    \includegraphics[width=1\linewidth]{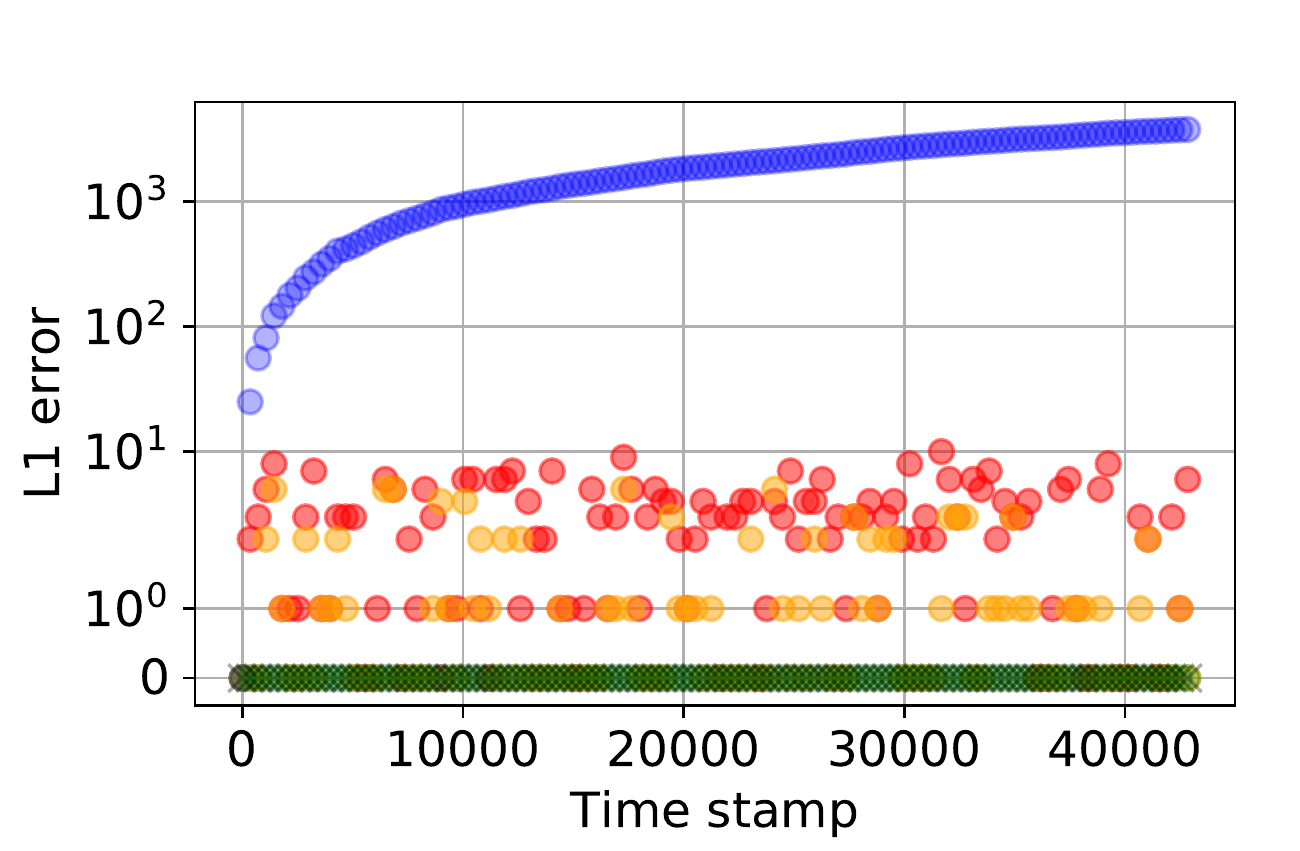}
        \caption{ObliDB: Q1 error}
        \label{fig:acc-q1-ob}\end{subfigure}%%
      \begin{subfigure}[b]{0.198\linewidth}
    \centering    \includegraphics[width=1\linewidth]{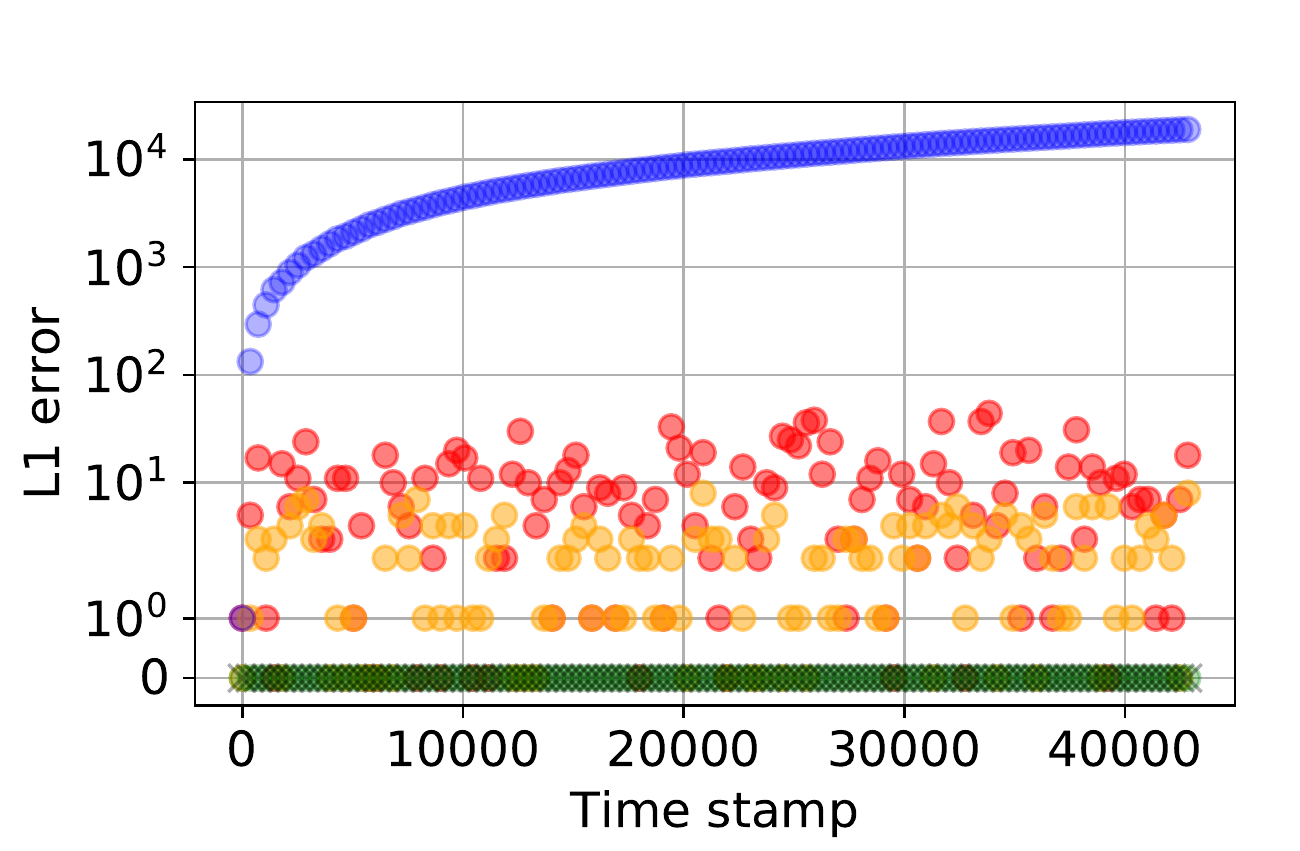}  
        \caption{ObliDB: Q2 error}
        \label{fig:acc-q2-ob}
    \end{subfigure}
    \begin{subfigure}[b]{0.198\linewidth}
    \centering    \includegraphics[width=1\linewidth]{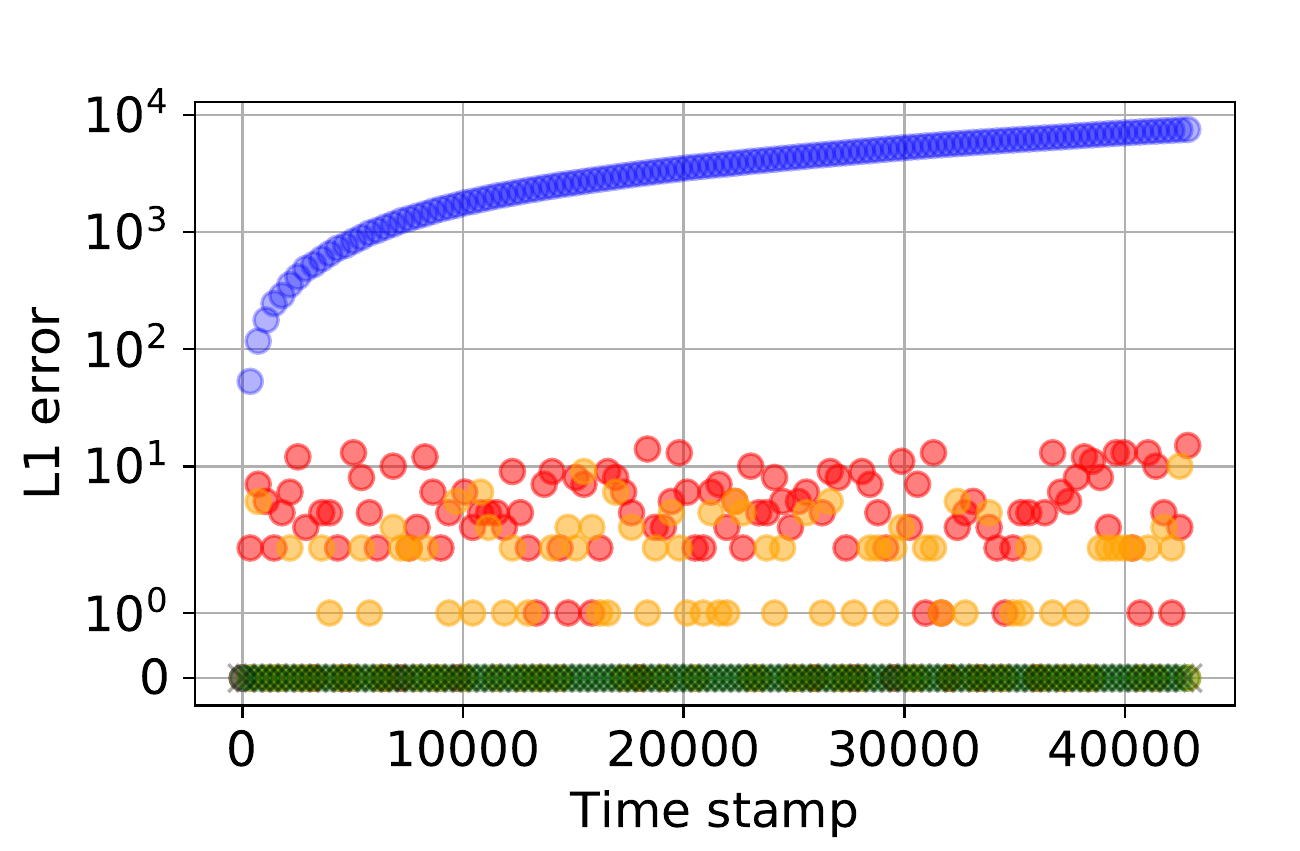}  
        \caption{ObliDB: Q3 error}
        \label{fig:acc-q3-ob}
    \end{subfigure}
    \newline
    \begin{subfigure}[b]{0.198\linewidth}
    \centering    \includegraphics[width=1\linewidth]{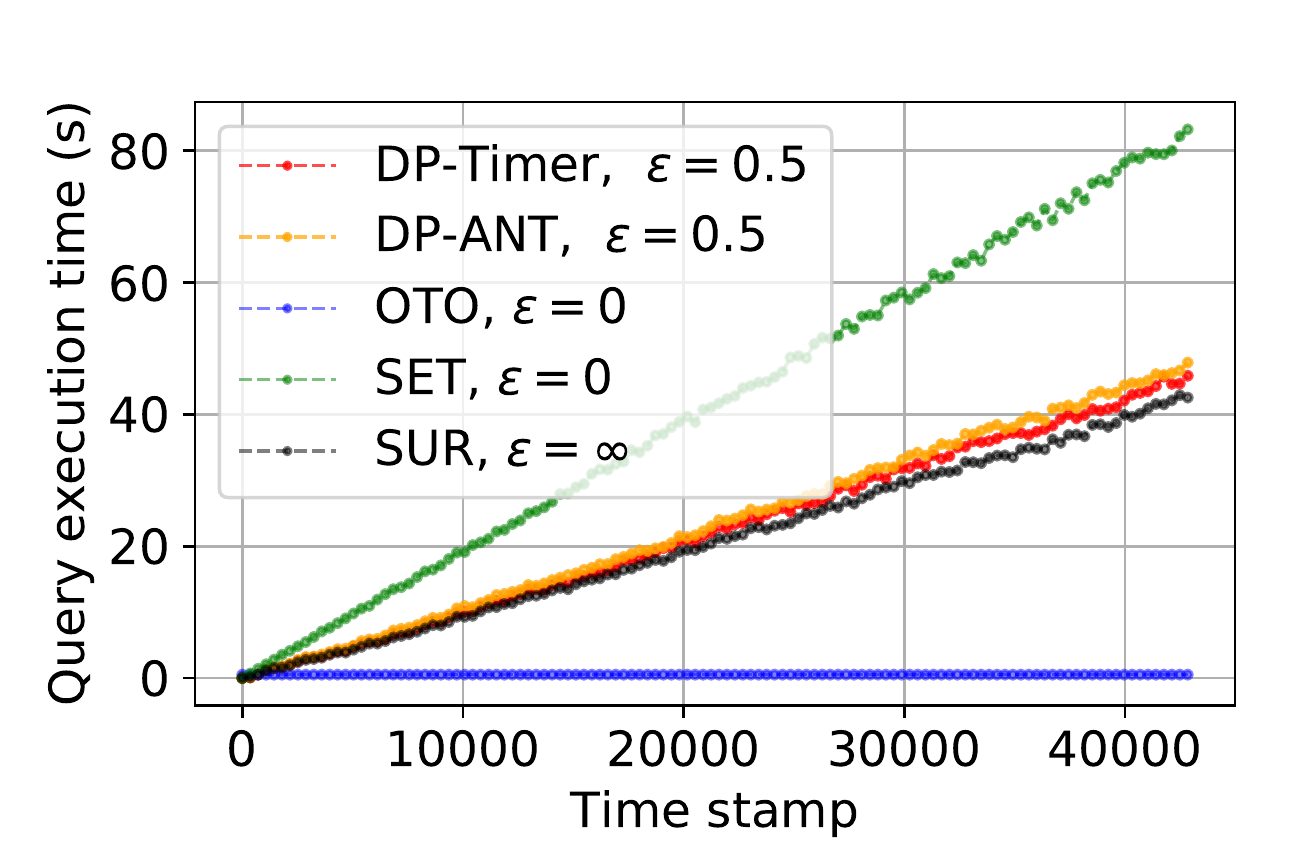}
        \caption{Crypt$\epsilon$: Q1 time}
        \label{fig:q1et-c}
        \end{subfigure}%%
      \begin{subfigure}[b]{0.198\linewidth}
    \centering    \includegraphics[width=1\linewidth]{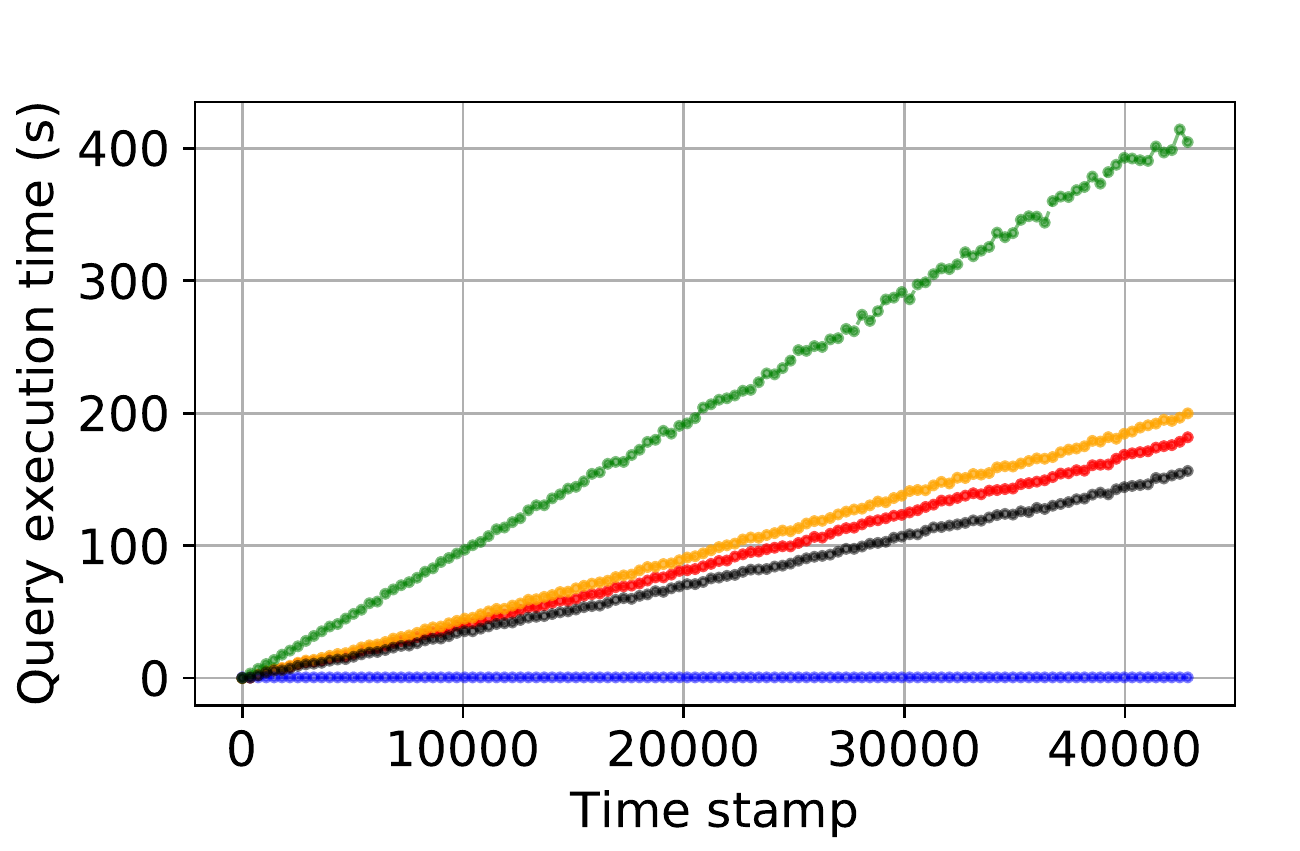}  
        \caption{Crypt$\epsilon$: Q2 time}
        \label{fig:q2et-c}
    \end{subfigure}
    \begin{subfigure}[b]{0.198\linewidth}
    \centering    \includegraphics[width=1\linewidth]{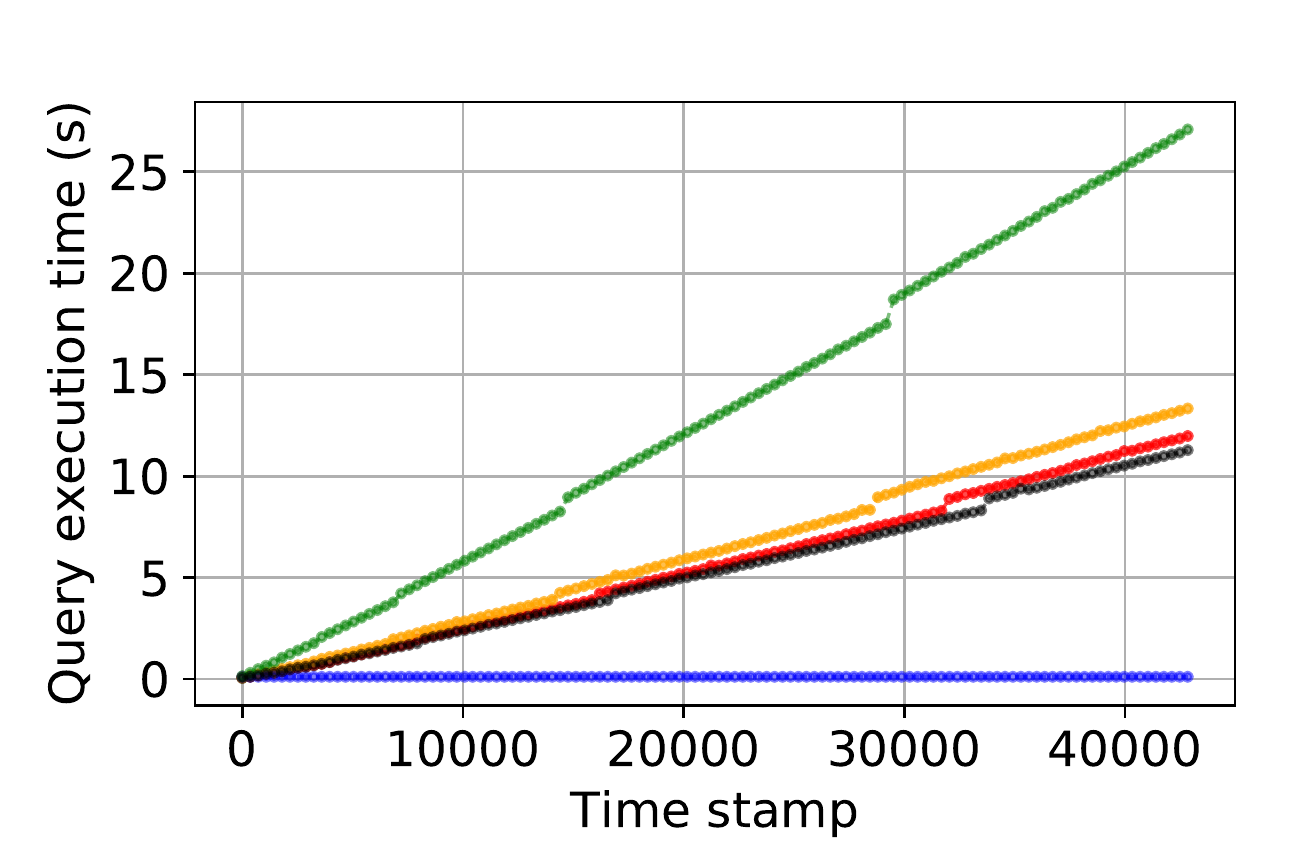}
        \caption{ObliDB: Q1 time}
        \label{fig:q1et-ob}
        \end{subfigure}%%
      \begin{subfigure}[b]{0.198\linewidth}
    \centering    \includegraphics[width=1\linewidth]{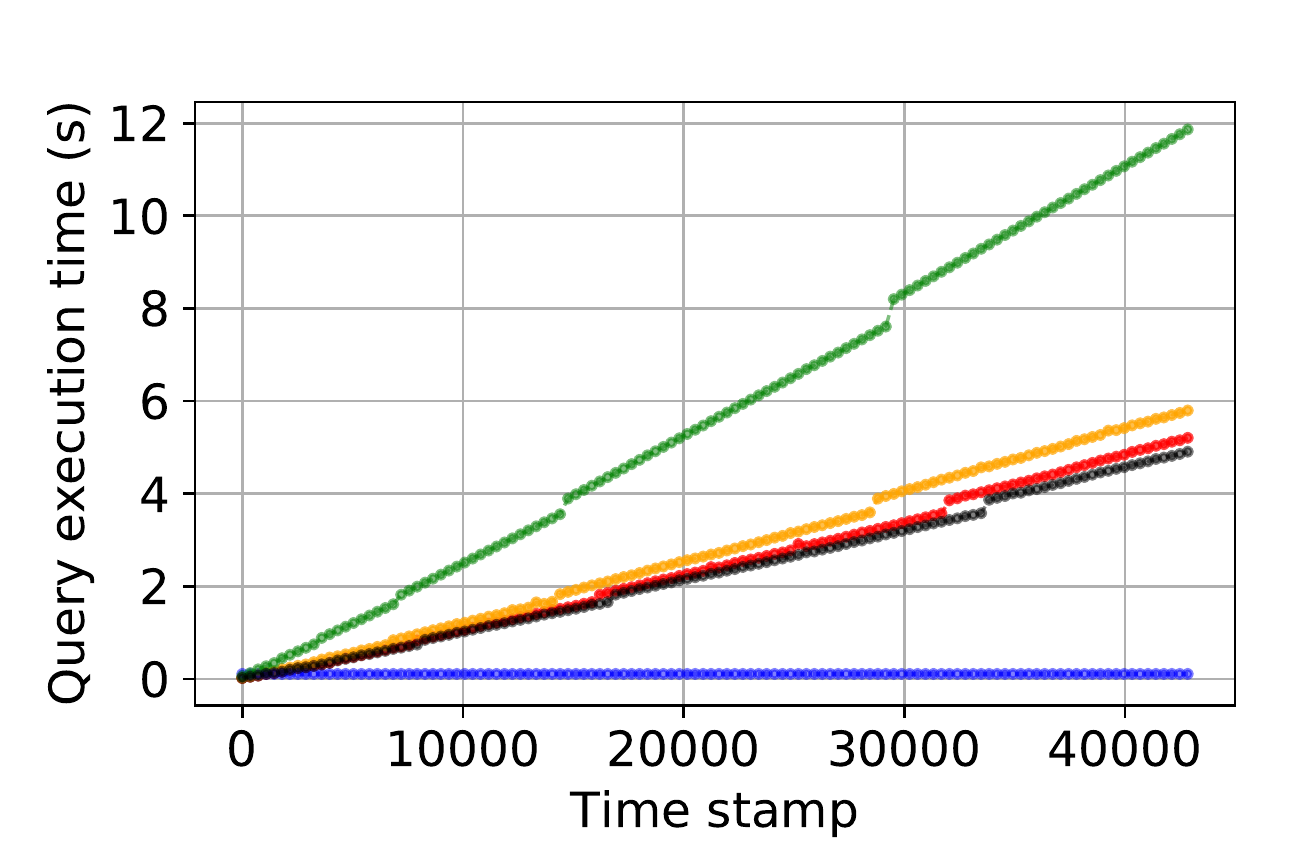}  
        \caption{ObliDB: Q2 time}
        \label{fig:q2et-ob}
    \end{subfigure}
     \begin{subfigure}[b]{0.198\linewidth}
    \centering    \includegraphics[width=1\linewidth]{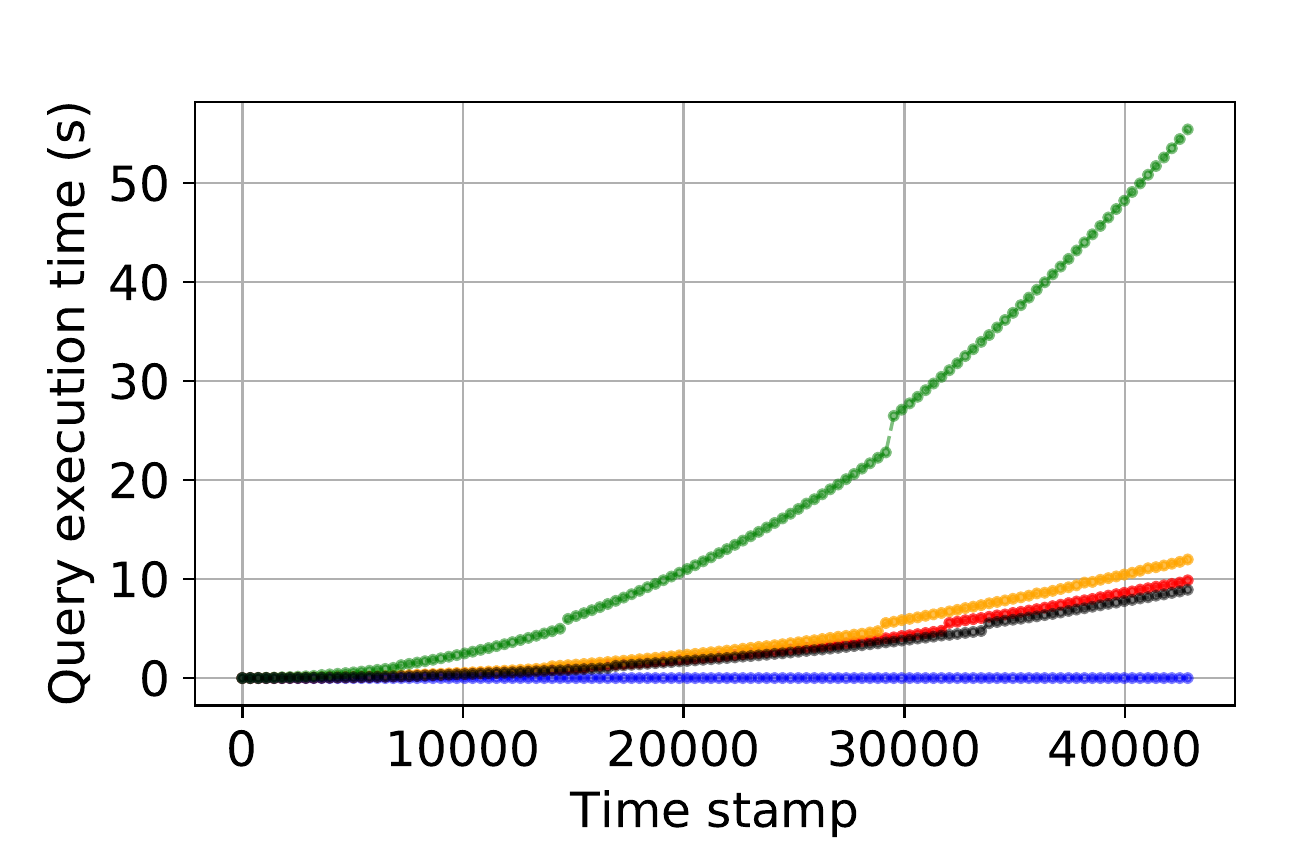}  
        \caption{ObliDB: Q3 time}
        \label{fig:q3et-ob}
    \end{subfigure}
    \vspace{-6mm}
   \caption{End-to-end comparison for synchronization strategies.} \vspace{-4mm}
   \label{fig:q-acc}
\end{figure*}
\vspace{-2mm}
\boldparagraph{Implementation and configuration.} To answer the above questions, we implement multiple instances of \appsystem, execute them with real-world datasets as inputs, and run queries on the deployed system to evaluate different metrics. We implement the \appsystem using two encrypted database schemes, ObliDB~\cite{eskandarian2017oblidb}, and Crypt$\epsilon$~\cite{chowdhury2019cryptc}, from L-0 group and L-DP group, respectively. All experiments are performed on IBM Bare metal servers with 3.8GHz Intel Xeon E-2174G CPU, 32Gb RAM and 64 bit Ubuntu 18.04.1. The ObliDB system is compiled with Intel SGX SDK version 2.9.1. We implement the client using Python 3.7, which takes as input a timestamped dataset, but consumes only one record per round.  The client simulates how a real-world client device would receive new records over time. In our experiment, we assume the time span between two consecutive time stamps is 1 minute. 
%kartik{SOLVED:do we ever go away from this assumption?}

\vspace{-2mm}
\boldparagraph{Data.} We evaluate the two systems using {\it June 2020 New York City Yellow Cab taxi trip record} and {\it June 2020 New York City Green Boro taxi trip record}. Both data sets can be obtained from the TLC Trip Record Project~\cite{taxi2020}. We multiplex the pickup time information of each data point as an indication of when the data \user received this record. We process the raw data with the following steps: (1) Delete invalid data points with incomplete or missing values;  (2) Eliminate duplicated records that occur in the same minute, keeping only one.\footnote{At most one record occurs at each timestamp.} The processed data contains 18,429 and 21,300 records for Yellow Cab and Green Taxi, respectively. (3) Since the monthly data for June 2020 should have 43,200 time units in total, for those time units without associated records, we input a null type record to simulate absence of received data.

\vspace{-2mm}
\boldparagraph{Testing query.}  We select three queries in our evaluation: a linear range query, an aggregation query and a join query.

Q1-Linear range query that counts the total number of records in Yellow Cab data with pickupID within 50-100: ``\texttt{SELECT COUNT(*) FROM YellowCab WHERE pickupID BETWEEN 50 AND 100}''
 
Q2-Aggregation query for Yellow Cab data that counts the number of pickups grouped by location:``\texttt{SELECT pickupID, COUNT(*) AS PickupCnt FROM YellowCab GROUP BY pickupID}''

Q3-Join query that counts how many times both providers have assigned trips: ``\texttt{SELECT COUNT(*) FROM YellowCab INNER JOIN GreenTaxi ON YellowCab.pickTime = GreenTaxi.pickTime}''.

\vspace{-2mm}
\boldparagraph{Default setting.}  Unless specified otherwise, we assume the following defaults. For both DP methods, we set the default privacy as $\epsilon=0.5$, and cache flush parameters as $f = 2000$ (flush interval) and $s = 15$ (flush size). For DP-Timer, the default $T$ is 30 and for DP-ANT the default $\theta=15$. %\kartik{for my understanding, is there a reason why we picked these values? Is this something that needs to be explained?}\nt{I pick those numbers just because I think sync every 30min and sync once get 15 records are reasonable for Taxi use case, should I explain this?} \kartik{do we have a sense of how much our results change if we change these aspects? just trying to ensure that our results are representative/robust.}\nt{For $T$ and $\theta$, we actually have an independent experiment evaluate the results when these values are different. For cache flush, the effect is very tiny. Like if we set cache flush interval from 500- 5000, the average query error changes around 0.1-0.5 for Q2. For flush size any value from 5-60 will get pretty similar accuracy results. } 
We set the ObliDB implementation as the default system and Q2 as the default testing query. 
%Unless further notice, we assume the following default setting for out experiments. \kartik{SOLVED:Unless specified otherwise, we assume the following defaults.}
\eat{
\\\textbf{Accuracy metrics.}\footnote{Note that, unless further explained, all metrics obtained are calculated by the average of 10 replicated experiments} 
We calculate how much unsynchronized data is stored locally before each data synchronization within $N$ time spans. This metric can be interpreted as how much data in total is unavailable between two consecutive synchronizations or the size of local cache whenever \texttt{Sync()} signals update. The smaller value denotes better accuracy.
\\\textbf{Performance metrics.} We count the total number of encrypted records transmitted from \user to \cs after $N$ steps. We use the count as the performance measure, and it can be interpreted as traffic load or storage overhead as well. A smaller value denotes better performance.
}

\subsection{End-to-end Comparison}
In this section, we evaluate {\bf Question-1} by conducting a comparative analysis between the aforementioned DP strategies' empirical accuracy and performance metrics and that of the na\"ive methods. % We run \appsystem under the 5 synchronization strategies and make query requests every 360 time stamps \kartik{time units?} (6 hours). 
We run \appsystem under 5 synchronization strategies and for each group we send testing queries\footnote{Crypt$\epsilon$ does not support join operators, thus we only test Q1 and Q2 for Crypt$\epsilon$} every 360 time units (corresponding to 6 hours). In each group, we report the corresponding L1 error and query execution time (QET) for each testing query as well as the outsourced and dummy data size over time. In addition, we set the privacy budget (used to distort the query answer) of Crypt$\epsilon$ as 3, and we use the default setting for ObliDB with ORAM enabled. % In what follows, we describe our observations according to the obtained results.

\vspace{-1mm}
\boldparagraph{Observation 1. The query errors for both DP strategies are bounded, and such errors are much smaller than that of OTO.} Figure~\ref{fig:q-acc} shows the L1 error and QET for each testing query, the aggregated statistics, such as the mean L1 error and mean QET for all testing queries is reported in Table~\ref{tab:cmp}. First we can observe from Figure~\ref{fig:acc-q1-c}~and~\ref{fig:acc-q1-ob} that the L1 query error of Q1 for both DP strategies fluctuate in the range 0-15. There is no accumulation of query errors as time goes by. Similarly, Figure~\ref{fig:acc-q2-c},~\ref{fig:acc-q2-ob}, and~\ref{fig:acc-q3-ob} show that the errors for both Q2 and Q3 queries are limited to 0-50 under the DP strategies. Note that the query errors in the Crypt$\epsilon$ group are caused by both the unsynchronized records at each time as well as the DP noise injected when releasing the query answer, but the query errors under ObliDB group are caused entirely by unsynchronized records at each time step. This is why, under the Crypt$\epsilon$ group, the SET and SUR methods have non-zero L1 query errors even if these two methods guarantee no unsynchronized data at any time. For the OTO approach, since the user is completely offline after the initial phase, the outsourced database under OTO misses all records after $t=0$, resulting in unbounded query errors. According to Table~\ref{tab:cmp}, the average L1 errors under OTO method are 1929.47, 9214.47, and 3702.6, respectively for Q1, Q2, and Q3, which are at least 520x of that of the DP strategies. 

%Even in the best case, for instance, OTO's L1 query errors in Q1 and Q2 reached 1929.47 and 9464.14, respectively for Crypt$\epsilon$ group, while the corresponding mean L1 query error for DP-Timer method are only 3.69 and 16.69, respectively. DP methods introduce only small performance overhead compared with SUR, and achieve up to 17.6x performance improvement in contrast with SET method
\begin{figure}[ht]
\captionsetup[sub]{font=small,labelfont={bf,sf}}
    \begin{subfigure}[b]{0.5\linewidth}
        \centering
         \includegraphics[width=1\linewidth]{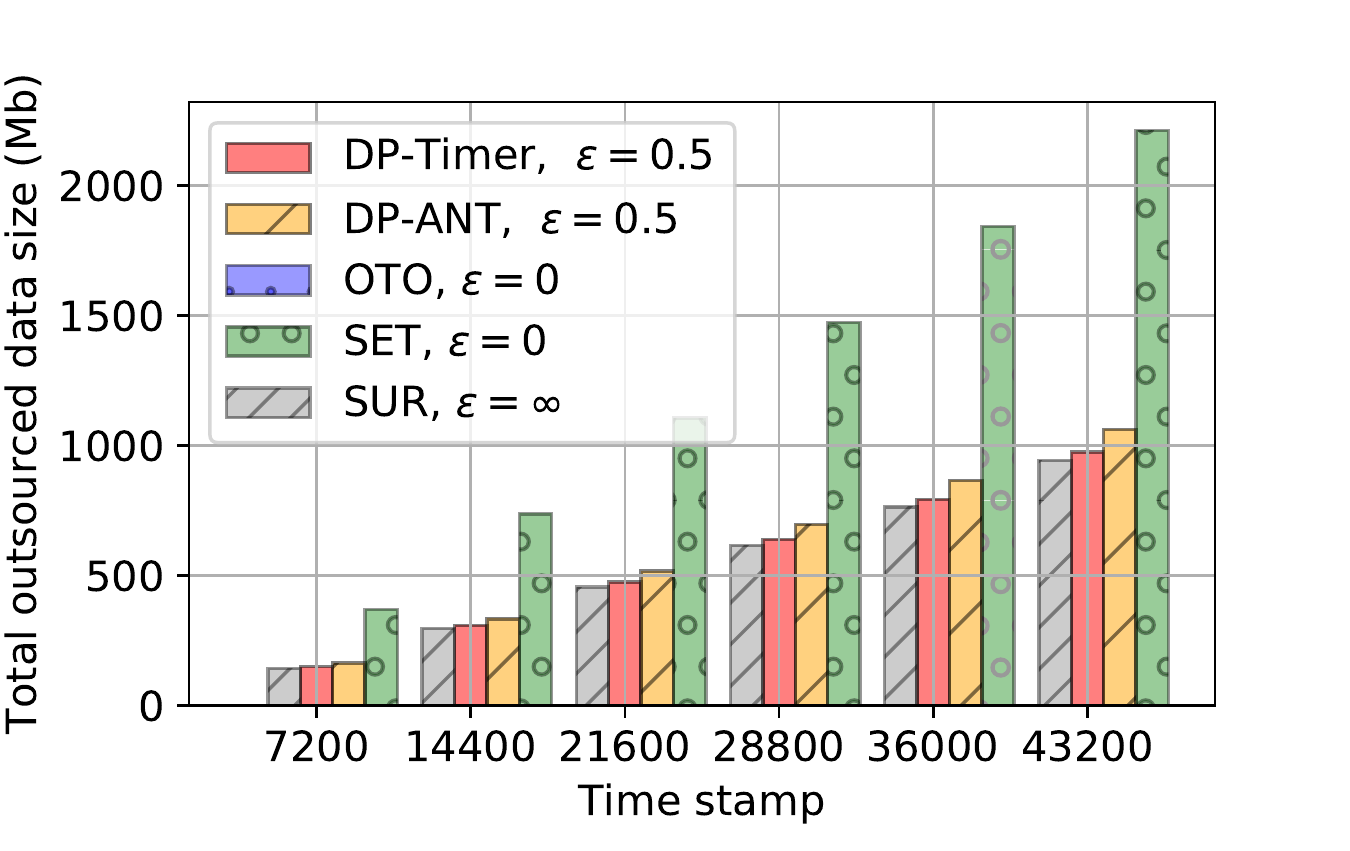}
        \caption{Crypt$\epsilon$: total data size}
        \label{fig:sz-c}
    \end{subfigure}%%
    \begin{subfigure}[b]{0.5\linewidth}
    \centering \includegraphics[width=1\linewidth]{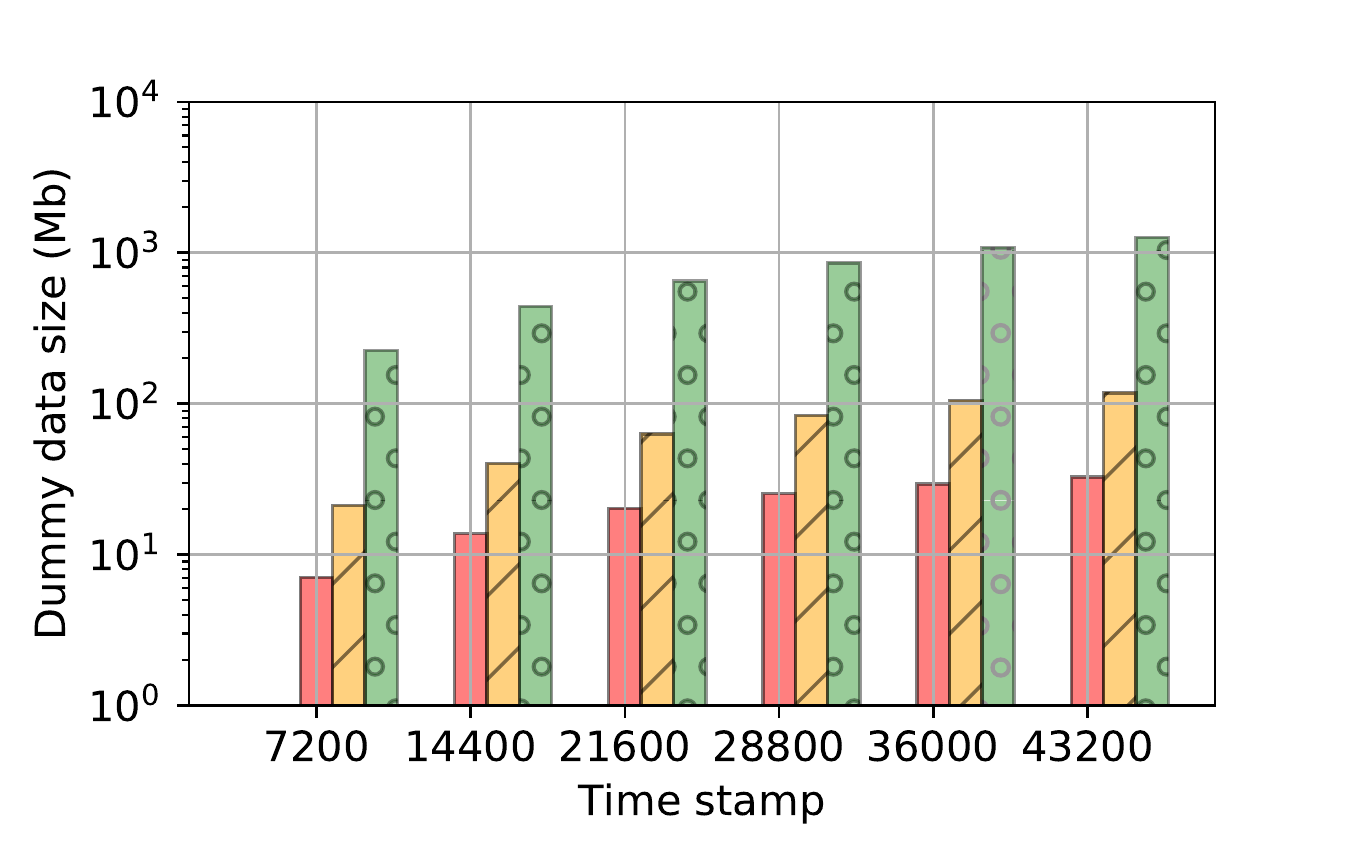}   
 \caption{Crypt$\epsilon$: dummy data size}
        \label{fig:dsz-c}\end{subfigure}%%
    \newline
    \begin{subfigure}[b]{0.5\linewidth}
    \centering    \includegraphics[width=1\linewidth]{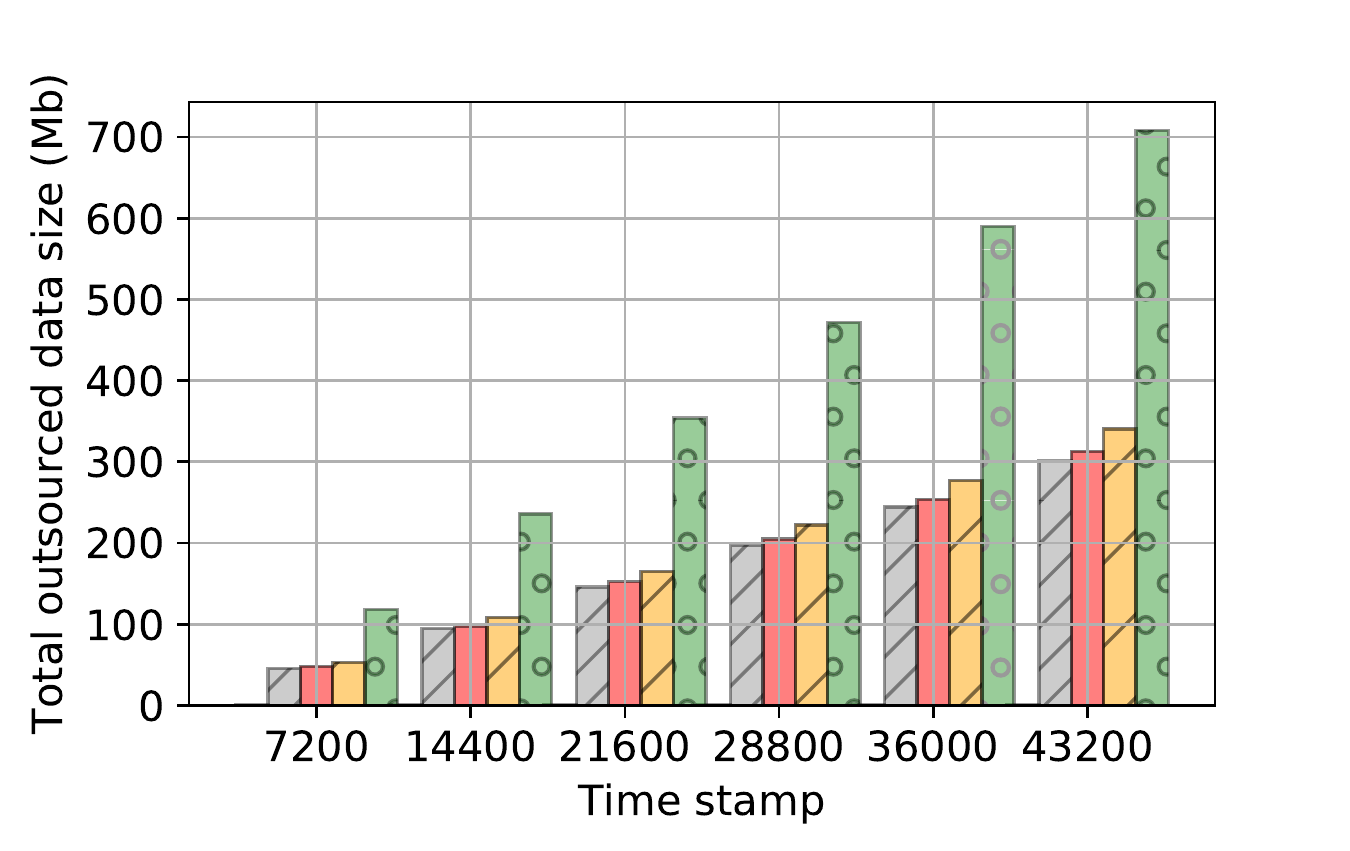}
        \caption{ObliDB: total data size}
        \label{fig:sz-o}\end{subfigure}%%
      \begin{subfigure}[b]{0.5\linewidth}
    \centering    \includegraphics[width=1\linewidth]{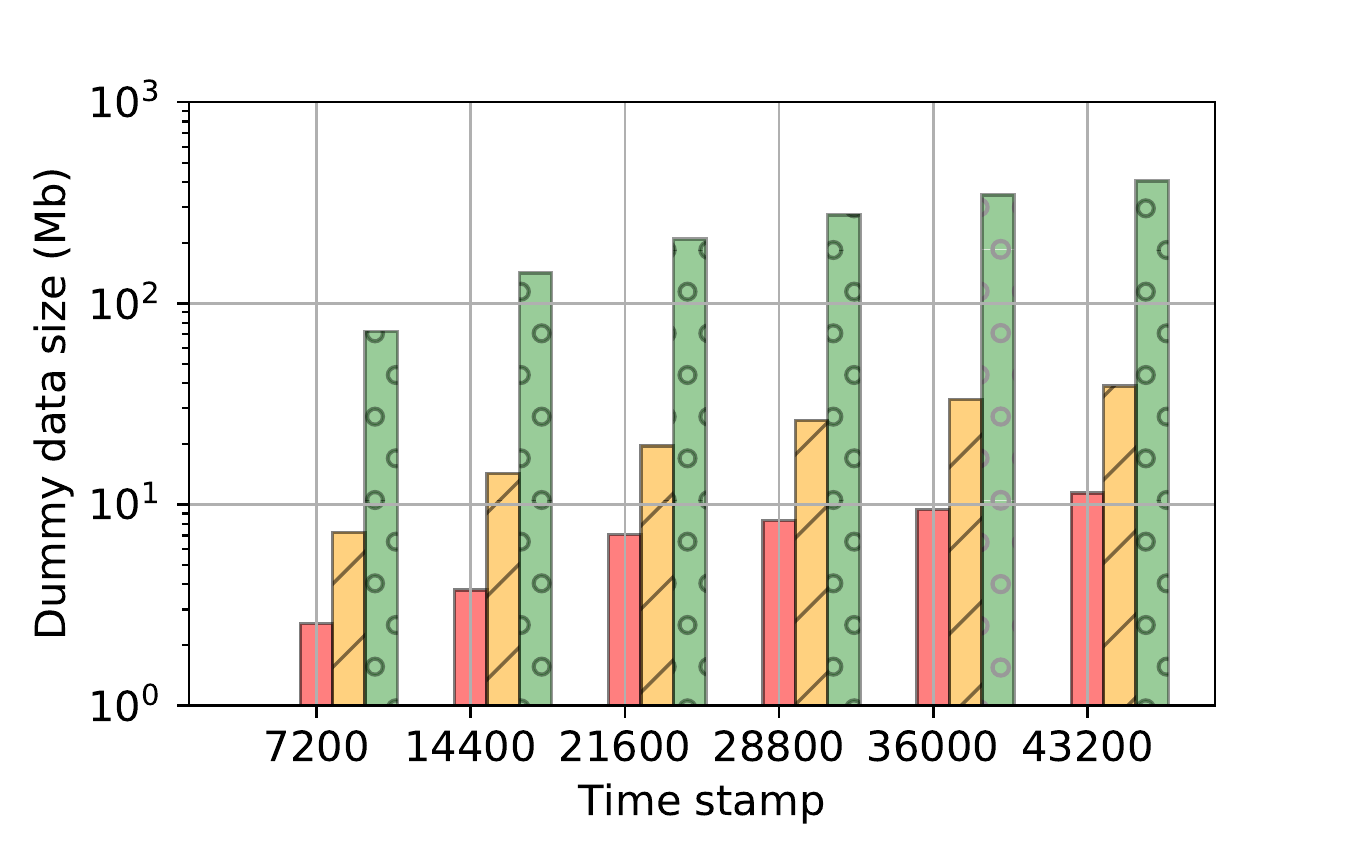}  
        \caption{ObliDB: dummy data size}
        \label{fig:dsz-o}
    \end{subfigure}
    \vspace{-8mm}
   \caption{Total and dummy data size.}%\kartik{increase font size for x-axis, y-axis, legend for all figures. fix overhead vs time. For the last set of figures, there are 5 colors on the legend but not all of them are represented as bars. we should explain that some of them are 0.}}
   \label{fig:perf}
\end{figure}
\vspace{-1mm}
\boldparagraph{Observation 2. The DP methods introduce a small performance overhead compared to SUR, and achieve performance gains up to 5.72x compared to the SET method.} We show the total and dummy data size under each method in Figure~\ref{fig:perf}. According to Figure~\ref{fig:sz-c}~and~\ref{fig:sz-o}, we find that at all time steps, the outsourced data size under both DP approaches are quite similar to that of SUR approach, with at most 6\% additional overhead. 
%Note that SUR updates records as soon as they are received by the client, so the amount of outsourced data under the SUR schema at any time is identical to the amount of data in the logical database. 
However, the SET method outsources at least twice as much data as the DP methods under all cases. In total (Table~\ref{tab:cmp}), SET outsources at least 2.24x and 2.10x more data than DP-Timer and DP-ANT, respectively. OTO always have fixed storage size (0.056 and 0.016 Mb for Crypt$\epsilon$ and ObliDB group) as it only outsources once. Note that the amount of outsourced data under the SUR schema at any time is identical to the amount of data in the logical database. Thus, any oversize of outsourcing data in contrast to SUR is due to the inclusion of dummy data. According to Figure~\ref{fig:dsz-c},~\ref{fig:dsz-o}, and Table~\ref{tab:cmp}, SET introduces at least 11.5x, and can achieve up to 35.6x, more dummy records than DP approaches. Adding dummy data not only inflates the storage, but also results in degraded query response performance. As DP approaches much fewer dummy records, they exhibit little degradation in query performance compared to the SUR method. The SET method, however, uploads many dummy records, thus its query performance drops sharply. According to Figure~\ref{fig:q1et-c},~\ref{fig:q1et-ob},~\ref{fig:q2et-c},~\ref{fig:q2et-ob},~\ref{fig:q3et-ob}, at almost all time steps, the \cs takes twice as much time to run Q1 and Q2 under the SET method than under DP strategies and take at least 4x more time to run Q3. Based on Table~\ref{tab:cmp}, the average QET for Q1 and Q2 under SET are at least 2.17x and 2.3x of that under the DP methods. It's important to point out that both Q1 and Q2 have complexity in $O(N)$, where $N$ is the number of outsourced data. Thus for queries with complexity of $O(N^2)$, such as Q3, the performance gap between the DP strategies and the SET is magnified, in this case boosted to 5.72x. Furthermore, the number of records that SET outsources at any time $t$ is fixed, $|\mathcal{D}_0| +t$. Thus, if the growing database $\mathcal{D} = \{\mathcal{D}_0, U\}$ is sparse (most of the logical updates $u_i \in U$ are $\emptyset$), the performance gap in terms of QET between SET and DP strategies will be further amplified. The the ratio of $(|\mathcal{D}_0| +t) / |\mathcal{D}_t|$ is relatively large if $\mathcal{D}$ is sparse.

%\kartik{do we need to mention that these numbers are sensitive to the dataset? what happens if the data records were sparse? or dense?}

\eat{For example, we provide the least case (ObliDB group, with time at time 7,200), where DP-Timer outsources 48.11Mb data and the DP-ANT outsources total 52.75Mb data while SET outsources 117.96Mb data. }

\begin{figure}[ht]
\captionsetup[sub]{font=small,labelfont={bf,sf}}
    \begin{subfigure}[b]{0.4\linewidth}
        \centering
         \includegraphics[width=1\linewidth]{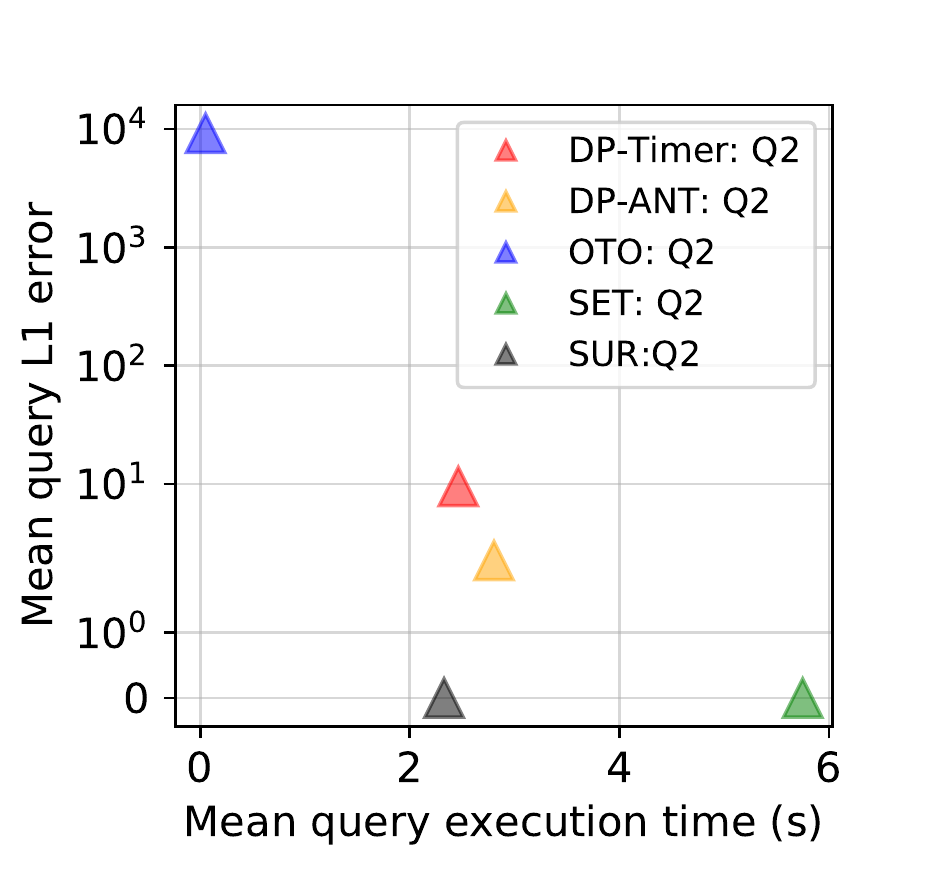}
        \caption{ObliDB group}
        \label{fig:cmp-o}
    \end{subfigure}%%
    \begin{subfigure}[b]{0.4\linewidth}
    \centering \includegraphics[width=1\linewidth]{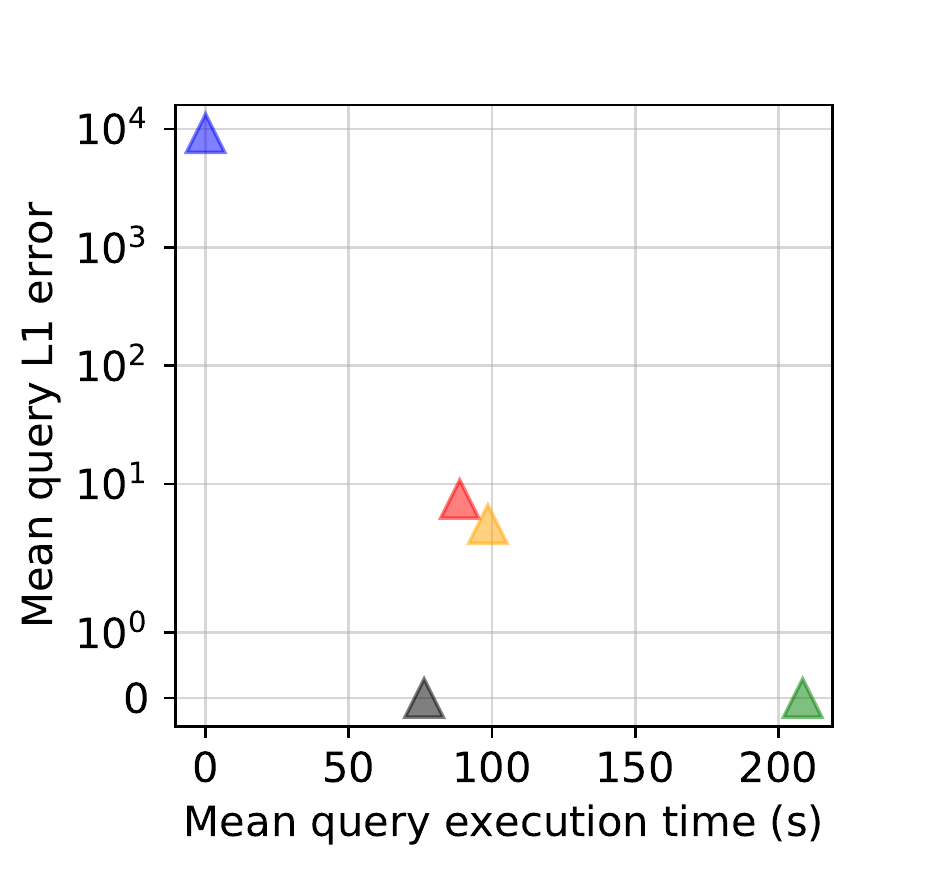}   
 \caption{Crypt$\epsilon$ group}
        \label{fig:cmp-c}\end{subfigure}%%
    \vspace{-3mm}
   \caption{QET v.s. L1 error}
   \label{fig:cross-cmp}
   \vspace{-1mm}
\end{figure}

\vspace{-1mm}
\boldparagraph{Observation 3. DP strategies are optimized for the dual objectives of accuracy and performance.} To better understand the advantage of DP strategies, we compare the default query (Q2) results with respect to DP strategies and nai\"ve methods in Figure~\ref{fig:cross-cmp}, where the x-axis is the performance metric (mean query QET for all queries posted over time), and the y-axis is the accuracy metric (mean query L1 error). Though it seems that SUR is ideal (least query error and no performance overhead), it has no privacy guarantee. Both SET and OTO provide complete privacy. We observe that, the data points of SET fall in the lower right corner of each figure, indicating that the SET method completely sacrifices performance in exchange for a better accuracy guarantee. Thus SET is a private synchronization method that is optimized solely for accuracy. Another extreme case is the OTO method, where the corresponding data points fall in the upper left corner. This means that OTO completely sacrifices accuracy for performance, thus it is optimized for performance only. DP strategies provide privacy guarantees bounded by $\epsilon$, and we observe that the corresponding data points fall in the lower left corner of the figure (close to SUR), indicating that the DP strategies provide considerable accuracy guarantees (or bounded error) at a small sacrifice in performance. This is further evidence that DP strategies are optimized for the dual objectives of accuracy and performance. 

%Though, it seems SUR is ideal but it has no privacy guarantee, 

%\kartik{mention that SUR seems ideal but it has no privacy guarantee. On the other hand, the other two have privacy better than differentially private methods.}

%According to Table~\ref{tab:cmp}, We know that, on average, DP-ANT has a smaller query error than the DP-Timer method in all cases. The reason is we have added relatively large noise ($\epsilon_1 = 0.5/2$) to the threshold as well as the count each time for DP-ANT, thus the synchronization condition (Algorithm~\ref{algo:cbuffer}:11) hits more frequently, resulting in a faster update frequency. However, more frequent synchronization means that more dummy data will be injected into the outsourcing database, so based on Figure~\ref{fig:dsz-c} and~\ref{fig:dsz-o}, we observe that at any time, DP-ANT raises more dummy data than the DP-Timer method, (around 2.3x on average over time denoted by Table~\ref{tab:cmp}). Similar performance divergence can also be observed in Figure~\ref{fig:q1et-c},~\ref{fig:q2et-c},~\ref{fig:q1et-o}, and~\ref{fig:q2et-o}. For almost all cases, the query execution time of DP-Timer method is slightly smaller than that obtained under DP-ANT strategy.

\eat{
\noindent\textbf{Experiment design.} We mainly consider two groups of experimental studies in our evaluation:  (1) Comparison with na\"ive methods: In this group we fix the parameters for DP strategies and compare the corresponding performance and accuracy with na\"ive methods. For DP-Timer, we set $K=20, \epsilon = 0.5, f = 50$, and $s = 15$. For DP-ANT, we set $\theta = 30, \epsilon = 0.5, f = 200$, and $s = 15$;~ (2) DP strategies with changing privacy: In this experiment group, we fix all other parameters for DP strategies ($K=40, f = 50, s = 15$ for DP-Timer, and $\theta = 20, f = 2000, s = 15$ for DP-ANT) except $\epsilon$, we then change $\epsilon$ from 0.1 to 2.0 and evaluate the \appsystem's accuracy and performance under each $\epsilon$. 
}

% Please add the following required packages to your document preamble:
% \usepackage[table,xcdraw]{xcolor}
% If you use beamer only pass "xcolor=table" option, i.e. \documentclass[xcolor=table]{beamer}

\begin{table}[]
\scalebox{0.95}{\small
\begin{tabular}{|llccccc|}
\hline
\multicolumn{2}{|l|}{\textbf{\begin{tabular}[c]{@{}l@{}}Comparison \\ Categories\end{tabular}}} & \multicolumn{1}{c|}{\textbf{SUR}} & \multicolumn{1}{c|}{\textbf{SET}}   & \multicolumn{1}{c|}{\textbf{OTO}} & \multicolumn{1}{c|}{\textbf{\begin{tabular}[c]{@{}c@{}}DP\\ Timer\end{tabular}}} & \textbf{\begin{tabular}[c]{@{}c@{}}DP\\ ANT\end{tabular}} \\ \hline
                                            &                                                     & \multicolumn{1}{l}{}              & \multicolumn{1}{l}{\textbf{Crypt$\epsilon$}} & \multicolumn{1}{l}{}              & \multicolumn{1}{l}{}                                                             & \multicolumn{1}{l|}{}                                     \\ \hline
\multicolumn{1}{|l|}{}                      & \multicolumn{1}{l|}{\textbf{Mean L1 Err}}           & \multicolumn{1}{c|}{0.75}         & \multicolumn{1}{c|}{0.71}           & \multicolumn{1}{c|}{1929.47}      & \multicolumn{1}{c|}{3.04}                                                        & 0.99                                                      \\
\multicolumn{1}{|l|}{\textbf{Q1}}           & \multicolumn{1}{l|}{\textbf{Max L1 Err}}            & \multicolumn{1}{c|}{2}            & \multicolumn{1}{c|}{3}              & \multicolumn{1}{c|}{8079}         & \multicolumn{1}{c|}{11}                                                          & 5                                                         \\
\multicolumn{1}{|l|}{}                      & \multicolumn{1}{l|}{\textbf{Mean QET}}              & \multicolumn{1}{c|}{20.94}        & \multicolumn{1}{c|}{41.70}          & \multicolumn{1}{c|}{0.33}         & \multicolumn{1}{c|}{22.51}                                                       & 23.54                                                     \\ \hline
\multicolumn{1}{|l|}{}                      & \multicolumn{1}{l|}{\textbf{Mean L1 Err}}           & \multicolumn{1}{c|}{3.36}         & \multicolumn{1}{c|}{3.39}           & \multicolumn{1}{c|}{9464.14}      & \multicolumn{1}{c|}{7.45}                                                        & 4.56                                                      \\
\multicolumn{1}{|l|}{\textbf{Q2}}           & \multicolumn{1}{l|}{\textbf{Max L1 Err}}            & \multicolumn{1}{c|}{13}           & \multicolumn{1}{c|}{15}             & \multicolumn{1}{c|}{18446}        & \multicolumn{1}{c|}{27}                                                          & 21                                                        \\
\multicolumn{1}{|l|}{}                      & \multicolumn{1}{l|}{\textbf{Mean QET}}              & \multicolumn{1}{c|}{76.34}        & \multicolumn{1}{c|}{208.47}         & \multicolumn{1}{c|}{0.72}         & \multicolumn{1}{c|}{88.75}                                                       & 98.58                                                     \\ \hline
\multicolumn{2}{|l|}{\textbf{Mean logical gap}}                                                            & \multicolumn{1}{c|}{0}            & \multicolumn{1}{c|}{0}              & \multicolumn{1}{c|}{9214.5}       & \multicolumn{1}{c|}{10.91}                                                       & 4.8                                                       \\ \hline
\multicolumn{2}{|l|}{\textbf{Total data (Mb)}}                                                    & \multicolumn{1}{c|}{943.5}        & \multicolumn{1}{c|}{2211.79}        & \multicolumn{1}{c|}{0.052}        & \multicolumn{1}{c|}{979.10}                                                      & 1027.31                                                   \\
\multicolumn{2}{|l|}{\textbf{Dummy data (Mb)}}                                                    & \multicolumn{1}{c|}{N/A}          & \multicolumn{1}{c|}{1268.29}        & \multicolumn{1}{c|}{N/A}          & \multicolumn{1}{c|}{35.6}                                                        & 83.8                                                      \\ \hline
                                            &                                                     & \multicolumn{1}{l}{}              & \multicolumn{1}{l}{\textbf{ObliDB}} & \multicolumn{1}{l}{}              & \multicolumn{1}{l}{}                                                             & \multicolumn{1}{l|}{}                                     \\ \hline
\multicolumn{1}{|l|}{}                      & \multicolumn{1}{l|}{\textbf{Mean L1 Err}}           & \multicolumn{1}{c|}{0}            & \multicolumn{1}{c|}{0}              & \multicolumn{1}{c|}{1929.47}      & \multicolumn{1}{c|}{2.95}                                                        & 0.91                                                      \\
\multicolumn{1}{|l|}{\textbf{Q1}}           & \multicolumn{1}{l|}{\textbf{Max L1 Err}}            & \multicolumn{1}{c|}{0}            & \multicolumn{1}{c|}{0}              & \multicolumn{1}{c|}{8801}         & \multicolumn{1}{c|}{10}                                                          & 5                                                         \\
\multicolumn{1}{|l|}{}                      & \multicolumn{1}{l|}{\textbf{Mean QET}}              & \multicolumn{1}{c|}{5.39}         & \multicolumn{1}{c|}{14.18}          & \multicolumn{1}{c|}{0.041}        & \multicolumn{1}{c|}{5.69}                                                        & 6.48                                                      \\ \hline
\multicolumn{1}{|l|}{\textbf{}}             & \multicolumn{1}{l|}{\textbf{Mean L1 Err}}           & \multicolumn{1}{c|}{0}            & \multicolumn{1}{c|}{0}              & \multicolumn{1}{c|}{9214.51}      & \multicolumn{1}{c|}{9.25}                                                        & 2.25                                                      \\
\multicolumn{1}{|l|}{\textbf{Q2}}           & \multicolumn{1}{l|}{\textbf{Max L1 Err}}            & \multicolumn{1}{c|}{0}            & \multicolumn{1}{c|}{0}              & \multicolumn{1}{c|}{18429}        & \multicolumn{1}{c|}{44}                                                          & 8                                                         \\
\multicolumn{1}{|l|}{\textbf{}}             & \multicolumn{1}{l|}{\textbf{Mean QET}}              & \multicolumn{1}{c|}{2.32}         & \multicolumn{1}{c|}{5.76}           & \multicolumn{1}{c|}{0.071}        & \multicolumn{1}{c|}{2.46}                                                        & 2.80                                                      \\ \hline
\multicolumn{1}{|l|}{\textbf{}}             & \multicolumn{1}{l|}{\textbf{Mean L1 Err}}           & \multicolumn{1}{c|}{0}            & \multicolumn{1}{c|}{0}              & \multicolumn{1}{c|}{3702.6}       & \multicolumn{1}{c|}{4.93}                                                        & 1.43                                                      \\
\multicolumn{1}{|l|}{\textbf{Q3}}           & \multicolumn{1}{l|}{\textbf{Max L1 Err}}            & \multicolumn{1}{c|}{0}            & \multicolumn{1}{c|}{0}              & \multicolumn{1}{c|}{7407}         & \multicolumn{1}{c|}{15}                                                          & 10                                                        \\
\multicolumn{1}{|l|}{\textbf{}}             & \multicolumn{1}{l|}{\textbf{Mean QET}}              & \multicolumn{1}{c|}{2.77}         & \multicolumn{1}{c|}{17.86}          & \multicolumn{1}{c|}{0.095}        & \multicolumn{1}{c|}{3.12}                                                        & 3.86                                                      \\ \hline
\multicolumn{2}{|l|}{\textbf{Mean Logical gap}}                                                   & \multicolumn{1}{c|}{0}            & \multicolumn{1}{c|}{0}              & \multicolumn{1}{c|}{9214.5}       & \multicolumn{1}{c|}{10.73}                                                       & 2.96                                                      \\ \hline
\multicolumn{2}{|l|}{\textbf{Total data (Mb)}}                                                    & \multicolumn{1}{c|}{301.90}       & \multicolumn{1}{c|}{707.79}         & \multicolumn{1}{c|}{0.016}        & \multicolumn{1}{c|}{316.68}                                                      & 337.09                                                    \\
\multicolumn{2}{|l|}{\textbf{Dummy data (Mb)}}                                                    & \multicolumn{1}{c|}{N/A}          & \multicolumn{1}{c|}{405.89}         & \multicolumn{1}{c|}{N/A}          & \multicolumn{1}{c|}{14.78}                                                       & 35.19                                                     \\ \hline
\end{tabular}
}
\caption{Aggregated statistics for comparison experiment}\vspace{-3mm}
\label{tab:cmp}
\vspace{-2mm}
\end{table}

\vspace{-3mm}
\subsection{Trade-off with Changing Privacy Level}
We address {\bf Question-2} by evaluating the DP policies with different $\epsilon$ ranging from 0.001 to 10. For other parameters associated with DP strategies, we apply the default setting and evaluate them with the default testing query Q2 on the default system (ObliDB based implementation). For each $\epsilon$, we report the mean query error and QET. We summarize our observations as follows.

\begin{figure}[ht]
\captionsetup[sub]{font=small,labelfont={bf,sf}}
    \begin{subfigure}[b]{0.49\linewidth}
        \centering
         \includegraphics[width=1\linewidth]{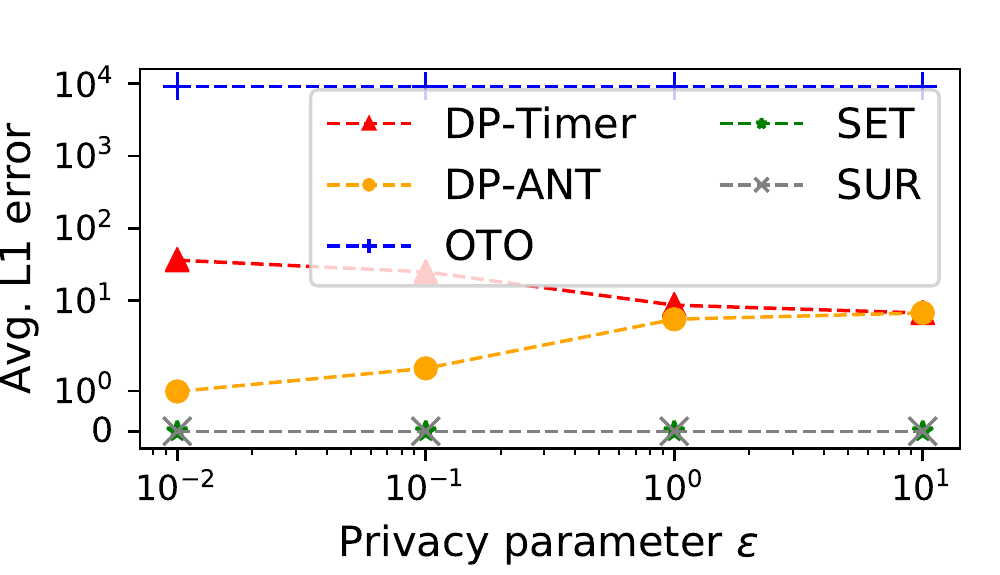}
        \caption{Avg. L1 error v.s. Privacy}
        \label{fig:eps-acc}
    \end{subfigure}%%
    \begin{subfigure}[b]{0.49\linewidth}
    \centering \includegraphics[width=1\linewidth]{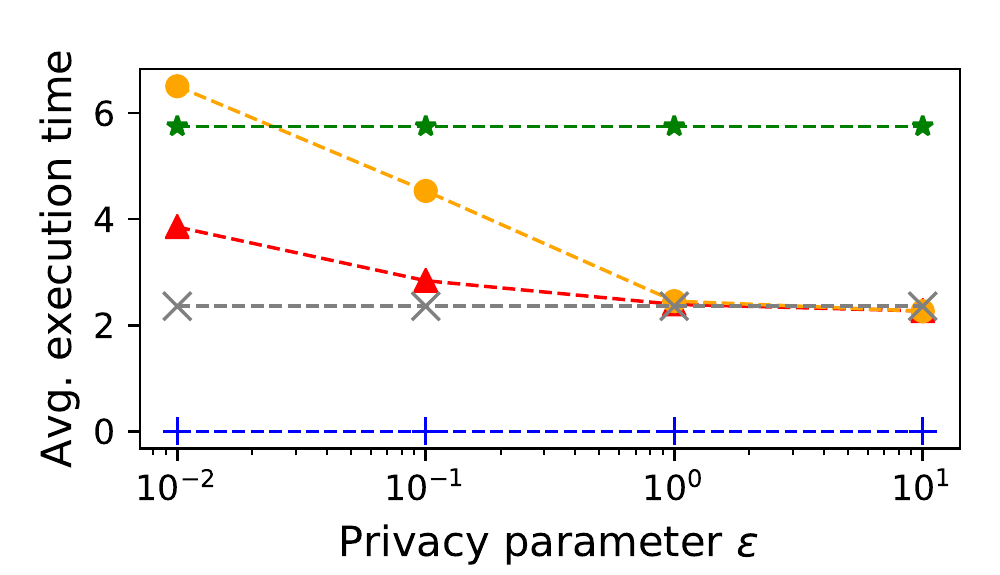}   
 \caption{Avg. execution time v.s. Privacy}
        \label{fig:eps-perf}\end{subfigure}%%
    \vspace{-4mm}
   \caption{Trade-off with changing privacy.}
   \label{fig:trade-off-privacy}
   \vspace{-1mm}
\end{figure}

\vspace{-1mm}
\boldparagraph{Observation 4. DP-Timer and DP-ANT exhibit different trends in accuracy when $\epsilon$ changes.} Figure~\ref{fig:eps-acc} illustrates the evaluation results for privacy versus accuracy. In general, we observe that, as $\epsilon$ increases from 0.01 to 1, the mean query error of DP-ANT increases while the error of DP-Timer decreases. Both errors change slightly from $\epsilon=1$ to $\epsilon=10$. Recall that DP-Timer's logic gap consists of the number of records received since the last update,$c_{t*}^{t}$, and the data delayed by the previous synchronization operation (bounded by $O(2\sqrt{k}/\epsilon)$). Since the update frequency of the DP-Timer is fixed, $c_{t*}^{t}$ is not affected when $\epsilon$ changes. However, when the $\epsilon$ is smaller, the number of delayed records increases, which further leads to higher query errors. For the DP-ANT though, when the $\epsilon$ is very small, the delayed records increases as well (bouned by $O(16\log{t}/\epsilon)$). However, smaller $\epsilon$  (large noise) will result in more frequent updates for the DP-ANT. This is because the large noise will cause the DP-ANT to trigger the upload condition early before it receives enough data. As a result, the number of records received since last update, $c_{t*}^{t}$, will be reduced, which essentially produces smaller query errors. In summary, for DP strategies, we observe that there is a trade-off between privacy and accuracy guarantee.
%Recall that DP-ANT updates the database if the \user has received enough new data. However, when $\epsilon$ is small (noise is large), then the \user is more likely to perform updates without receiving sufficient data, which increases the synchronization frequency. In such a case, with a small $\epsilon$, the accuracy guarantee for DP-ANT approaches that of SET and SUR. As $\epsilon$ increases, the update frequency of DP-ANT decreases, and as a result, query errors increase. 

%For DP-Timer, when $\epsilon$ is relatively small, as the noise injected to distort the update record size is large, the chance of delaying large amounts of data per synchronization increases, which causes large query errors.\kartik{Not sure I understand the previous statement. What is this chance over?} However, as $\epsilon$ increases, less data is delayed thus the query error decreases accordingly. In a summary, the experiments show that, DP strategies can be tuned to obtain different accuracy by adjusting the privacy level.
 
%When the noise is reduced to a certain level (i.e. $\epsilon\geq1$), the update frequency remains stable, thus the query error keeps. 

\vspace{-2mm}
\boldparagraph{Observation 5. Both DP strategies show decreasing performance overhead when $\epsilon$ increases.}
 Both DP methods show similar tendencies in terms of the performance metrics (Figure~\ref{fig:eps-perf}). When $\epsilon$ increases, the QET decreases. This can be explained by Theorem~\ref{sz:timer} and ~\ref{sz:ant}. That is, with a relatively large $\epsilon$, the dummy records injected at each update will be reduced substantially. As a result, less overhead will be introduced and the query response performance is then increased. Similarly, for DP strategies, there is a trade-off between privacy and performance. 
 %The Crypt$\epsilon$ group shows similar results, for space concern, we move the corresponding results to Appendixxxx. 
\vspace{-2mm}
\subsection{Trade-off with Fixed Privacy Level}
We address {\bf Question-3} by evaluating the DP policies with default $\epsilon$ but changing $T$ and $\theta$ for DP-Timer and DP-ANT, respectively. 

\begin{figure}[ht]
\captionsetup[sub]{font=small,labelfont={bf,sf}}
    \begin{subfigure}[b]{0.49\linewidth}
        \centering
         \includegraphics[width=1\linewidth]{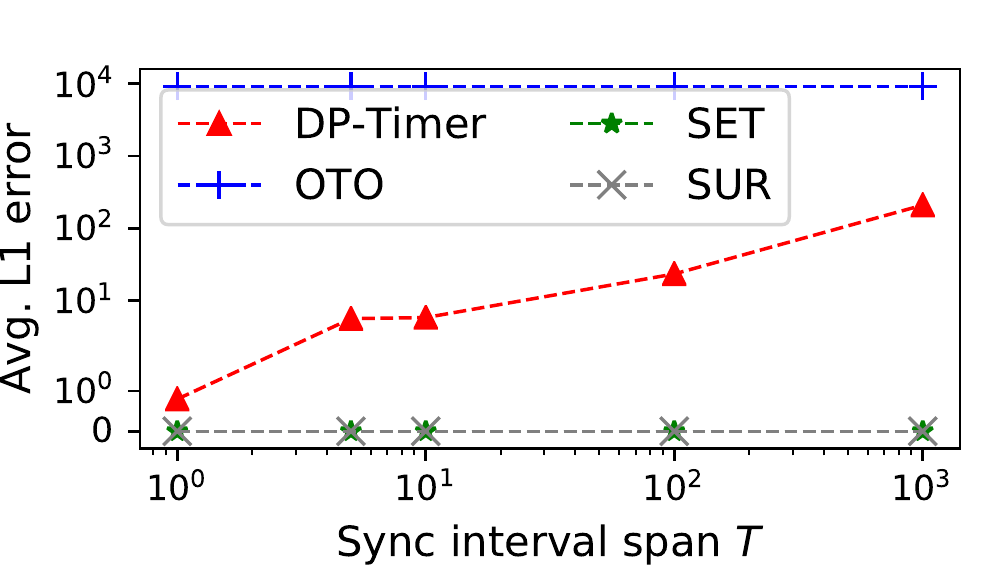}
        \caption{Avg. L1 error v.s. $T$}
        \label{fig:td-acc-timer}
    \end{subfigure}%%
    \begin{subfigure}[b]{0.49\linewidth}
    \centering \includegraphics[width=1\linewidth]{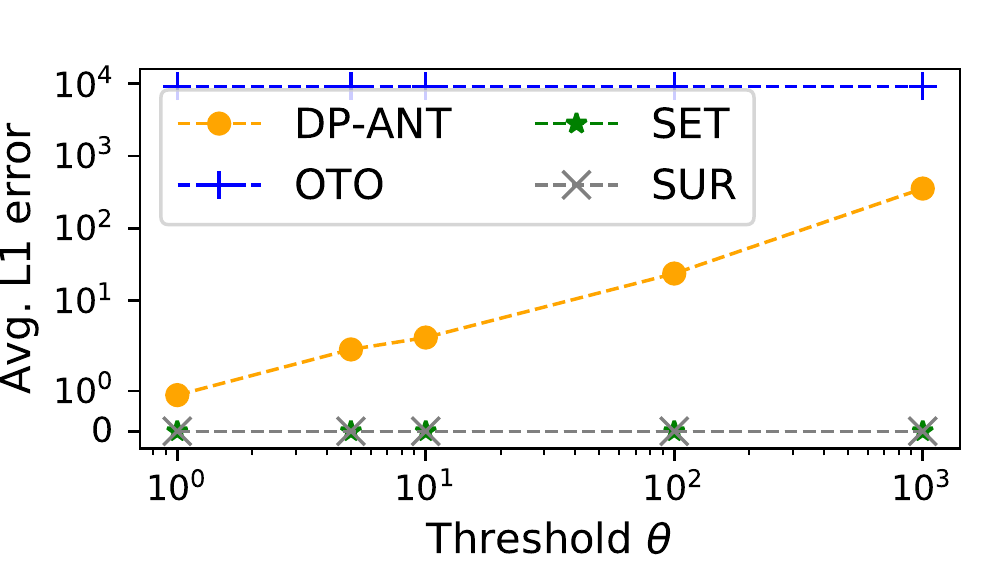}   
 \caption{Avg. L1 error v.s. $\theta$}
        \label{fig:td-acc-sparse}\end{subfigure}%%
        \newline
    \begin{subfigure}[b]{0.49\linewidth}
        \centering
         \includegraphics[width=1\linewidth]{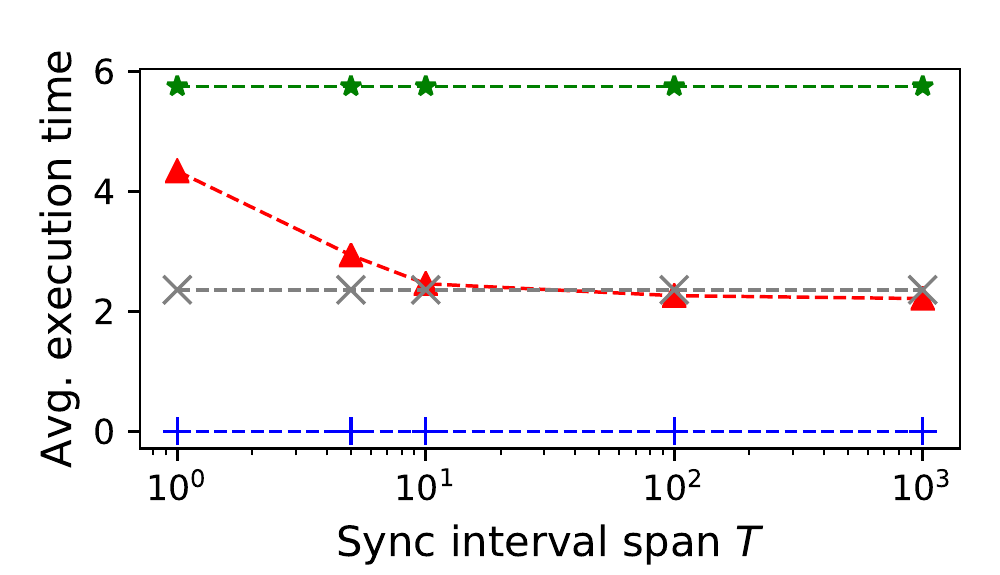}
        \caption{Avg. execution time v.s. $T$}
        \label{fig:td-perf-timer}
    \end{subfigure}%%
    \begin{subfigure}[b]{0.49\linewidth}
    \centering \includegraphics[width=1\linewidth]{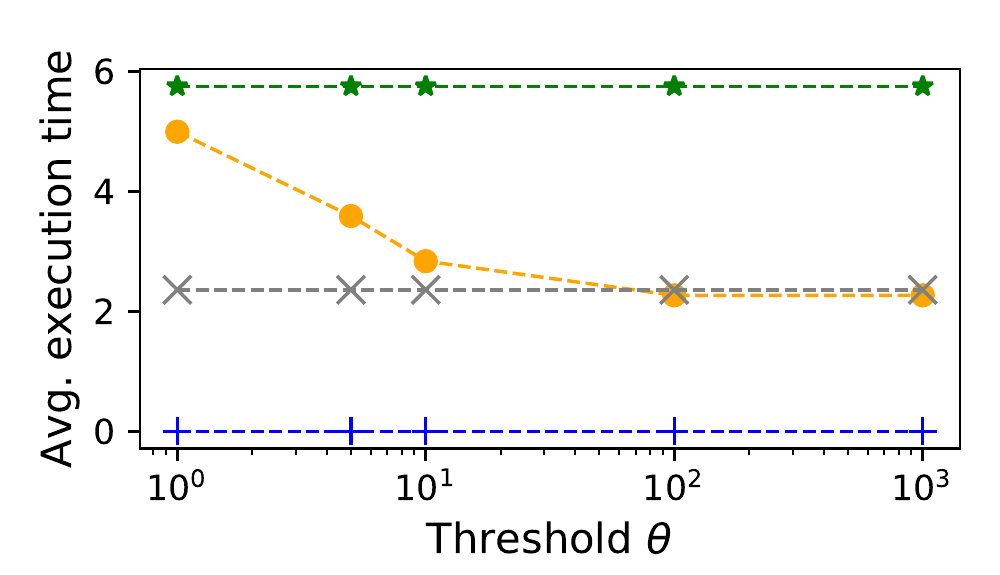}   
 \caption{Avg. execution time v.s. $\theta$}
        \label{fig:td-perf-sparse}\end{subfigure}%%
    \vspace{-3mm}
   \caption{Trade-off with non-privacy parameters}
   \label{fig:trade-np}
   \vspace{-2mm}
\end{figure}

\boldparagraph{Observation 6. Even with fixed privacy, the DP strategies can still be tuned to obtain different performance or accuracy by adjusting non-privacy parameters.} From Figure~\ref{fig:td-acc-timer} and~\ref{fig:td-acc-sparse}, we observe that the mean query errors for both methods increase when $T$ or $\theta$ increases. This is because once $T$ or $\theta$ is increased, the \user waits longer before making an update, which increases the logical gap. Also Figure~\ref{fig:td-perf-timer} and~\ref{fig:td-perf-sparse} shows that the performance metric decreases as $T$ or $\theta$ increases. This is because as $T$ or $\theta$ increases, the \user updates less frequently, which reduces the number of dummy records that could be injected into the outsourced database.  

\eat{
\begin{figure}[ht]
    \begin{subfigure}[b]{0.49\linewidth}
        \centering
         \includegraphics[width=1\linewidth]{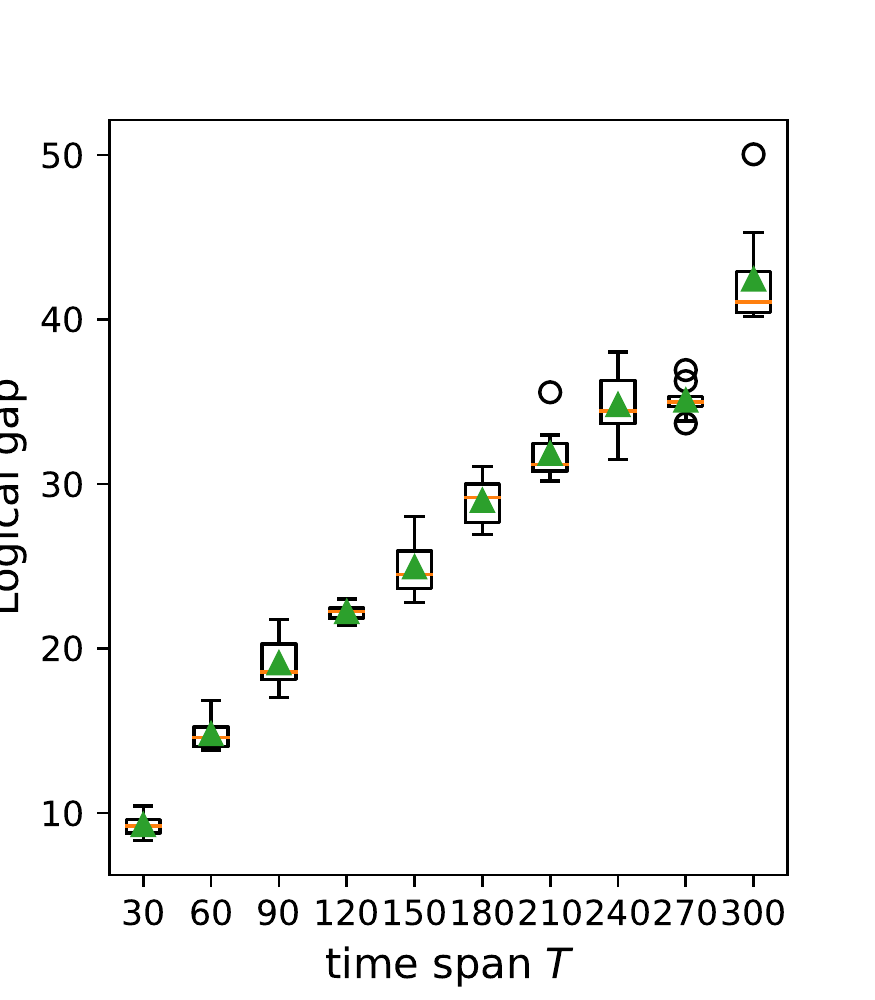}
        \caption{DP-Timer: logical gap }
        \label{fig:td-acc-timer}
    \end{subfigure}%%
    \begin{subfigure}[b]{0.46\linewidth}
    \centering \includegraphics[width=1\linewidth]{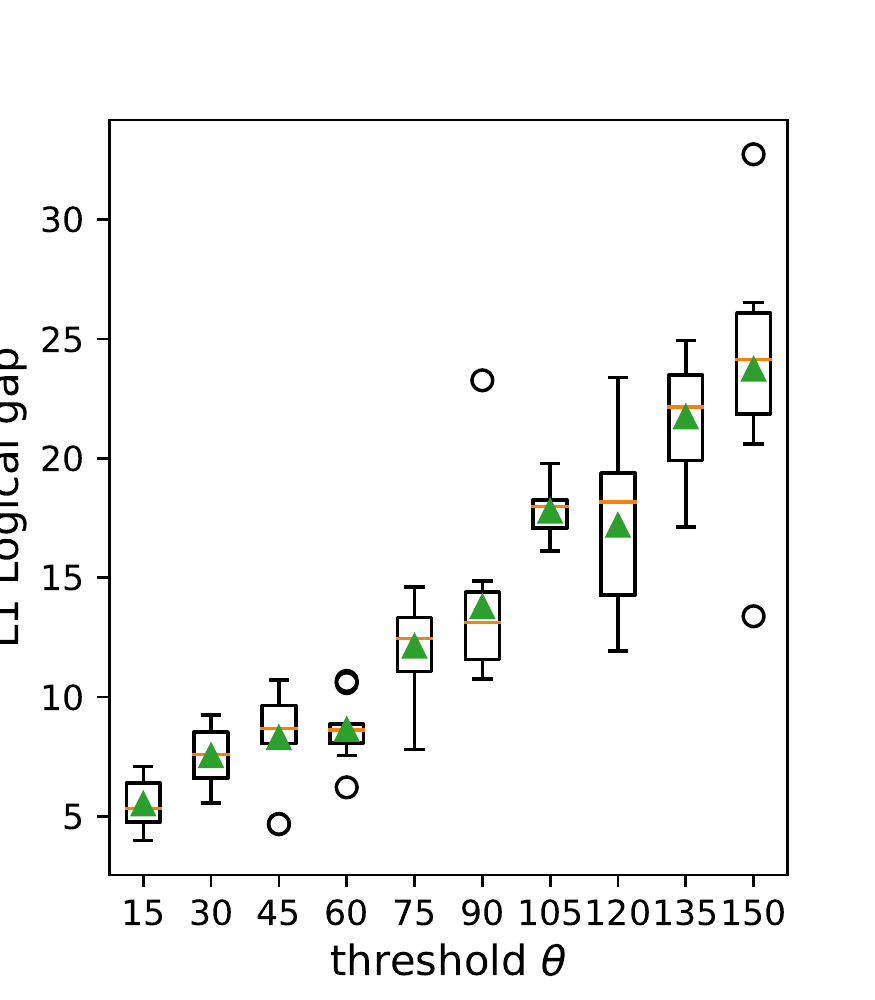}   
 \caption{DP-ANT: logical gap }
    \label{fig:td-acc-sparse}\end{subfigure}%%
    \newline
    \begin{subfigure}[b]{0.46\linewidth}
    \centering    \includegraphics[width=1\linewidth]{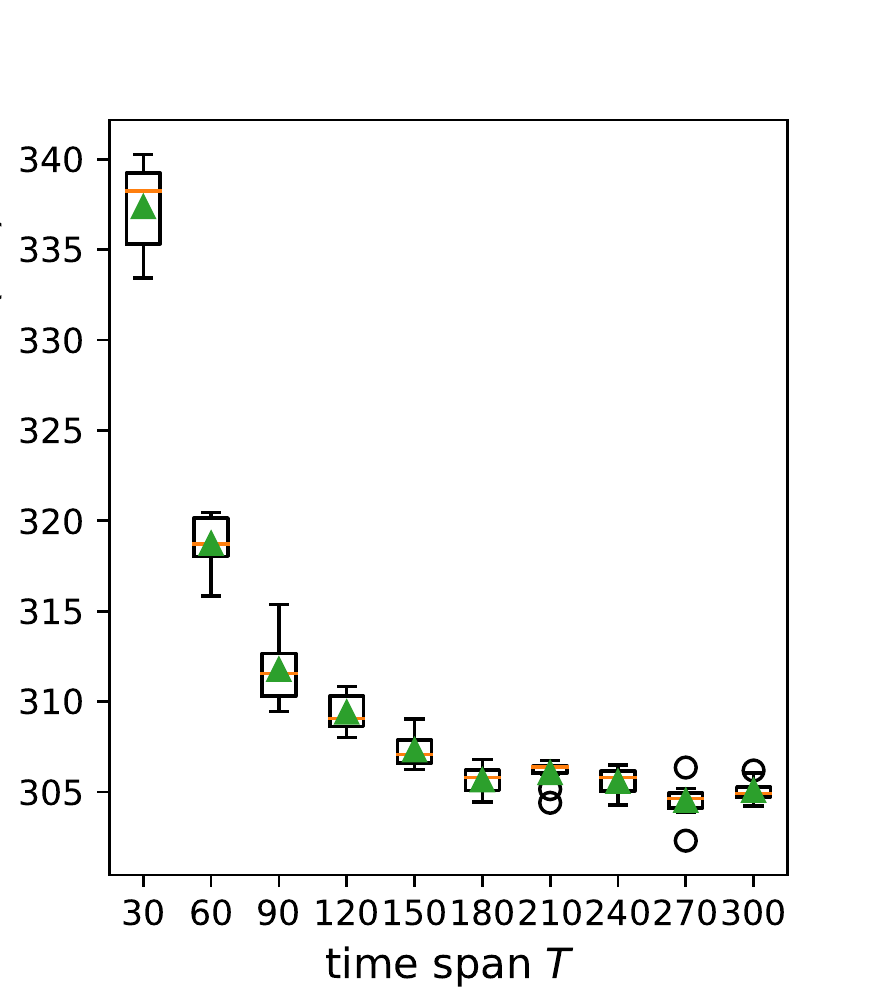}
        \caption{DP-Timer: outsourced size }
        \label{fig:td-perf-timer}\end{subfigure}%%
      \begin{subfigure}[b]{0.46\linewidth}
    \centering    \includegraphics[width=1\linewidth]{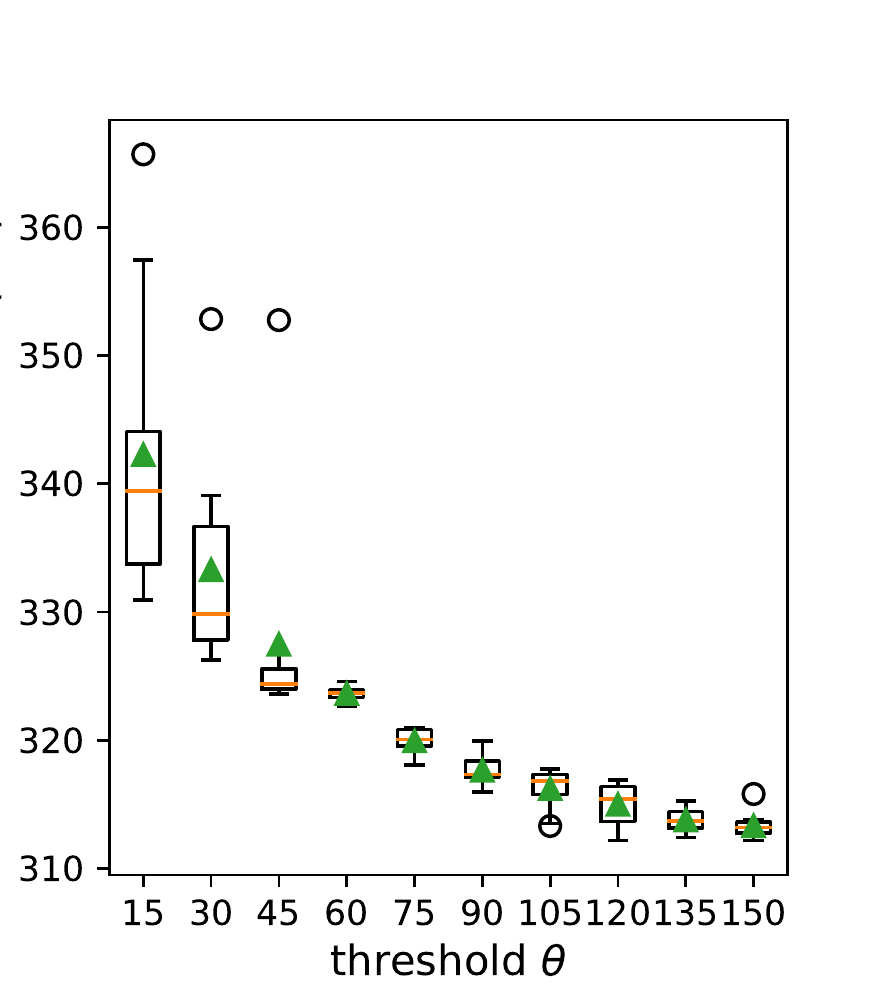}  
        \caption{DP-ANT: outsourced size}
        \label{fig:td-perf-sparse}
    \end{subfigure}
    \vspace{-4mm}
   \caption{Experiment results with changing $T$ and $\theta$}
   \label{trade-off-np}
\end{figure}
}

\vspace{-1mm}
\section{Related Work}
\label{sec:related}
\noindent{\bf Encrypted databases and their leakage.} Encrypted databases is a broadly studied research topic. Existing solutions utilize techniques such as bucketization~\cite{hacigumucs2002executing, hore2012secure, hore2004privacy}, predicate encryption~\cite{shi2007multi,lu2012privacy}, oblivious RAM~\cite{bater2017smcql, crooks2018obladi, demertzis2020seal, naveed2014dynamic, ishai2016private}, structural encryption and symmetric searchable encryption (SSE)~\cite{curtmola2011searchable, stefanov2014practical, cash2014dynamic, kamara2012dynamic, kamara2018sql, kamara2019computationally, patel2019mitigating, ghareh2018new, amjad2019forward},  functional encryption~\cite{boneh2004public, shen2009predicate}, property-preserving encryption~\cite{agrawal2004order,bellare2007deterministic,boldyreva2009order, pandey2012property},  order-preserving encryption~\cite{agrawal2004order, boldyreva2011order}, trusted execution environments~\cite{priebe2018enclavedb, eskandarian2017oblidb, vinayagamurthy2019stealthdb} and homomorphic encryption~\cite{gentry2009fully, boneh2005evaluating, chowdhury2019cryptc, samanthula2014privacy}. Recent work has revealed that these methods may be subject to information leakage through query patterns~\cite{blackstone2019revisiting, cash2015leakage}, identifier patterns~\cite{blackstone2019revisiting}, access patterns~\cite{cash2015leakage, kellaris2016generic, dautrich2013compromising} and query response volume~\cite{kellaris2016generic, blackstone2019revisiting, grubbs2019learning, grubbs2018pump, gui2019encrypted}. In contrast, our work analyzes information leakage for encrypted databases through update patterns. Recent work on backward private SSE~\cite{bost2017forward, ghareh2018new, sun2018practical, amjad2019forward}, which proposes search (query) protocols that guarantee limits on information revealed through data update history, shares some similarity with our work. However, this approach is distinct from our work as they hide the update history from the query protocol. Moreover, backward private SSE permits insertion pattern leakage, revealing how many and when records have been inserted. In contrast, our work hides insertion pattern leakage through DP guarantees. Similar to our work, Obladi~\cite{crooks2018obladi} supports updates on top of outsourced encrypted databases. However, it focuses on ACID properties for OLTP workloads and provides no accuracy guarantees for the analytics queries. 

%\kartik{there are a couple of works by Roxana Geambasu and Obladi that may have some relationship to update pattern leakage}
%\nt{I found these papers about backward private SSE, not sure if I found the correct work to cite.}\kartik{one is this: https://arxiv.org/pdf/1909.01502.pdf}
%\kartik{not all of the following are ``solutions'' for encrypted databases. Perhaps say something like, existing encrypted database \cite{encrypted database solutions} solutions use the techniques to reduce leakage.}
\vspace{-1mm}
\boldparagraph{Differentially-private leakage}. The concept of DP leakage for encrypted databases was first introduced by Kellaris et al.~\cite{kellaris2017accessing}. Interesting work has been done on DP access patterns~\cite{bater2018shrinkwrap, mazloom2018secure, chen2018differentially, wagh2018differentially}, DP query volume~\cite{patel2019mitigating} and DP query answering on encrypted data~\cite{chowdhury2019cryptc, agarwal2019encrypted, lecuyer2019sage}. However, most of this work focuses on the static database setting. \re{Agarwal et al.~\cite{agarwal2019encrypted} consider the problem of answering differentially-private queries over encrypted databases with updates. However, their work focuses mainly on safeguarding the query results from revealing sensitive information, rather than protecting the update leakage. L\'{e}cuyer et al.~\cite{lecuyer2019sage} investigate the method to privately update an ML model with growing training data. Their work ensures the adversary can not obtain useful information against the newly added training data by continually observing the model outputs. However, they do not consider how to prevent update pattern leakage. Kellaris et al.~\cite{kellaris2017accessing} mention distorting update record size by adding dummy records, but their approach always overcounts the number of records in each update, which incorporates large number of dummy records. Moreover, their main contribution is to protect the access pattern of encrypted databases rather than hiding the update patterns. In addition, none of these approaches formally defined the update pattern as well as it's corresponding privacy, and none of them have considered designing private synchronization strategies.
}

%\nt{SOLVED:So basically, these methods use truncated Laplace mechanism to calculate how much dummy records to be added, but it turns out that when use truncated LP mechanism, it achieves $(\epsilon,\delta)-DP$, we use normal Laplace mechanism and ensures $\epsilon-DP$. Not sure if we should mention that, may need to check with Ashwin.} \johes{I'd say go ahead and add it. }
\vspace{-1mm}
\boldparagraph{Differential privacy under continual observation}. The problem of differential privacy under continual observation was first introduced by Dwork et al. in \cite{dwork2010differential}, and has been studied in many recent works~\cite{dwork2010newdifferential, chen2017pegasus, bolot2013private, zhang2019statistical, cummings2018differential}. These approaches focus on designing DP streaming algorithms and are not specific to outsourced databases. In particular, although ~\cite{cummings2018differential} analyzes privacy for growing databases, unlike our work, their model assumes that the \cs has full access to all outsourced data. 

%These works mostly focus on the design of differential privacy algorithms applicable to continuous data or data streams, but pose less interest on processing or updating outsourced data.\kartik{do not understand previous statement} ~\cite{cummings2018differential} considered the privacy for growing database, while in their model, the \cs has access to all data.
\vspace{-3mm}
\section{Conclusion}
In this paper, we have introduced a new type of leakage associated with modern encrypted databases called update pattern leakage. We formalize the definition and security model of \system with DP update patterns. We also proposed the framework \appsystem, which extends existing encrypted database schemes to \system with DP update patterns. \appsystem guarantees that the entire data update history over the outsourced data structure is protected by differential privacy. This is achieved by imposing differentially-private strategies that dictate the \user's synchronization of local~data. 
%\re{The proposed strategies can be easily extended to multiple-table cases if all tables are considered uncorrelated (inserting one record in a table does not affect other tables). The data owner can insert records to each table independently, and for each table, maintain an independent synchronization algorithm. With parallel composition, the entire mechanism still satisfies $\epsilon$-DP. Supporting multi-table updates that have constraints like foreign keys would require secure protocols for computing sensitivity~\cite{johnson2018towards} and data truncation~\cite{kotsogiannis2019privatesql}. We consider this as an interesting direction for future work. }

\re{Note that \appsystem currently only supports single table schema. Supporting multi-relational table databases require additional security models, data truncation techniques~\cite{kotsogiannis2019privatesql} and secure protocols to compute the sensitivity~\cite{johnson2018towards} over multiple tables. We leave the design of these techniques for future work.} 

%Schemas where an update alters a subset of tables require additional obfuscation, such as dummy records, such that seeing which tables are altered does not reveal information about the update. 

%\nt{So shall we just say we support "single table" database but not "multi-relational table" database? It is a little bit confusing to say the schema where an update alters a subset of tables.} \nt{In addition to obfuscation techniques, I believe to support multi-relational table database, it also requires an independent secure protocol to compute the sensitivity~\cite{johnson2018towards} when cover multiple tables as well as some data truncation techniques~\cite{kotsogiannis2019privatesql}}

%\kartik{as mentioned in one of my previous comments, currently we do not mention what our leakage is.} 

%%
%% The acknowledgments section is defined using the "acks" environment
%% (and NOT an unnumbered section). This ensures the proper
%% identification of the section in the article metadata, and the
%% consistent spelling of the heading.
\vspace{-3mm}
\begin{acks}
This work was supported by the National Science Foundation under grants 2016393, 2029853; and by DARPA and SPAWAR under contract N66001-15-C-4067.
\end{acks}

%%
%% The next two lines define the bibliography style to be used, and
%% the bibliography file.
\bibliographystyle{ACM-Reference-Format}
\bibliography{ref}

%%
%% If your work has an appendix, this is the place to put it.

\appendix
%\subsection{Query Processing \& Answering}
 %In our design, the $\cs$ supports a rich set of secure queries programs expressed in terms of a small collection of secure operators. We discuss the implementations of these operators and illustrate how to translate the user specified queries by taking the use of these operators. Let us consider a local database $\mathcal{D}$ with schema $<A_1, A_2, ..., A_l>$, then the corresponding outsourced data structure $\mathcal{DS}$ is an encrypted table where each attribute is represented as it's encrypted on-hot-encoding form. 
 
\section{Security Model Continued}\label{sec:security-continued}
We describe the detailed security model in this section. Recall that we consider the security and privacy of the \user against a semi-honest \cs. To keep the definitions general, we extend the adaptive security definition in~\cite{curtmola2011searchable}. We now describe the following experiments:
%\johes{Marked.differential privacy is not defined yet}.
%\kartik{SOLVED:say analyst is trusted?}

%\am{you can drop the itemize}
\begin{itemize}
    \item {\bf Real}$_{\mathcal{A}}^{\Sigma}(\lambda)$: The environment $\mathcal{Z}$ samples a growing database via $\mathcal{D} \gets \mathsf{DBGen}$ and sends the challenger a ``{\it setup}'' message with $\mathcal{D}_0$. The challenger then runs the $\setup$ protocol with the adversary $\mathcal{A}$. Next, until $\mathcal{A}$ halts and outputs a bit, repeat: In each time step, $\mathcal{Z}$ sends a logical update  value to the challenger according to $\mathcal{D}$. The challenger decides whether to update based on its $\sync$ algorithm. If yes, the challenger evaluates $\update$ jointly with $\mathcal{A}$. Either after $\update$ terminates, or the challenger decides not to update, $\mathcal{A}$ adaptively chooses a polynomial number of queries $q$ and evaluate $\query$ accordingly, over the most recent outsourced structure revealed to $\mathcal{A}$. Finally, the adversary outputs a bit at the end of the experiment.\\
    \item {\bf Ideal}$_{\mathcal{A, S}}^{\Sigma}(\lambda)$: The environment $\mathcal{Z}$ samples a growing database via $\mathcal{D} \gets \mathsf{DBGen}$ and sends the simulator $\mathcal{S}$ with leakage $\lsetup$. The simulator $\mathcal{S}$ then simulates and reveals an output to $\mathcal{A}$. Next, until $\mathcal{A}$ halts and outputs a bit, repeat: In each time step, $\mathcal{S}$ is given the leakage $\lupdate$ from $\mathcal{Z}$, and will decide whether to simulate a result to $\mathcal{A}$ or do nothing based on $\update$. Either after $\mathcal{A}$ receives a simulated result or $\mathcal{S}$ decides to do nothing, %\kartik{SOLVED:previous statement is unclear.} 
    $\mathcal{A}$ adaptively chooses and sends a polynomial number of queries $q$ to $\mathcal{Z}$, $\mathcal{Z}$ then notifies $\mathcal{S}$ with leakage profile $\lquery(q)$. The simulator $\mathcal{S}$ simulates the outputs to $\mathcal{A}$ with input of $\lquery(q)$. Finally, the adversary outputs a bit at the end of the experiment.
\end{itemize}

\begin{definition}\label{def:sogdb}
Let $\Sigma$ = (\sync, $\setup$, $\update$, $\query$), given a security parameter $\lambda$, a stateful PPT adversary $\mathcal{A}$, a stateful simulator $\mathcal{S}$, and stateful leakage functions $\mathcal{L}=(\lsetup, \lupdate, \lquery)$. We say that $\Sigma$ is $\mathcal{L}$-adaptively-secure if there exists a PPT simulator $\mathcal{S}$ such that for all PPT adversary $\mathcal{A}$, if the advantage of $\mathcal{A}$ satisfies:
$$\mid\textup{Pr}[{\bf Real}_{\mathcal{A}}^{\Sigma}(\lambda) = 1] - \textup{Pr}[{\bf Ideal}_{\mathcal{A, S}}^{\Sigma}(\lambda) = 1]\mid \leq \textup{negl}(\lambda)$$

\end{definition}

Definition~\ref{def:sogdb} says that if $\Sigma$ is $\mathcal{L}$-adaptively-secure, it does not leak anything beyond leakage profile $\mathcal{L}=(\lsetup, \lupdate, \lquery)$. In what follows, we introduce the definitions of $\mathcal{L}$-adaptively-secure \system with DP update patterns. 

\begin{definition}[$\mathcal{L}$-adaptively-secure \appsystem ~/w DP update pattern]\label{def:lsecdp} Given a $\mathcal{L}$-adaptively-secure system $\Sigma$, and let $\lupdate$ to be the stateful update leakage for $\Sigma$. The \system $\Sigma$ is said to have differentially-private (DP) update pattern if $\lupdate$ can be written as:
$$\lupdate(\mathcal{D}) = \mathcal{L}'\left(\upatt(\Sigma, \mathcal{D})\right)$$
where $\mathcal{L}'$ is stateless, and for any two neighboring growing databases $\mathcal{D}$ and $\mathcal{D}'$, and any $O\subseteq\mathcal{O}$, where $\mathcal{O}$ is the range of all possible update pattern, it satisfy that:
$$\textup{Pr}\left[\lupdate(\mathcal{D})\in O \right] \leq e^{\epsilon}\cdot\textup{Pr}\left[\lupdate(\mathcal{D}') \in O\right]$$
\end{definition}
%\kartik{is defn 4 trying to talk about the entire system or only updates? If we never talk about leakage for setup and querying, should this talk about the security of $\Sigma$?}
%\nt{def 4 only says that the update pattern is in DP, but the rest leakage may not bounded in DP, such as query response volume, so it's actually only talks about the update pattern.}
%\kartik{what about setup? it seems you perturb and then invoke setup protocol?} \nt{Yes, the initial outsource also provide DP guarantee. Can discuss how to formulate it.}

%\kartik{Definition~\ref{def:sogdb} says that if $\Sigma$ is $(\lsetup, \lupdate, \lquery)$-secure, it does not leak anything beyond leakage profile $(\lsetup, \lupdate, \lquery)$. } 

In what follows, we provide the complete security analysis that shows the \appsystem we proposed satisfies Definition~\ref{def:lsecdp}. Recall that \appsystem have the constraints the underlying encrypted database to have update leakage that is a stateless function of $\upatt$. 

\begin{theorem}
Given \appsystem $ = (\sync, \sigma, \mathsf{edb})$, where $\mathsf{edb}=$ $(\mathsf{Setup}$, $\mathsf{Update}$, $\mathsf{Query})$ is an encrypted database scheme with stateful leakages profiles $\mathcal{L}_{s}^{\mathsf{edb}}, \mathcal{L}_{u}^{\mathsf{edb}}$ and $\mathcal{L}_{q}^{\mathsf{edb}}$. If $\mathsf{edb}$ is ($\mathcal{L}_{s}^{\mathsf{edb}}, \mathcal{L}_{u}^{\mathsf{edb}}$, $\mathcal{L}_{q}^{\mathsf{edb}}$ )-adaptively-secure~\cite{curtmola2011searchable}, and  $\mathcal{L}_{u}^{\mathsf{edb}} = \mathcal{L}'(\upatt(\mathsf{edb}, \cdot))$, where $\mathcal{L}'$ is a function. Let $\lsetup$, $\lupdate$, and $\lquery$ corresponds to the leakage profiles associated with the given \appsystem, and let $\lsetup = \mathcal{L}_{s}^{\mathsf{edb}}$, $\lupdate = \upatt_{t}$, and $\lquery = \mathcal{L}_{q}^{\mathsf{edb}}$,  then \appsystem is ($\lsetup$, $\lupdate$, $\lquery$)-secure (Definition~\ref{def:sogdb}). 
%Then \appsystem satisfies Definition~\ref{def:sogdb}. \kartik{I did not find CQA security in [47]. Also, when you say satisfies Definition 2, you should mention the corresponding leakage tuple of the resulting system. Also, we only perform appends/inserts, right? So, perhaps we don't need op.}
\end{theorem}
\begin{proof} We describe a polynomial time simulator $\mathcal{S}$ such that the advantage of any p.p.t. adversary $\mathcal{A}$ to distinguish the output between ${\bf Ideal}^{\appsystem}$ and ${\bf Real}^{\appsystem}$ is negligible. Since $\mathsf{edb}$ is ($\mathcal{L}_{s}^{\mathsf{edb}}, \mathcal{L}_{u}^{\mathsf{edb}}$, $\mathcal{L}_{q}^{\mathsf{edb}}$ )-adaptively-secure, there must exist a p.p.t simulator $\mathcal{S}^{\mathsf{edb}}$ s.t. the outputs of $\mathcal{S}^{\mathsf{edb}}$ with inputs $\mathcal{L}_{s}^{\mathsf{edb}}, \mathcal{L}_{u}^{\mathsf{edb}}$, $\mathcal{L}_{q}^{\mathsf{edb}}$ are computationally indistinguishable from the results produced from real protocols $\mathsf{Setup}$, $\mathsf{Update}$, $\mathsf{Query}$, respectively. We define the leakage of \appsystem as  $\lsetup = \mathcal{L}_{s}^{\mathsf{edb}}$, $\lquery = \mathcal{L}_{q}^{\mathsf{edb}}$, $\lupdate = \upatt_{t}$,  and we build a simulator $\mathcal{S}$ in ${\bf Ideal}^{\appsystem}$ as follows: If $\mathcal{S}$ receives $\lsetup$ or $\lquery$ from the environment, then it outputs the results of $\mathcal{S}^{\mathsf{edb}}(\lsetup)$ or $\mathcal{S}^{\mathsf{edb}}(\lquery)$, respectively. If the simulator receives $\lupdate$ at time $t$, then it first checks if $(t, |\gamma_t|) \in \upatt$. If yes, then it outputs $\mathcal{S}^{\mathsf{edb}}(\mathcal{L}_{u}^{\mathsf{edb}}(|\gamma_t|))$, otherwise it does nothing.
According to the above description, the indistinguishability of the simulated outputs and the real protocols outputs follow the adaptively-secure property of $\mathsf{edb}$. Therefore, the probability of any p.p.t. adversary $\mathcal{A}$ to distinguish between the real and the ideal experiment with the aforementioned $\mathcal{S}$, is negligible.
%\kartik{our setup can have a differentially-private leakage, right?} %\kartik{SOLVED:The reason why you do not mention $\mathcal{L}_{u}^{\mathsf{edb}}$ is because $\upatt_{t}$ already includes it?} 
\end{proof}
Next, we prove that \appsystem that implemented with proposed DP strategies satisfies Definition~\ref{def:lsecdp}. To capture the update pattern leakage, we rewrite the DP algorithms to output the total number of synchronized records at each update, instead of signaling the update protocol. The rewritten mechanisms simulate the update pattern when applying the DP strategies.

\begin{table}[]
\scalebox{0.98}{\small
\begin{tabular}{|rl|}
\hline
                                 & \multicolumn{1}{c|}{{\bf $\mathcal{M}_{\mathsf{timer}}(\mathcal{D},  \epsilon, f, s, T)$}}\\
$\mathcal{M}_{\mathsf{setup}}$:  & {\bf output} $\left(0, |\mathcal{D}_0| + \lap(\frac{1}{\epsilon}) \right)$ \\
$\mathcal{M}_{\mathsf{update}}$: & $\forall i \in \mathbb{N}^{+}$, run $\mathcal{M}_{\mathsf{unit}}(U[i\cdot T, (i+1)T], \epsilon, T)$\\
                                 & $\mathcal{M}_{\mathsf{unit}}$: {\bf output} $\left(i\cdot T, ~\lap(\frac{1}{\epsilon}) + \sum_{k = i\cdot T + 1}^{(i+1)T} 1 \mid u_{k}\neq\emptyset \right)$\\
$\mathcal{M}_{\mathsf{flush}}$:  & $\forall j \in \mathbb{N}^{+}$, {\bf output} $\left(j\cdot f, ~s\right)$ \\\hline
\end{tabular}
}
\caption{Mechanisms to simulate the update pattern}
\label{tab:lmech-1}
\vspace{-4mm}
\end{table}

\begin{lemma}[Sequential Composition]\label{lem:seq} $\mathcal{M}_1$ is an $\epsilon_1$-DP mechanism and $\mathcal{M}_2$ is an $\epsilon_2$-DP mechanism. The combined mechanism $\mathcal{M}_{1,2}$ with the mapping of  $\mathcal{M}_{1,2}(x) = (\mathcal{M}_1(x), \mathcal{M}_2(x))$ satisfies $(\epsilon_1 + \epsilon_2)-DP$.
\end{lemma}
\begin{proof}
Let $x$, and $y$ be two neighboring database such that $||x - y||\leq 1$, then for any output pair $(o_1, o_2)$, we have:
\begin{equation}
    \begin{split}
        \frac{\textup{Pr}\left[\mathcal{M}_{1,2}(x) = (o_1, o_2)\right]}{\textup{Pr}\left[\mathcal{M}_{1,2}(y) = (o_1, o_2)\right]} = & ~\frac{\textup{Pr}\left[\mathcal{M}_1(x) = o_1\right]\textup{Pr}\left[\mathcal{M}_2(x) = o_2\right]}{\textup{Pr}\left[\mathcal{M}_{1}(y) =  o_1\right]\textup{Pr}\left[\mathcal{M}_{2}(y) =  o_2\right]}\\
        = &~\frac{\textup{Pr}\left[\mathcal{M}_1(x) = o_1\right]}{\textup{Pr}\left[\mathcal{M}_{1}(y) =  o_1\right]}\times\frac{\textup{Pr}\left[\mathcal{M}_2(x) = o_1\right]}{\textup{Pr}\left[\mathcal{M}_{2}(y) =  o_1\right]}\\
        \leq &~ e^{\epsilon_1 + \epsilon_2}
    \end{split}
\end{equation}
\end{proof}
\begin{lemma}[Parallel Composition]\label{lem:par}$\mathcal{M}_1$ is an $\epsilon_1$-DP mechanism and $\mathcal{M}_2$ is an $\epsilon_2$-DP mechanism. The combined mechanism $\mathcal{M}_{1,2}$ with the mapping of  $\mathcal{M}_{1,2}(x,y) = (\mathcal{M}_1(x), \mathcal{M}_2(y))$ satisfies $\max(\epsilon_1, \epsilon_2)-DP$, where $x \cap y =\emptyset$
\end{lemma}
\begin{proof}
According to the parallel composition, the two mechanisms are applied on disjoint data, which means that the mechanisms are applied independently and thus privacy is constrained by the maximum privacy budget.
\end{proof}

\begin{theorem}\label{ftm:dptimer}
\appsystem implemented with the DP-Timer strategy satisfies Definition~\ref{def:lsecdp}. 
\end{theorem}
\vspace{-2mm}
\begin{proof}\label{fpf:dp-timer}
Since we have constrained that the update leakage of the given DP-Sync is a function only related to the update pattern (provided in Table~\ref{tab:lmech-1}). Thus we prove this theorem by illustrating that the composed privacy guarantee of $\mathcal{M}_{\mathsf{timer}}$ satisfies $\epsilon$-DP. The mechanism $\mathcal{M}_{\mathsf{timer}}$ is a composition of several separated mechanisms. We now analysis the privacy guarantees of each. \\

\DeclareRobustCommand{\rchi}{{\mathpalette\irchi\relax}}
\newcommand{\irchi}[2]{\raisebox{\depth}{$#1\chi$}} % inner command, used by \rchi

\noindent (1) $\mathcal{M}_{\mathsf{setup}}$.\\
Let $\rchi$ be the collection of all possible initial database, and let $\mathcal{D}^x_0 \in \rchi$, and $\mathcal{D}^y_0 \in \rchi$ are two neighboring databases that differ by addition and removal of only 1 record. We use $n_x$, $n_y$ denotes the number of records in $\mathcal{D}^x_0$ and $\mathcal{D}^y_0$, respectively. Let $z$ denote of the size of an arbitrary initial database, and let $p_x$, $p_y$ denote the output distribution of $\mathcal{M}_{\mathsf{setup}}(\mathcal{D}^x_0, \epsilon)$, and $\mathcal{M}_{\mathsf{setup}}(\mathcal{D}^y_0, \epsilon)$, respectively. We compare the two terms under arbitrary $z$:
\begin{equation}
\begin{split}
    \frac{p_x(z)}{p_y(z)} = & ~\frac{\frac{1}{2b}e^{\frac{-|n_x - z|}{b}}}{\frac{1}{2b}e^{\frac{-|n_y - z|}{b}}} = e^{\frac{|n_y - z| - |n_x - z|}{b}} \leq e^{\frac{|n_y - n_x|}{b}}
\end{split}
\end{equation}
Note that, we set $b = \frac{1}{\epsilon}$, and since $\mathcal{D}^x_0 $, and $\mathcal{D}^y_0$ are neighboiring databases, thus $|n_y - n_x| \leq 1$, therefore.
\begin{equation}
\begin{split}
 e^{\frac{|n_y - n_x|}{b}} \leq e^{\frac{1}{\frac{1}{\epsilon}}} = e^{\epsilon} \rightarrow \frac{p_x(z)}{p_y(z)} \leq e^{\epsilon}.
\end{split}
\end{equation}
Note that the ratio $\frac{p_x(z)}{p_y(z)} \geq e^{-\epsilon}$ follows by symmetry. Thus we can conclude that $\mathcal{M}_{\mathsf{setup}}$ satisfies $\epsilon$-DP.\\

\noindent(2) $\mathcal{M}_{\mathsf{unit}}$.\\
Let $\rchi'$ denote all possible logical updates within a period of time $T$. Let $U_x \in \rchi'$, and $U_y \in \rchi$, denotes two neighboring updates (differ by addition or removal of 1 logical update). We define $f = \sum_{\forall u_t \in U} 1 | u_t \neq \emptyset$, and:
\begin{equation}
    \begin{split}
        \Delta f = & \max_{\forall U_x,U_y\in\rchi' \wedge  ||U_x - U_y||_1 \leq 1}| f(U_x) - f(U_y)|
    \end{split}
\end{equation}

According to the definition, $f$ is a counting function that counts how many logical updates happened within a given $U$, and we can conclude that $\Delta f = 1$. Then, let $p_x$, $p_y$ denote the density function of $\mathcal{M}_{\mathsf{unit}}(U_x, \epsilon)$, and $\mathcal{M}_{\mathsf{unit}}(U_y, \epsilon)$, respectively. We compare the two terms under arbitrary point $z$:
\begin{equation}
\begin{split}
    \frac{p_x(z)}{p_y(z)} = & ~\frac{\frac{1}{2b}e^{\frac{-|f(U_x) - z|}{b}}}{\frac{1}{2b}e^{\frac{-|f(U_y) - z|}{b}}} = e^{\frac{|f(U_y) - z| - |f(U_x) - z|}{b}} \leq e^{\frac{|f(U_y) - f(U_x)|}{b}}
\end{split}
\end{equation}

\eat{
\begin{equation}
\begin{split}
    \frac{p_x(z)}{p_y(z)} = & ~\Pi_{i=1}^{k}\left(\frac{e^{-\epsilon(\frac{| f(U_x) - z_i|}{\Delta f})}}{e^{-\epsilon(\frac{| f(U_y) - z_i|}{\Delta f})}}\right) \\
    = & ~\Pi_{i=1}^{k}e^{\left( \epsilon(\frac{|f(U_y) - z_i| - |f(U_x) - z_i|}{\Delta f})\right)}\\
    \leq & ~\Pi_{i=1}^{k}e^{\left( \epsilon(\frac{|f(U_y) - f(U_x)|}{\Delta f})\right)} = e^{\left( \epsilon(\frac{||f(U_y) - f(U_x)||_1}{\Delta f})\right)} \leq e^{\epsilon}
\end{split}
\end{equation}}
Note that, we set $b = \frac{1}{\epsilon}$, and we know that $\Delta f = 1$, therefore.
\begin{equation}
\begin{split}
 e^{\frac{|f(U_y) - f(U_x)|}{b}} \leq e^{\frac{\Delta f}{\frac{1}{\epsilon}}} = e^{\epsilon} \rightarrow \frac{p_x(z)}{p_y(z)} \leq e^{\epsilon}.
\end{split}
\end{equation}
Note that the ratio $\frac{p_x(z)}{p_y(z)} \geq e^{-\epsilon}$ follows by symmetry. According the this, $\mathcal{M}_{\mathsf{unit}}$ satisfies $\epsilon$-DP.\\

\noindent(3) $\mathcal{M}_{\mathsf{flush}}$.\\
$\mathcal{M}_{\mathsf{flush}}$ releases a fixed number every $T$ times, thus it satisfies 0-DP.\\

We have analyzed the privacy guarantees of the three basic mechanisms, we now illustrate the analysis of the composed privacy guarantees. $\mathcal{M}_{\mathsf{update}}$ is a mechanism that repeatedly calls  $\mathcal{M}_{\mathsf{unit}}$ and applies it over disjoint data, the privacy guarantee of $\mathcal{M}_{\mathsf{unit}}$ follows Lemma~\ref{lem:par}, thus satisfying $\epsilon$-DP. The composition of $\mathcal{M}_{\mathsf{setup}}$ and $\mathcal{M}_{\mathsf{update}}$ also follows Lemma~\ref{lem:par} and the composition of $\mathcal{M}_{\mathsf{flush}}$ follows Lemma~\ref{lem:seq}. Thus the entire algorithm $\mathcal{M}_{\mathsf{timer}}$ satisfies $\left(\max(\epsilon, \epsilon) + 0 \right)$ -DP, which is $\epsilon$-DP.
\end{proof}

\begin{table}[]
\scalebox{0.98}{\small
\begin{tabular}{|rl|}
\hline
                                 & \multicolumn{1}{c|}{{\bf $\mathcal{M}_{\mathsf{ANT}}(\mathcal{D},  \epsilon, f, s, \theta)$}} \\
$\mathcal{M}_{\mathsf{setup}}$:  & {\bf output} $\left(0, |\mathcal{D}_0| + \lap(\frac{1}{\epsilon}) \right)$ \\
$\mathcal{M}_{\mathsf{update}}$: & $\epsilon_1 = \epsilon_2 = \frac{\epsilon}{2}$, repeatedly run $\mathcal{M}_{\mathsf{sparse}(\epsilon_1, \epsilon_2, \theta)}$.                                                                                                                       \\
                                 & $\mathcal{M}_{\mathsf{sparse}}$:                                                                                                                                                                                                                      \\
                                 & $\tilde{\theta} = \theta + \lap(\frac{2}{\epsilon_1})$, $t^*\gets$ last time $\mathcal{M}_{\mathsf{sparse}}$'s output $\neq \perp$.                                                                                                                   \\
                                 & $ \forall i\in\mathbb{N}^{+}$, {\bf output} $\begin{cases}    \left(t^{*}+i, c_i + \lap(\frac{1}{\epsilon_2})\right)  &  \text{if}~ v_i + c_i \geq  \tilde{\theta},\\    \perp               &  \text{otherwise}.\end{cases}$ \\
                                 & where $c_i = \sum_{k=t^*}^{t^*+i} 1 \mid u_k\neq \emptyset$, and $v_i = \lap(\frac{4}{\epsilon_1})$.                                                                                                                                                  \\
                                 & {\bf abort} the first time when output $\neq \perp$.                                                                                                                                                                                 \\
$\mathcal{M}_{\mathsf{flush}}$:  & $\forall j \in \mathbb{N}^{+}$, {\bf output} $\left(j\cdot f, ~s\right)$  \\ \hline
\end{tabular}
}
\caption{Mechanisms to simulate the update pattern}
\label{tab:lmech-2}
\vspace{-4mm}
\end{table}

\begin{theorem}
The \appsystem system implemented with the ANT strategy satisfies Definition~\ref{def:lsecdp}. 
\end{theorem}
\begin{proof}
We first provide $\mathcal{M}_{ANT}$ (Table~\ref{tab:lmech-2}) that simulates the update pattern of ANT strategy. We prove this theorem by illustrating the composed privacy guarantee of $\mathcal{M}_{\mathsf{ANT}}$ satisfies $\epsilon$-DP. 

The mechanism $\mathcal{M}_{\mathsf{ANT}}$ is a composition of several separated mechanisms. We have demonstrated $\mathcal{M}_{\mathsf{setup}}$ and $\mathcal{M}_{\mathsf{flush}}$ satisfy $\epsilon$-DP and 0-DP, respectively. We abstract the $\mathcal{M}_{\mathsf{update}}$ as a composite mechanism that repeatedly spawns $\mathcal{M}_{\mathsf{sparse}}$ on disjoint data. Hence, in what follows we show that $\mathcal{M}_{\mathsf{sparse}}$, and thus also $\mathcal{M}_{\mathsf{update}}$ (repeatedly call $\mathcal{M}_{\mathsf{sparse}}$), satisfies $\epsilon$-DP guarantee.

Assume a modified version of $\mathcal{M}_{\mathsf{sparse}}$, say $\mathcal{M'}_{\mathsf{sparse}}$, where it outputs $\top$ once the condition $v_i + c_i >\tilde{\theta}$ is satisfied,  and outputs $\bot$ for all other cases. Then the output of $\mathcal{M'}_{\mathsf{sparse}}$ can be written as $O = \{o_1, o_2,...,o_m \}$, where $\forall~ 1 \leq i < m$, $o_i = \bot$, and $o_m = \top$. Suppose that $U$ and $U'$ are the logical updates of two neighboring growing databases and we know that for all $i$, $\textup{Pr}\left[\tilde{c}_i < x \right] \leq \textup{Pr}\left[\tilde{c}'_i < x+1 \right]$ is satisfied, where $\tilde{c}_i$ and $\tilde{c}'_i$ denotes the $i^{th}$ noisy count when applying  $\mathcal{M'}_{\mathsf{sparse}}$ over $U$ and $U'$ respectively, such that:
\begin{equation}
\begin{split}
  & \textup{Pr}\left[~ \mathcal{M'}_{\mathsf{sparse}}(U) = O \right] \\
  =  \int_{-\infty}^{\infty} & \textup{Pr}\left[\tilde{\theta} = x \right]\left( \prod_{1 \leq i < m}\textup{Pr}\left[\tilde{c}_i < x\right] \right) \textup{Pr}\left[\tilde{c}_m \geq x \right]dx\\
\leq \int_{-\infty}^{\infty} & e^{\epsilon/2}\textup{Pr}\left[\tilde{\theta} = x + 1 \right]\left( \prod_{1 \leq i < m}\textup{Pr}\left[\tilde{c}'_i < x + 1\right] \right) \textup{Pr}\left[ v_m \geq x - {c}_m \right]dx\\
\leq \int_{-\infty}^{\infty} & e^{\epsilon/2}\textup{Pr}\left[\tilde{\theta} = x + 1 \right]\left( \prod_{1 \leq i < m}\textup{Pr}\left[\tilde{c}'_i < x + 1\right] \right) \\
  \times & e^{\epsilon/2} \textup{Pr}\left[ v_m + c'_m \geq x + 1 \right]dx\\
  = \int_{-\infty}^{\infty} & e^{\epsilon}\textup{Pr}\left[\tilde{\theta} = x + 1 \right]\left( \prod_{1 \leq i < m}\textup{Pr}\left[\tilde{c}'_i < x + 1\right] \right) \textup{Pr}\left[ \tilde{c}'_m \geq x + 1 \right]dx\\
  = e^{\epsilon}\textup{Pr} & [\mathcal{M'}_{\mathsf{sparse}}(U') = O]
\end{split}
\end{equation}
Thus $\mathcal{M'}_{\mathsf{sparse}}$ satisfies $\epsilon$-DP, and  $\mathcal{M}_{\mathsf{sparse}}$ is essentially a composition of a $\mathcal{M'}_{\mathsf{sparse}}$ satisfying $\frac{1}{2}\epsilon$-DP together with a Laplace mechanism with privacy parameter equal to $\frac{1}{2}\epsilon$. Hence by applying Lemma~\ref{lem:seq}, we see that $\mathcal{M}_{\mathsf{sparse}}$ satisfies $(\frac{1}{2}\epsilon + \frac{1}{2}\epsilon)-DP$. Knowing that $\mathcal{M}_{\mathsf{update}}$ runs $\mathcal{M}_{\mathsf{sparse}}$ repeatedly on disjoint data, with Lemma~\ref{lem:par}, the $\mathcal{M}_{\mathsf{update}}$ then satisfies $\epsilon$-DP. Finally, combined with $\mathcal{M}_{\mathsf{setup}}$ and $\mathcal{M}_{\mathsf{flush}}$, we conclude that $\mathcal{M}_{\mathsf{ANT}}$ satisfies $\epsilon$-DP, thus the theorem holds.
\end{proof}

\section{Query Rewriting}\label{sec:rewrite}
We discuss in this section how to use query rewriting to allow certain secure outsourced database scheme to ignore dummy records when computes the query results on relational tables. We assume that such database scheme supports fully oblivious query processing and reveals nothing about the size pattern. The query rewriting is not applicable to those databases that leaks the size pattern (i.e. how many encrypted records that matches a given query). We consider the following operators:

\noindent{\bf Filter.} $\phi({\bf T}, p)$: This operator filters the rows in ${\bf T}$ where the respectively attributes satisfy the predicate $p$. To ignore dummy records, we need to make sure that only real rows that satisfy predicate $p$ is returned. To achieve this, we rewrite the predicate $p$ as ``$p \wedge (isDummy = False)$''.

\noindent{\bf Project.} $\pi({\bf T}, A)$ This operator projects ${\bf T}$ on a subset of attributes defined by $A$. We rewrite the operator as $\pi(\phi({\bf T}, p), A)$, where the predicate $p$ is defined as ``$(isDummy = False)$''.

\noindent{\bf CrossProduct}. $\times({\bf T}, A_i, A_j)$: This operator transforms
the two attributes $A_i$ and $A_j$ in {\bf T} into a new attribute $A'$. The attribute domain of $A$ is the cross product of  $A_i$, and $A_j$. We rewrite the operator as  $\times(\phi({\bf T}, p), A_i, A_j)$, where $p$ denotes ``$(isDummy = False)$''.

\noindent{\bf GroupBy}. $\chi({\bf T}, A')$ This operator groups the rows in ${\bf T}$ into summary rows based on a set of attributes $A'$. In order to make this operator works correctly with dummy records, we need to ensure that dummy data will never get grouped with the real records. Thus we first group the entire relation into two groups based on attribute ``{\it isDummy}'', then apply $\chi({\bf T'}, A')$, where ${\bf T'}$ is the group of records where ``${isDummy=False}$''.

\noindent{\bf Join}. $\Join({\bf T}_1, {\bf T}_2, c)$: This operator combines columns from one or more relations by using specified values, $c$, that is common to each. We require that real data can not be joined with dummy ones, thus we rewrite the operator as $\Join(\phi({\bf T}_1, p), \phi({\bf T}_2, p), c)$, where $p$ denotes ``$(isDummy = False)$''.

\section{Theorem Proofs}\label{sec:proofs}
We provide in this section theoretical analysis and formal proofs with respect to the key theorems we provided in our paper. 

\subsection{Proof of Theorem~\ref{lg:timer}}
\begin{lemma}\label{lemma:sum}
Given $Y_1, Y_2,...,Y_k$ are $k$ independent and identically distributed Laplace random variables, with distribution $\textup{Lap}(b)$. Let $Y = \sum_i^k Y_i$, and $0 < \alpha \leq kb$, then 
$$\textup{Pr}\left[~ Y \geq \alpha \right] \leq e^{\left( \frac{-\alpha^2}{4kb^2} \right)}$$
\end{lemma}
\begin{proof}
The moment generating function of Laplace random variables $Y_i$ can be denoted as $\mathbb{E}\left[e^{(tY_i)}\right] = 1/(1 - t^2b^2)$. As for $0< x < \frac{1}{2}$, we have $(1-x)^{-1} \leq e^{2x}$. Thus $\mathbb{E}\left[e^{(tY_i)}\right] \leq e^{(2t^2b^2)}$, when $|t| < \frac{1}{2b}$. As $\alpha < kb$, let $t = \frac{\alpha}{2kb^2} < \frac{1}{2b}$ then:
\begin{equation}
\begin{split}
  \textup{Pr}\left[~ Y \geq \alpha \right] & = \textup{Pr}\left[~\mathbb{E}\left[e^{(tY)}\right]  \leq \mathbb{E}\left[e^{(t\alpha)}\right] \right]\\
 & \leq e^{(-t\alpha)}\mathbb{E}\left[e^{(tY)}\right] ~~ (Chernoff ~ bound)\\
 & = e^{(-t\alpha)}\prod_i^k\mathbb{E}\left[e^{(tY_i)}\right]\\
 & \leq e^{(-t\alpha ~+~ kt^2b^2)} = e^{(\frac{-\alpha^2}{4kb^2})}
\end{split}
\end{equation}
\end{proof}

\begin{corollary}\label{col:sum}
Given $Y_1, Y_2,...,Y_k$ be $k$ i.i.d. Laplace random variables with distribution $\textup{Lap}(b)$. Let $Y=\sum_{i=1}^k Y_i$, and $\beta \in (0,1)$, the following inequality holds
$$\textup{Pr}\left[~ Y \geq 2b\sqrt{k\log{\frac{1}{\beta}}} ~\right] \leq \beta $$
\end{corollary}
\begin{proof}
Continue with the proof of lemma 4.1. Let $e^{(\frac{-\alpha^2}{4kb^2})} = {\beta}$, then $\alpha = 2b\sqrt{k\log{\frac{1}{\beta}}}$, when $k > 4\log{\frac{1}{\beta}}$ the corollary holds.
\end{proof}

\begin{corollary}\label{col:sum-2}
Given $Y_1, Y_2,...,Y_k$ be $k$ i.i.d. Laplace random variables with distribution $\textup{Lap}(b)$. Let $S_j \gets \sum_{i=0}^j Y_j$, where $0<j\leq k$, and $\beta \in (0,1)$, the following inequality holds
$$\textup{Pr}\left[~ \max_{0<j\leq k} S_j \geq 2b\sqrt{k\log{\frac{1}{\beta}}} ~\right] \leq \beta $$
\end{corollary}
\begin{proof}
According to Corollary~\ref{col:sum}, for any $S_j$, it satisfies that $\textup{Pr}\left[~ S_j \geq 2b\sqrt{j\log{{1}/{\beta}}} ~\right] \leq \beta $.  Since $0 < j \leq k$, then $\frac{\sqrt{j\log{{1}/{\beta}}}} {\sqrt{k\log{{1}/{\beta}}}} \leq 1$, and thus $\forall S_j$, $\textup{Pr}\left[~ S_j \geq 2b\sqrt{k\log{{1}/{\beta}}} ~\right] \leq \beta $ holds as well.
\end{proof}

\begin{proof}({\bf Theorem~\ref{lg:timer}}) According to Algorithm~\ref{algo:timer}, the local cache size direct reflects the logical gap. Thus we prove this theorem by providing the local cache boundaries. Let $W_k$ denote the local cache size after completing $k^{th}$ updates, and $W_0 = 0$. Let $C_k$ denote the number of records received between $k-1^{th}$ and $k^{th}$ updates, and $\tilde{C}_{k}$ as the number of records read from the local cache at $k^{th}$ update. Then we derive the following recursion forms of local cache size:
\begin{equation}\label{eq:lind}
 W_{k} \gets ( W_{k-1} + C_k - \tilde{C}_k )^{+} = ( W_{k-1} - Y_k )^{+}
\end{equation}
where $Y_k$ is the Laplace noise used to distort the true read count at $k^{th}$ update, and the term $(x)^+$ equals to $\max(0,x)$. Note that Equation~\ref{eq:lind} is a Lindley type recursion. Thus if we set $S_{j} \gets Y_1 + Y_2 + ... Y_j$, for all $0<j\leq k$, we have 
\begin{equation}\label{eq:lind-2}
W_k \overset{d}{=} \max_{0<j\leq k} (S_k - S_j) \overset{d}{=} \max_{0<j\leq k} S_j
\end{equation}
where $\overset{d}{=}$ means equality in distribution. Combining Equation~\ref{eq:lind-2} and Corollary~\ref{col:sum-2}, we obtain that with probability at most $\beta$
$$\textup{Pr}\left[~ W_k \overset{d}{=} \max_{0<j\leq k} S_j \geq 2b\sqrt{k\log{{1}/{\beta}}} ~\right]$$
the theorem thus holds.
\end{proof}

\subsection{Proof of Theorem~\ref{sz:timer}}
\begin{proof}
For each time $t$, the total number of outsourced data can be written as:
\begin{equation}
\begin{split}
  |\mathcal{DS}_t| & = \sum_{i=0,1,2...} c_{i+1}^{(i+1)T} + \sum_{j = 1}^{\floor*{t/T}} Y_j + s\floor*{t/f}\\
  & = |\mathcal{D}_t| + s\floor*{t/f} + \sum_{j = 1}^{k} Y_j
\end{split}
\end{equation}

where $Y_j$ is the Laplace noise drawn at each synchronization, and $k$ is the number of total updates been posted so far, thus by applying lemma~\ref{lemma:sum}, we conclude that for any time $t > 4T\log{\frac{1}{\beta}}$, it satisfies that with probability at least $1 - \beta$,  $|\mathcal{DS}_t|$ is bounded by $|\mathcal{D}_t| + s\floor*{t/f} + \frac{2}{\epsilon}\sqrt{k\log{\frac{1}{\beta}}} $, thus the theorem holds.
\end{proof}

\subsection{Proof of Theorem~\ref{lg:ant}}
\eat{
\begin{theorem}
Given timer $\eta$ and an update stream $U$ length of $T$ with $m$ records arrived after initialization, where $m \leq T$. For $\beta \in (0,1)$, when $\alpha \geq \frac{\Delta}{\epsilon}\sqrt{8(T/\eta)\ln(2/\beta)}$ the total number of records $P(U, \eta)$ uploaded after processing stream $U$ satisfy:

$$\textup{Pr}\left[~ (P(U,\eta) \geq m + \alpha ) \cup (P(U,\eta) \leq m - \alpha ) \right] \leq \beta$$

\end{theorem}

\begin{proof}
According to the timer algorithm, we use $S_i$ and $\tilde{S_i}$ denote the total number of records arrived between $i-1$-th and $i$-th synchronizations and the number of records uploaded at the $i$-th synchronization. Let $Y_i = \tilde{S_i} - S_i = \textup{Lap}(\frac{\Delta}{\epsilon})$, and $Y = \sum_{i=1}^{k} Y_i$, where $k = T/\eta$, then

\begin{equation}
\begin{split}
  & ~~\textup{Pr}\left[~ |Y| \geq b\sqrt{8k\ln(2/\beta)} ~\right] \leq \beta\\
  \Rightarrow & ~~\textup{Pr}\left[~ |\sum_i^{k}\tilde{S_i} - S_i| \geq \frac{\Delta}{\epsilon}\sqrt{8k\ln(2/\beta)} ~\right] \leq \beta\\
  \Rightarrow & ~~\textup{Pr}\left[ \left(\sum_i^{k}\tilde{S_i} \geq m + \frac{\Delta}{\epsilon}\sqrt{8k\ln(2/\beta)}\right) \cup \left(\sum_i^{k}\tilde{S_i} \leq m - \frac{\Delta}{\epsilon}\sqrt{8k\ln(2/\beta)}\right) ~\right] \leq \beta\\
\end{split}
\end{equation}

When setting $\alpha \geq \frac{\Delta}{\epsilon}\sqrt{8(T/\eta)\ln(2/\beta)}$,  then the claim holds.
\end{proof}

\begin{theorem}
\label{ag2:perf}
Given an update scheduler stream prefix $U=\{u_t\}_{0\leq t \leq T}$ with length $T$, the event log $E$ with $\Pr[\pi(u_t)=1] = p$, the privacy parameter $\epsilon$, and the threshold is $\theta$. Given $0 \leq \beta \leq 1$, when $\alpha \geq \frac{\Delta \ln(1/\beta)}{\epsilon} $ with probability at least $1-\beta$, the performance of timer method is bounded within the range $(\frac{T}{\eta}(\eta - \alpha), ~\frac{T}{\eta}(\eta + \alpha))$.
\end{theorem}
}

\begin{proof} 
Let $t$ denotes the current time, $c_t$ counts how many records received since last update, and $k$ equals to the total number of synchronizations have been posted until time $t$. We assume a set of timestamps $t' = \{t'_0, t'_1, t'_2, ..., t'_k\}$, where each $t'_i\in t'$ indicates one time unit that $\sync$ signals, and we set $t'_0 = 0$. Let $A =\{a_1, a_2, ..., a_t\}$ as the collection of DP-ANT's outputs, where $a_j \in A$ is either $\perp$ (no sync) or equals to $c_j + \lap(\frac{2}{\epsilon})$, and $\tilde{\theta}_1, \tilde{\theta}_2, ...\tilde{\theta}_k$ to be all previous obtained noisy thresholds until time $t$. Next, we proof this theorem by shwoing the existence of $\alpha > 0$ and $\beta \in (0,1)$, such that with probability at most $\beta$, for all $i \in t'$  it satisfies, $(a_i \neq \perp)\wedge (c_i \leq \theta + \alpha) \wedge (|c_i - a_i| \geq \alpha)$. And for all $i\notin t'$, $(a_i = \perp) \wedge (c_i \geq \theta - \alpha)$. 
In terms of the noise added to the threshold $\theta$, we know that  $\forall_{i = 1,2,...,k} \tilde{\theta}_i \sim \theta + \lap(\frac{4}{\epsilon})$. Then according the {\it Fact 3.7} in~\cite{dwork2014algorithmic}:
\begin{equation}
\begin{split}
& ~~\textup{Pr}\left[ \forall_{i=1,2,...k}|\tilde{\theta}_i - \theta| \geq x \times \frac{4}{\epsilon} \right] = e^{-x}\\
\Rightarrow & ~~\textup{Pr}\left[ \forall_{i=1,2,...k}|\tilde{\theta}_i - \theta| \geq \frac{\alpha}{4} \right] = e^{-\frac{\epsilon\alpha}{16}}\\
\Rightarrow & ~~\textup{Pr}\left[ \sum_{i=1,2,...k}|\tilde{\theta}_i - \theta| \geq \frac{\alpha}{4} \right] = k\times e^{-\frac{\epsilon\alpha}{16}}\\
%\rightarrow & ~~\textup{Pr}\left[ \sum_{i=1}^{k} |\tilde{\theta}_i - \theta| \geq \frac{\alpha}{4} \right] \leq k\times e^{-\frac{\epsilon\alpha}{16}}
\end{split}
\end{equation}
Let $k\times e^{-\frac{\epsilon\alpha}{16}}$ to be at most $\beta/4$, then $\alpha \geq \frac{16(\log{k} + \log{(4/\beta}))}{\epsilon}$.
Similarly, for each time $t$, we know that $\tilde{c}_t - c_t = \lap(\frac{8}{\epsilon})$, thus it satisfies:
\begin{equation}
\begin{split}
 %& ~~\textup{Pr}\left[ \forall_{j=1,2,...t}|\tilde{c}_j - c_j| \geq \frac{\alpha}{2} \right] \leq e^{-\frac{\epsilon\alpha}{16}}\\
 & ~~\textup{Pr}\left[ \forall_{0 < j \leq t} |\tilde{c}_j - c_j| \geq \frac{\alpha}{2} \right] \leq e^{-\frac{\epsilon\alpha}{16}}\\
\Rightarrow & ~~\textup{Pr}\left[ \forall_{j=1,2,...,k}\max_{j=t'_{i-1} + 1}^{t'_{i}} |\tilde{c}_j - c_j| \geq \frac{\alpha}{2} \right] \leq (t'_{i} - t'_{i-1})\times e^{-\frac{\epsilon\alpha}{16}}\\
\Rightarrow & ~~\textup{Pr}\left[ \sum_{j=1,2,...,k}\max_{j=t'_{i-1}+1}^{t'_{i}} |\tilde{c}_j - c_j| \geq \frac{\alpha}{2} \right] \leq (t'_{k} - t'_{0})\times e^{-\frac{\epsilon\alpha}{16}} \leq t\times e^{-\frac{\epsilon\alpha}{16}}
\end{split}
\end{equation}
Let  $t \times e^{-\frac{\epsilon\alpha}{16}}$ to be at most $\beta/2$, we have $\alpha \geq \frac{16(\log{t} + \log{(2/\beta}))}{\epsilon}$. 

Finally, we set the following conditions $\forall_{i\in t'} |a_i - c_i = \lap(\frac{2}{\epsilon})| \geq \alpha$ holds with probability at most $\beta/4$, we obtain $\alpha \geq \frac{4\log{(4/\beta)}}{\epsilon}$. By combining the above analysis, we can obtain if set $\alpha \geq \frac{16(\log{t} + \log{(2/\beta}))}{\epsilon}$ the following holds.
\begin{equation}
\begin{split}
 & ~~\textup{Pr}\left[ LG(t) = \max\left(\sum_{\forall_{j \in t'}} c_j - \sum_{\forall_{j \in t'}} a_j, 0 \right) \geq {\alpha} \right] \leq ~~\textup{Pr}\left[ \sum_{\forall i\in t'} |a_i - c_i|  \geq {\alpha}\right] \\
 \leq & ~~\textup{Pr}\left[ \left(\sum_{i=1}^{k} |\tilde{\theta}_i - \theta|  + \sum_{j=1,2,...,k}\max_{j=t'_{i-1}+1}^{t'_{i}} |\tilde{c}_j - c_j| + \sum_{\forall i\in t'} |a_i - c_i|  \right) \geq {\alpha}\right] \leq \beta. 
\end{split}
\end{equation}
Therefore, for any time $t$, the logical gap under DP-ANT method is greater than $\alpha \geq \frac{16(\log{t} + \log{(2/\beta}))}{\epsilon}$ with probability at most $\beta$.

\eat{

Suppose it requires $\kappa$ times of synchronizations to process update stream $U$, thus we denote $p^i$ as the $p$-cache size at $i$-th synchronization, and $m = \sum_i^{\kappa}p^i$. Let $\hat{p}^{\kappa} = \sum_i^{\kappa} p^i / \kappa$ is the average value of $p^i$ after $\kappa$ times of synchronizations, for $\alpha \geq \frac{9c(\ln k + \ln(4/\beta))}{\epsilon}$ and $\beta \in (0,1)$ then
\begin{equation}
\begin{split}
 & ~~ \textup{Pr}\left[~ \left( \hat{p}^{\kappa} \geq \theta + \alpha \right) \cup \left( \hat{p}^{\kappa} \leq \theta - \alpha \right) ~\right] \leq \beta\\
 \Rightarrow & ~~ \textup{Pr}\left[~ \left( \kappa\hat{p}^{\kappa} \geq \kappa(\theta + \alpha) \right) \cup \left( \kappa \hat{p}^{\kappa} \leq \kappa(\theta - \alpha) \right) ~\right] \leq \beta\\
  \Rightarrow & ~~ \textup{Pr}\left[~ \left( m \geq \kappa(\theta + \alpha) \right) \cup \left( m \leq \kappa(\theta - \alpha) \right) ~\right] \leq \beta\\
\Rightarrow & ~~ \textup{Pr}\left[~ \left( \frac{m}{\theta + \alpha} \geq \kappa \right) \cup \left( \frac{m}{\theta - \alpha}  \leq \kappa \right) ~\right] \leq \beta\\
\end{split}
\end{equation}
Above formulation provides the upper and lower bounds on the total number of synchronization under given condition. Given such bound on $\kappa$, let $a^i$ corresponds to the $|a_t|$ for $i$-th synchronization, and the total records been uploaded is then defined as $\sum_{i=1}^{\kappa} a^i $ and the average value of $a^i$ after $\kappa$ synchronizations is $\hat{a}$, then 
\begin{equation}
\begin{split}
 & ~~ \textup{Pr}\left[~ \left( \kappa\hat{a}^{\kappa} \geq \kappa(\theta + \alpha) \right) ~\right] \leq\frac{\beta}{2}\\
  \Rightarrow & ~\textup{Pr}\left[~ \left( \kappa\hat{a}^{\kappa} \geq \kappa(\theta + \alpha) \right)\cap  \left( \kappa \geq \frac{m}{\theta - \alpha} \right) ~\right] \leq\frac{\beta}{2}\\
\Rightarrow & ~~ \textup{Pr}\left[~ \left( \sum_{i=1}^{\kappa} a^i \geq \frac{m}{\theta - \alpha}(\theta + \alpha) \right) ~\right] \leq \frac{\beta}{2}\\
\Rightarrow & ~~ \textup{Pr}\left[~ \left( \sum_{i=1}^{\kappa} a^i \geq \frac{m}{\theta - \alpha}(\theta + \alpha) \right) \cup  \left( \sum_{i=1}^{\kappa} a^i \leq \frac{m}{\theta + \alpha}(\theta - \alpha) \right) ~\right] \leq \beta\\
\end{split}
\end{equation}
}

\end{proof}

\subsection{Proof of Theorem~\ref{sz:ant}}
\begin{proof}
The size of outsoured data at time $t$ can be written as: 
For each time $t$, let $\eta =  s\floor*{\frac{t}{f}} $ the total number of outsourced data can be written:
\begin{equation}
\begin{split}
  |\mathcal{DS}_t| & = \sum_{\forall_{j \in t'}} \tilde{c}_j - \sum_{\forall_{j \in t'}} c_j + \eta
\end{split}
\end{equation}

Follow the proof of theorem~\ref{lg:ant}:
\begin{equation}
\begin{split}
& \textup{Pr}\left[\sum_{\forall_{j \in t'}} \tilde{c}_j - \sum_{\forall_{j \in t'}} c_j \geq \frac{16(\log{t} + \log{(2/\beta}))}{\epsilon} \right] \leq \beta \\
\rightarrow & \textup{Pr}\left[|\mathcal{DS}_t| \geq |\mathcal{D}_t| + \eta + \frac{16(\log{t} + \log{(2/\beta}))}{\epsilon} \right] \leq \beta
\end{split}
\end{equation}
\end{proof}
%The theorem thus holds.

%\subsubsection{Discussion}

\eat{If $p$ is a very large value close to 1(very dense), then for accuracy the accuracy of ALG1 we have $Err\sim Tk\cdot\frac{\eta-1}{2}$, and if $\theta$ is small enough, $\frac{\eta-1}{2}>>\theta$, then ALG2 performance better for accuracy purpose. While $\frac{T}{\eta-1} << \frac{T}{\theta}$, then AGL2 need more frequent update. Conversely, if we have a very sparse event log stream, then Timer method will provide higher accuracy but cascade buffer will have better performance.

Given an update scheduler stream prefix $U=\{u_t\}_{0\leq t \leq T}$ with length $T$, the event log $E$ with $\Pr[\pi(u_t)=1] = p$, the privacy parameter $\epsilon$, and the expectation of total number of dummy records added each time $n_d$. The updating strategy can be divided in to the following 4 categories according to the number of total updates and the size of each update. The table \ref{tab:cat} shows the classification of all aforementioned algorithms.

\begin{table}[]
\centering
\label{tab:cat}
\caption{Strategy v.s. # of Updates}
\begin{tabular}{l|l|l|}
\cline{2-3}
                                                   & \multicolumn{2}{c|}{{\ul \textbf{\# of Records}}} \\ \hline
\multicolumn{1}{|l|}{{\ul \textbf{Update Scheduler (when)}}} & \textbf{Fixed }  & \textbf{Noisy }  \\ \hline
\multicolumn{1}{|l|}{\textbf{Fixed}}               &    UET(Epoch)                      &   OT, Timer Method                       \\ \hline
\multicolumn{1}{|l|}{\textbf{Noisy}}               &    Randomized Response                      &  Cascading Buffer                        \\ \hline
\end{tabular}
\end{table}

\section{Evaluation}
\subsection{Tradeoff with changing privacy}
\begin{figure}[ht]
\includegraphics[width=0.8\linewidth]{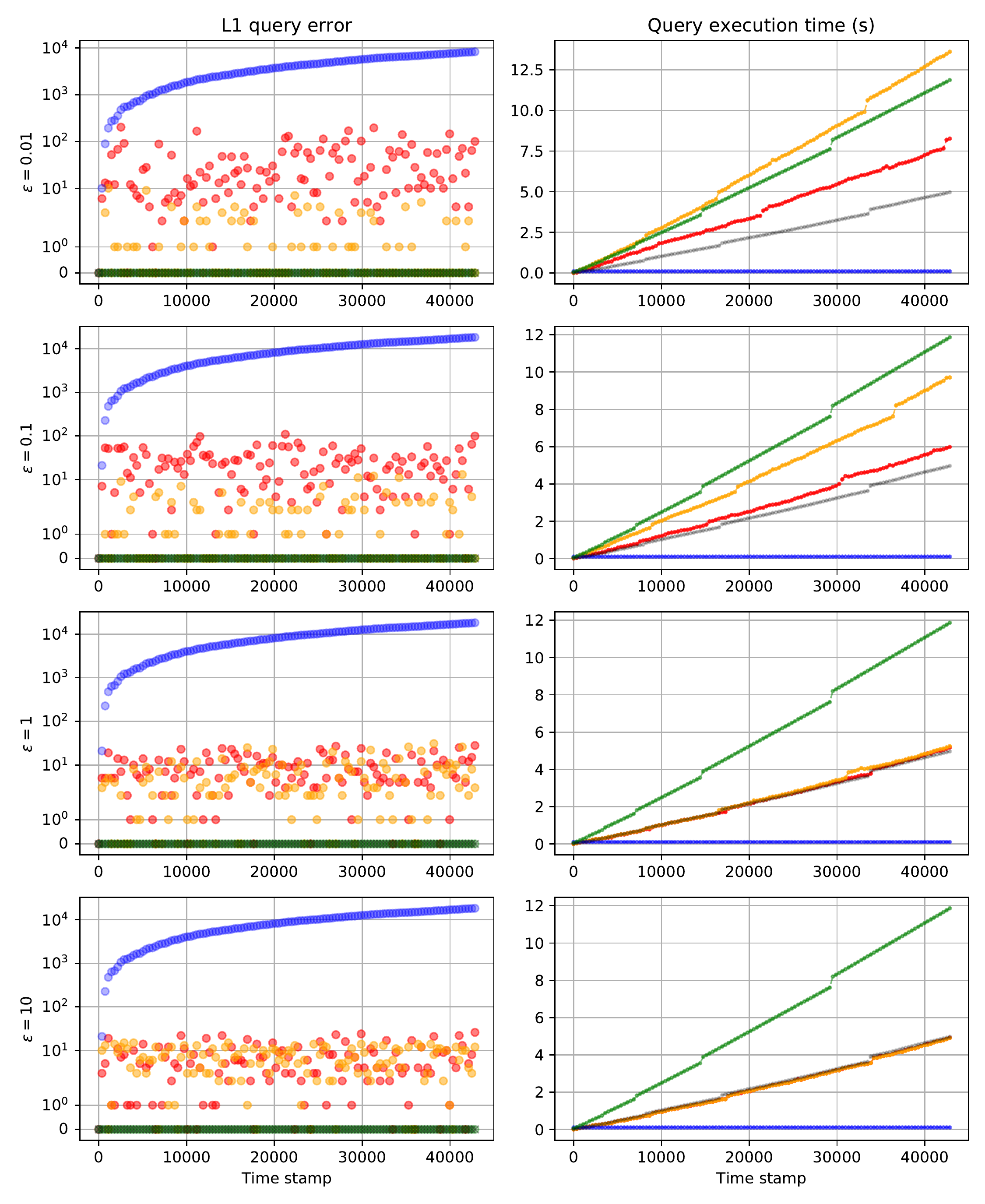}
   \caption{Experiment results with changing $\epsilon$\kartik{0.01, 0.1, 1, 10 on x-axis}}
   \label{trade-off}
\end{figure}

}

\end{document}